\documentclass[11pt]{article}

\usepackage[T1]{fontenc}

\usepackage[margin=1in]{geometry}

\usepackage{wrapfig}

\usepackage{tcolorbox}

\usepackage{authblk}

\usepackage{hyperref}

\usepackage{amsmath}

\usepackage{relsize}

\usepackage{pgfplots}
\usepackage{pgfplotstable}
\usetikzlibrary{matrix,arrows,shapes,positioning,calc,snakes,decorations.markings}
\usetikzlibrary{fadings,shapes.arrows,shadows}   
\usetikzlibrary{arrows.meta,patterns}

\usetikzlibrary{arrows.meta}
\tikzset{>={Latex[width=1mm,length=1mm]}}

\pgfdeclarelayer{background}
\pgfsetlayers{background,main}

\definecolor{lightgray}{gray}{0.7}
\definecolor{medgray}{gray}{0.55}
\definecolor{darkgray}{gray}{0.4}

\definecolor{pastelblue}{rgb}{0.6, 0.8, 0.91}
\definecolor{pastelred}{rgb}{1.0, 0.41, 0.38}

\definecolor{darkpastelgreen}{rgb}{0.01, 0.75, 0.24}
\definecolor{pastelyellow}{rgb}{0.01, 0.75, 0.24}

\definecolor{shade1}{rgb}{0.87, 0.91, 0.92} 
\definecolor{shade2}{rgb}{1, 1, 0.88} 
\definecolor{shade3}{rgb}{0.69, 0.94, 0.92} 
\definecolor{shade4}{rgb}{1, 0.79, 0.78} 

\definecolor{darkpastelblue}{rgb}{0.27, 0.47, 0.70}
\definecolor{darkpastelred}{rgb}{0.59, 0.18, 0.10}
\definecolor{darkpastelbrown}{rgb}{0.39, 0.31, 0.25}
\definecolor{darkpastelmagenta}{rgb}{0.59, 0.44, 0.84}

\tikzset{fontscale/.style = {font=\relsize{#1}}}

\newcommand{\mysquare}[3]{
\node [rectangle, fill=#3,minimum size=4.75mm] () at (#1,#2) {};
}

\newcommand{\mytwocoloursquare}[2]{
	\fill [pastelred]  (#1-0.5, #2-0.5) -- (#1-0.5, #2+0.5) -- (#1+0.5,#2+0.5) -- cycle;
	\fill [pastelblue] (#1-0.5, #2-0.5) -- (#1+0.5, #2-0.5) -- (#1+0.5,#2+0.5) -- cycle;
	\draw[-] (#1-0.5,#2-0.5) -- (#1+0.5,#2+0.5);
	\node [rectangle, draw=black, minimum size=4.75mm] () at (#1,#2) {};
}

\newcommand{\myuparrow}[3]{
	\mysquare{#1}{#2}{lightgray}
	\draw[->, white] (#1,#2-0.25) -- ($(#1,#2-0.25)!0.5cm!(#1,#2+1)$);
}
\newcommand{\mydnarrow}[3]{
	\draw[->, darkgray] (#1,#2+0.25) -- ($(#1,#2+0.25)!0.5cm!(#1,#2-1)$);
}

\newcommand{\myrightarrow}[3]{
	\mysquare{#1}{#2}{lightgray}
	\draw[->, white] (#1-0.25,#2) -- ($(#1-0.25,#2)!0.5cm!(#1+1,#2)$);
}
\newcommand{\myleftarrow}[3]{
	\draw[->, darkgray] (#1+0.25,#2) -- ($(#1+0.25,#2)!0.5cm!(#1-1,#2)$);
}

\newcommand{\zeroS}[3]{
	\node[white, fill=#3, minimum size=4.75mm] at (#1,#2) {$0$};
}

\newcommand{\updncol}[2]{ 
	\pgfmathsetmacro{\limbelow}{#2 - 1}
	\pgfmathsetmacro{\limabove}{#2 + 1}

	\foreach \j in {1,...,\limbelow}
		\myuparrow{#1}{\j}{black};
	\zeroS{#1}{#2}{pastelblue};
	\foreach \j in {\limabove,...,8}
		\mydnarrow{#1}{\j}{black};
}

\newcommand{\leftrightcol}[2]{ 
	\pgfmathsetmacro{\limleft}{#1 - 1}
	\pgfmathsetmacro{\limright}{#1 + 1}

	\foreach \i in {1,...,\limleft}
		\myrightarrow{\i}{#2}{black};
	\zeroS{#1}{#2}{pastelred};
	\foreach \i in {\limright,...,10}
		\myleftarrow{\i}{#2}{black};
}

\newcommand{\patharrow}[4]{
\draw[->, thick] (#1,#2) -- (#3,#4);
}

\usepackage[utf8]{inputenc}
\usepackage{enumitem}
\setlist[itemize]{leftmargin=*}

\usepackage{amsthm}
\usepackage{amsmath}

\newtheorem{conjecture}{Conjecture}
\newtheorem{theorem}{Theorem}
\newtheorem{lemma}[theorem]{Lemma}
\newtheorem{definition}[theorem]{Definition}
\newtheorem{observation}[theorem]{Observation}
\newtheorem{corollary}[theorem]{Corollary}

\usepackage{cite}

\usepackage[disable]{todonotes}

\usepackage{microtype}

\title{\huge End of Potential Line}

\author[1]{John Fearnley}
\author[2]{Spencer Gordon}
\author[3]{Ruta Mehta}
\author[1]{Rahul Savani}

\affil[1]{University of Liverpool\\
  \texttt{\{john.fearnley, rahul.savani\}@liverpool.ac.uk}}
\affil[2]{California Institute of Technology\\
\texttt{slgordon@caltech.edu}}
\affil[3]{University of Illinois at Urbana-Champaign\\
  \texttt{rutamehta@cs.illinois.edu}}

\usepackage{tabularx,amssymb,nicefrac,amsmath,multirow,stmaryrd}
\usepackage{comment,amsfonts} 
\usepackage{epsfig} \usepackage{latexsym,nicefrac,bbm}
\usepackage{xspace}
\usepackage{color,fancybox,graphicx,url}
\usepackage{tabularx} 
\usepackage{booktabs}
\usepackage{color,colortbl}
\usepackage{mathtools}

\newcommand{\tf}{{\tilde{f}}}
\newcommand{\mz}{{\times\zeta}}
\newcommand{\size}{{\mbox{size}}}

\renewcommand{\ss}{{\mathbf s}}
\renewcommand{\tt}{{\mathbf t}}
\newcommand{\yy}{{\mathbf y}}
\newcommand{\hh}{{\mathbf h}}
\newcommand{\ww}{{\mathbf w}}
\newcommand{\bb}{{\mathbf b}}

\usepackage{algorithmicx}
\usepackage{algorithm}
\usepackage{algpseudocode}

\newcommand{\ra}{\rightarrow}
\long\def\symbolfootnote[#1]#2{\begingroup%
\def\thefootnote{\fnsymbol{footnote}}\footnote[#1]{#2}\endgroup}

\DeclareMathOperator{\Slice}{Slice}
\DeclareMathOperator{\DSlice}{DSlice}

\def\cc#1{\mathsf{#1}}
\def\CLS{\ensuremath{\cc{CLS}}\xspace}
\def\P{\ensuremath{\cc{P}}\xspace}
\def\NP{\ensuremath{\cc{NP}}\xspace}
\def\coNP{\ensuremath{\cc{coNP}}\xspace}
\def\NPcoNP{\ensuremath{\cc{NP} \cap \cc{coNP}}\xspace}
\def\TFNP{\ensuremath{\cc{TFNP}}\xspace}
\def\PPAD{\ensuremath{\cc{PPAD}}\xspace}
\def\PLS{\ensuremath{\cc{PLS}}\xspace}
\def\PPADPLS{\ensuremath{\cc{PPAD} \cap \cc{PLS}}\xspace}
\def\EOPLc{\ensuremath{\cc{EOPL}}\xspace}
\def\UniqueEOMLc{\ensuremath{\cc{UniqueEOML}}\xspace}
\def\UniqueEOPLc{\ensuremath{\cc{UniqueEOPL}}\xspace}

\newcommand{\LinearFIXP}{\ensuremath{\cc{LinearFIXP}}\xspace}

\def\problem#1{{\scshape #1}}
\def\CM{\problem{Contraction}\xspace}
\def\LCM{\problem{PL-Contraction}\xspace}
\def\DCM{\problem{Discrete-Contraction}\xspace}
\def\GCM{\problem{GeneralContraction}\xspace}
\def\MMCM{\problem{MetametricContraction}\xspace}

\def\EOL{\problem{EndOfLine}\xspace}
\def\EOPL{\problem{EndOfPotentialLine}\xspace}
\def\UFEOPL{\problem{UniqueForwardEOPL}\xspace}
\def\EOML{\problem{EndOfMeteredLine}\xspace}
\def\SOVL{\problem{SinkOfVerifiableLine}\xspace}
\def\CLO{\problem{ContinuousLocalOpt}\xspace}

\def\PromisePLCP{\problem{Promise-P-LCP}\xspace}
\def\PLCP{\problem{P-LCP}\xspace}
\def\ContractionMap{\problem{ContractionMap}\xspace}
\def\MBanach{\problem{MetricBanach}\xspace}
\def\e{\epsilon}
\def\eps{\varepsilon}

\def\ite{\mbox{ItoE}}
\def\eti{\mbox{EtoI}}
\def\pot{\mbox{$V$}}
\def\isvalid{\mbox{IsValid}}
\def\PLo{\mbox{Q1}\xspace}
\def\PLt{\mbox{Q2}\xspace}

\def\Real{\mathbb{R}}

\def\Integer{\mathbb{Z}}
\def\Natural{\mathbb{N}}

\def\Rational{\mathbb{Q}}

\let\N\Natural

\let\R\Real

\def\poly{\operatorname{poly}}

\def\Set#1{\left\{ #1 \right\}}
\def\Abs#1{\left| #1 \right|}
\def\Card#1{\left| #1 \right|}
\def\Norm#1{\left\| #1 \right\|}
\def\Paren#1{\left( #1 \right)}		
\def\Brack#1{\left[ #1 \right]}		

\makeatletter
\def\Bigbar#1{\mathrel{\left|\vphantom{#1}\right.\n@space}}
\def\Setbar#1#2{\Set{#1 \Bigbar{#1 #2} #2}}

\def\vert{\operatorname{\mathsf{vert}}}

\def\begin@lgo{\begin{minipage}{1in}\begin{tabbing}
		\quad\=\qquad\=\qquad\=\qquad\=\qquad\=\qquad\=\qquad\=\qquad\=\qquad\=\qquad\=\qquad\=\qquad\=\qquad\=\kill}
\def\end@lgo{\end{tabbing}\end{minipage}}

\makeatother

\newcommand{\CPol}{\mbox{${\mathcal P}$}}
\newcommand{\CI}{\mbox{${\mathcal I}$}}
\newcommand{\CL}{\mbox{${\mathcal L}$}}
\newcommand{\CE}{\mbox{${\mathcal E}$}}
\newcommand{\CQ}{\mbox{${\mathcal Q}$}}

\newcommand{\uu}{\mbox{\boldmath $u$}}
\newcommand{\vv}{\mbox{\boldmath $v$}}
\newcommand{\qq}{\mbox{\boldmath $q$}}
\newcommand{\xx}{\mbox{\boldmath $x$}}
\newcommand{\one}{\mbox{\boldmath $1$}}
\newcommand{\ones}{\mbox{\boldmath $1$}}
\newcommand{\zeros}{\mbox{\boldmath $0$}}
\newcommand{\MM}{\mbox{$M$}}
\newcommand{\ps}{\mbox{\boldmath $s$}}
\newcommand{\pq}{\mbox{\boldmath $q$}}

\newcommand{\blank}{\ensuremath{\mathtt{*}}}
\newcommand{\up}{\ensuremath{\mathsf{up}}}
\newcommand{\down}{\ensuremath{\mathsf{down}}}
\newcommand{\zero}{\ensuremath{\mathsf{zero}}}
\newcommand{\vblank}{\ensuremath{\mathtt{-}}}
\DeclareMathOperator{\point}{Point}

\DeclareMathOperator{\fixed}{fixed}
\DeclareMathOperator{\free}{free}

\DeclareMathOperator{\FindFP}{\mbox{\textsc{FindFP}}}
\DeclareMathOperator{\ApproxFindFP}{\mbox{\textsc{ApproxFindFP}}}

\newcommand{\restr}[2]{#1_{\left|#2\right.}}
\newcommand{\Restr}[2]{\tilde{#1}_{\left|#2\right.}}

\begin{document}

\maketitle
\thispagestyle{empty}

\begin{abstract}
We introduce the problem \EOPL and the corresponding complexity class \EOPLc of
all problems that can be reduced to it in polynomial time. This class captures
problems that admit a single combinatorial proof of their joint membership in
the complexity classes \PPAD of fixpoint problems and \PLS of local search
problems. \EOPLc is a combinatorially-defined alternative to the class \CLS (for
Continuous Local Search), which was introduced in
\cite{daskalakis2011continuous} with the goal of capturing the complexity of
some well-known problems in $\PPAD \cap \PLS$ that have resisted, in some cases
for decades, attempts to put them in polynomial time. Two of these are \CM, the
problem of finding a fixpoint of a contraction map, and \PLCP,
the problem of solving a P-matrix Linear Complementarity Problem.

We show that \EOPL is in \CLS via a two-way reduction to \EOML. The latter was
defined in~\cite{hubavcek2017hardness} to show query and cryptographic 
lower bounds for \CLS.
Our two main results are to show that both \LCM (Piecewise-Linear \CM) 
and \PLCP are in \EOPLc.
Our reductions imply that the promise versions of
\LCM and \PLCP are in the promise class \UniqueEOPLc, which corresponds to
the case of a single potential line. 
This also shows that simple-stochastic, discounted, mean-payoff, and parity
games are in \EOPLc.

Using the insights from our reduction for \LCM, we obtain the first
polynomial-time algorithms for finding fixed points of 
contraction maps in fixed dimension for any $\ell_p$
norm, where previously such algorithms were only known for the $\ell_2$ and $\ell_\infty$ norms.
Our reduction from \PLCP to \EOPL allows a technique of Aldous to be applied, which
in turn gives the fastest-known randomized algorithm for~\PLCP.
	
\end{abstract}

\newpage
\thispagestyle{empty}


\tableofcontents

\newpage

\clearpage
\pagenumbering{arabic} 

\section{Introduction}

\textbf{Total function problems in NP.} The complexity class \TFNP
contains search problems that are guaranteed to have a solution, and whose
solutions can be verified in polynomial time~\cite{megiddo1991total}.
While it is a semantically defined complexity class and thus unlikely to
contain complete problems, a number of syntactically defined subclasses of
\TFNP have proven very successful at capturing the complexity of total search 
problems.
In this paper, we focus on two in particular, \PPAD and \PLS.
The class \PPAD was introduced in
\cite{papadimitriou1994complexity} to capture the difficulty of problems
that are guaranteed total by a parity argument. It has attracted intense
attention in the past decade, culminating in a series of papers showing that the
problem of computing a Nash-equilibrium in two-player games is \PPAD-complete
\cite{chen2009settling,daskalakis2009complexity}, and more recently a conditional
lower bound that rules out a PTAS for the problem~\cite{Rubinstein16}.
No polynomial-time algorithms for \PPAD-complete problems are known, and recent work
suggests that no such algorithms are likely to exist~\cite{BPR15,garg2016revisiting}. 
\PLS is the class of problems that
can be solved by local search algorithms (in perhaps exponentially-many steps).
It has also attracted much interest since it was introduced in
\cite{johnson1988easy}, and looks similarly unlikely to have polynomial-time
algorithms. Examples of problems that are complete for \PLS include the problem
of computing a pure Nash equilibrium in a congestion
game~\cite{fabrikant2004complexity}, a locally optimal max
cut~\cite{schaffer1991simple}, or a stable outcome in a hedonic game~\cite{GairingS10}.

\medskip

\noindent \textbf{Continuous Local Search.} 
If a problem lies in both \PPAD and \PLS then it is unlikely to be complete for 
either class, since this would imply an extremely surprising containment of one class in the other.
In their 2011 paper~\cite{daskalakis2011continuous}, Daskalakis and
Papadimitriou observed that there are several prominent total function problems
in \PPADPLS for which researchers have not been able to design polynomial-time
algorithms. Motivated by this they introduced
the class $\CLS$, a syntactically defined subclass of $\cc{PPAD} \cap \cc{PLS}$.
\CLS is intended to capture the class of optimization problems over a continuous
domain in which a continuous potential function is being minimized and the
optimization algorithm has access to a polynomial-time continuous improvement
function.
They showed that many classical problems of unknown
complexity are in $\CLS$, including the problem of solving a simple
stochastic game, the more general problems of solving a Linear Complementarity
Problem with a P-matrix, and finding an approximate fixpoint
to a contraction map. Moreover, $\CLS$ is the smallest known subclass of
$\cc{TFNP}$ not known to be in \P, and hardness results for it imply hardness
results for $\cc{PPAD}$ and $\cc{PLS}$ simultaneously. 

Recent work by Hub\'{a}\v{c}ek and Yogev~\cite{hubavcek2017hardness} proved 
lower bounds for \CLS. They introduced
a problem known as \EOML which they showed was in \CLS, and for which they
proved a query complexity lower bound of $\Omega(2^{n/2}/\sqrt{n})$ and
hardness under the assumption that there were one-way permutations and
indistinguishability obfuscators for problems in $\cc{P_{/poly}}$.
Another recent result showed that the search version of the Colorful
Carath\'eodory Theorem is in $\cc{PPAD} \cap \cc{PLS}$, and left open
whether the problem is also in $\CLS$~\cite{colorfulcara2017}.

Until recently, it was not known whether there was a natural \CLS-complete problem 
.
In their original paper, Daskalakis and Papadimitriou suggested two natural
candidates for \CLS-complete problems, \ContractionMap and \PLCP, which we 
study in this paper.
Recently, two variants of \ContractionMap have been shown to be \CLS-complete.
Whereas in the original definition of \ContractionMap it is assumed that 
an $\ell_p$ or $\ell_\infty$ norm is fixed, and the contraction property 
is measured with respect to the metric induced by this fixed norm, 
in these two new complete variants, a metric~\cite{DTZ17} and
meta-metric~\cite{FGMS17} 
are given as input to the problem\footnote{
The result for meta-metrics appeared in an unpublished earlier version of 
this paper~\cite{FGMS17}. We include it Section~\ref{sec:MMCMisCLScomplete} with a 
comparison to~\cite{DTZ17}. This current version
of our paper has new results and a corresponding different title.}. 

\medskip

\noindent
\textbf{Our contribution.}
Our main conceptual contribution is to define the problem \EOPL, which unifies
in an extremely natural way the circuit-based views of \PPAD and of \PLS.
The canonical
\PPAD-complete problem is \EOL, a problem that provides us with an
exponentially large graph consisting of lines and cycles, and asks us to find
the end of one of the lines. The canonical \PLS-complete problem provides us
with an exponentially large DAG, whose acyclicity is guaranteed by the
existence of a \emph{potential function} that increases along each edge.
The problem \EOPL is an instance of \EOL that \emph{also} has a potential
function that increases along each edge.

We define the corresponding complexity class \EOPLc, consisting of all problems
that can be reduced to \EOPL. This class captures problems that admit a
\emph{single} combinatorial proof of their joint membership in the classes \PPAD
of fixpoint problems and \PLS of local search problems, and is a
combinatorially-defined alternative to the class~\CLS. 


The problem \EOML was introduced in~\cite{hubavcek2017hardness} to show query
and cryptographic lower bounds for \CLS. We show that \EOPL is contained in \CLS via a
two-way reduction to \EOML, and in doing so,
we show that \EOPLc is the closest class to \P among the syntactic sub-classes of \TFNP.
\EOML captures problems that have a \PPAD directed
path structure where it is possible to keep count of \emph{exactly} how far each
vertex is from the start of the path. In a sense, this may seem rather
unnatural, as many common problems do not seem to have this property. However,
our reduction from \EOPL to \EOML shows that the latter is more general than it 
may at first seem.
The reduction from \EOML to \EOPLc implies that the query 
and cryptographic lower bounds for \CLS also hold for~\EOPLc.

For our two main technical results, we show that two problems with
long-standing unresolved complexity belong to \EOPLc. These problems are:

\begin{itemize}
\item The P-matrix Linear Complementarity Problem (\PLCP).
Designing a polynomial-time algorithm for this problem has been open for
decades, at least since the 1978 paper of Murty~\cite{murty1978computational}
that provided exponential-time examples for \emph{Lemke's algorithm}~\cite{lemke1965bimatrix} 
for P-matrix LCPs. 

\item Finding a fixpoint of a Piecewise-Linear Contraction Map (\LCM).
The problem of finding a fixpoint of a \LinearFIXP circuit
with purported Lipschitz constant $c$ is complete for \PPAD. 
In \LCM we are asked to solve this 
problem when $c$ is strictly less than~1.
A contraction map on a 
continuous bounded domain has a \emph{unique} fixpoint and captures 
many important problems including discounted games and simple stochastic games.
\end{itemize}
We summarize our complexity-theoretic results in Figure~\ref{fig:reductions}.
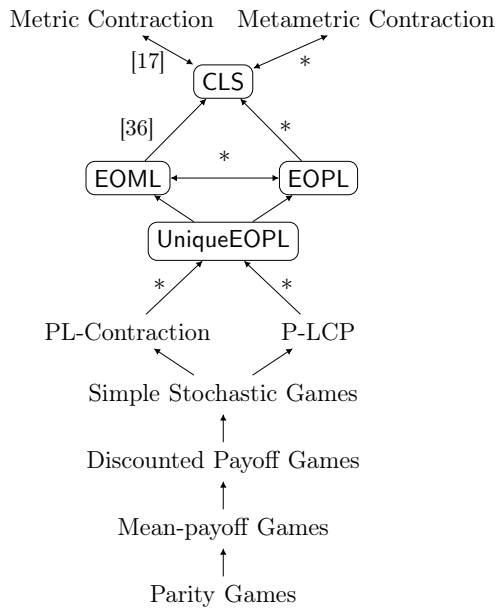
\begin{wrapfigure}{R}{0.4\textwidth}
\scalebox{0.85}{ 
\begin{tikzpicture}[node distance=1.4cm]

\tikzstyle{cc}=[rectangle, draw, rounded corners]

\node (ssg) {Simple Stochastic Games};
\node[below of=ssg,yshift=3.5mm] (dpg) {Discounted Payoff Games};
\node[below of=dpg,yshift=3.5mm] (mpg) {Mean-payoff Games};
\node[below of=mpg,yshift=3.5mm] (parity) {Parity Games};
\node[above left of=ssg, xshift=-0.5cm] (contraction) {PL-Contraction};
\node[above right of=ssg, xshift=0.5cm] (plcp) {P-LCP};
\node[cc, above of=ssg, yshift=1cm] (ueopl) {$\mathsf{UniqueEOPL}$};
\node[cc, above left of=ueopl, xshift=-0.5cm] (eoml) {$\mathsf{EOML}$};
\node[cc, above right of=ueopl, xshift=0.5cm] (eopl) {$\mathsf{EOPL}$};
\node[cc, above of=ssg, yshift=3.5cm] (cls) {$\mathsf{CLS}$};

\node[above left of=cls, xshift=-0.75cm] (mcontraction) {Metric Contraction};
\node[above right of=cls, xshift=1.25cm] (mmcontraction) {Metametric Contraction};

\path[->] (parity) edge (mpg);
\path[->] (mpg) edge (dpg);
\path[->] (dpg) edge (ssg);

\path[->] (ssg) edge (contraction);
\path[->] (ssg) edge (plcp);

\path[->] (contraction) edge node[left] {*} (ueopl);
\path[->] (plcp) edge node[right] {*} (ueopl);


\path[<->] (eoml) edge node[above] {*} (eopl);

\path[->] (eoml) edge node[left, xshift=-2mm] {{\small \cite{hubavcek2017hardness}}} (cls);
\path[->] (eopl) edge node[right] {*} (cls);
 
\path[->] (ueopl) edge (eoml);
\path[->] (ueopl) edge (eopl);

\path[<->] (mcontraction) edge node[left, xshift=1mm, yshift=-2mm] {{\small \cite{DTZ17}}} (cls);
\path[<->] (mmcontraction) edge node[right,yshift=-2mm] {*} (cls);

\end{tikzpicture}
}
\caption{The relationship of \EOPLc to other classes and problems. 
Results from this paper are shown as stars.}
\label{fig:reductions}
\end{wrapfigure}
Our reductions imply that the promise versions of \LCM and \PLCP are in the
promise class \UniqueEOPLc, which corresponds to instances of \EOPL that have a
single line.  
Other well-known problems also lie in \UniqueEOPLc:
Our two reductions that put \PLCP and \LCM in \EOPLc both independently also
show that simple stochastic, discounted, mean-payoff, and  
parity games are all in \UniqueEOPLc (unique because these 
games are \emph{guaranteed} to produce contraction maps and P-matrix LCPs).
Recently several papers have studied the problem ARRIVAL, which is a
reachability problem for a switching system.
ARRIVAL was introduced in \cite{DGKMW17} and shown to be in \NPcoNP,
then in \PLS~\cite{Kar17}, and then most recently in \CLS~\cite{GHHKMS18}.
In fact, the latter proof also shows that ARRIVAL belongs to \UniqueEOPLc.
In addition, we reduce the non-promise version of \PLCP, the problem of solving the LCP or returning a
violation of the P-matrix property, to (non-unique) \EOPLc.

Our final contributions are algorithmic and arise from the structural
insights provided by our two main reductions.
Using the ideas from our reduction for \LCM, we obtain the first
polynomial-time algorithms for finding fixpoints of 
contraction maps in fixed dimension for any $\ell_p$
norm, where previously such algorithms were only known for $\ell_2$ and $\ell_\infty$ norms.
We also show that these results can be extended to the case where the contraction map is given by a general arithmetic circuit.
As noted in~\cite{GHHKMS18},
our reduction from \PLCP to \EOPL allows a technique of Aldous to be applied, which
in turn gives the fastest known randomized algorithm for~\PLCP.

\medskip
\noindent
\textbf{Related work.}
Papadimitriou showed that \PLCP, the problem of solving the LCP or returning a
violation of the P-matrix property, is in
\PPAD~\cite{papadimitriou1994complexity} using Lemke's algorithm.
The relationship between Lemke's algorithm and \PPAD has been studied
by Adler and Verma~\cite{AV11}.
Later, Daskalakis and Papadimitrou showed that \PLCP is in
\CLS~\cite{daskalakis2011continuous}, using the potential reduction method
in~\cite{kojima1992interior}.  
Many algorithms for \PLCP have been studied, e.g.,~\cite{murty1978computational,morris2002randomized,kojima1991unified}.
However,
no polynomial-time algorithms are known for \PLCP, or for the promise
version where one can assume that the input matrix is a P-matrix.
This promise version is motivated by the reduction from Simple Stochastic Games 
to P-matrix LCPs. 

The best known algorithms for \PLCP are based on
a reduction to Unique Sink Orientations (USOs) of cubes~\cite{StickneyW78}.
For an P-matrix LCP of size $n$, the USO algorithms of~\cite{SzaboW01} apply, and 
give a deterministic 
algorithm that runs in time $O(1.61^n)$ and a randomized algorithm with
expected running time $O(1.43^n)$.
The application of Aldous' algorithm to the \EOPL instance that we produce from
a P-matrix LCP takes expected time $2^{n/2}\cdot \poly(n) = O(1.4143^n)$ in the worst case.

Simple Stochastic Games are related to Parity games, which are an extensively studied class of two-player zero-sum
infinite games that capture important problems in formal verification and logic~\cite{EmersonJ91}.
There is a sequence of polynomial-time reductions from parity games 
to mean-payoff games to discounted games to simple stochastic 
games~\cite{puri1996theory,gartner2005simple,jurdzinski2008simple,zwick1996complexity,hansen2013complexity}.
The complexity of solving these problems is unresolved and has received 
much attention over many years~(see, for example, 
\cite{zwick1996complexity,condon1992complexity,fearnley2010linear,jurdzinski1998deciding,bjorklund2004combinatorial,fearnley2016complexity}).
In a recent breakthrough~\cite{parity}, a quasi-polynomial time algorithm for parity games
have been devised, and there are now several algorithms with this running time~\cite{parity,JL17,FJS0W17}.
For mean-payoff, discounted, and simple stochastic games, the best-known 
algorithms run in randomized subexponential
time~\cite{ludwig1995subexponential}. The existence of a polynomial time
algorithm for solving any of these games would be a major breakthrough. Simple
stochastic games can also be reduced in polynomial time to \LCM with the
$\ell_\infty$ norm~\cite{EtessamiY10}.

The problem of computing a fixpoint of a continuous map $f:\mathcal{D}\mapsto \mathcal{D}$
with Lipschitz constant
$c$ has been extensively studied, in both continuous and discrete variants~\cite{ChenD09,ChenD08,DengQSZ11}.
For arbitrary maps with $c>1$, exponential bounds on the query
complexity of computing fixpoints are known~\cite{HirschPV89,ChenD05}.
In \cite{BoonyasiriwatSX07,HuangKS99,Sik09}, algorithms for computing fixpoints 
for specialized maps such as weakly ($c=1$) or strictly ($c < 1$) contracting maps are studied.
For both cases, algorithms are known for the case of $\ell_2$ and $\ell_\infty$ norms, 
both for absolute approximation ($||x-x^*||\le \epsilon$ where $x^*$ is
an exact fixpoint) and relative approximation ($||x-f(x)||\le \epsilon$). 
A number of algorithms are known for
the $\ell_2$ norm handling both types of approximation \cite{NemYud83,HuangKhachSik99,Sik01}. 
There is an exponential lower bound for absolute approximation with $c=1$ \cite{Sik01}.
For relative approximation and a domain of dimension $d$, an
$O(d\cdot\log{1/\epsilon})$ time algorithm is known \cite{HuangKhachSik99}. 
For absolute approximation with $c<1$, an ellipsoid-based algorithm with time complexity
$O(d \cdot [\log(1/\epsilon) + \log(1/(1-c))])$ is known \cite{HuangKhachSik99}. 
For the $\ell_\infty$ norm, \cite{ShellSik03} gave an algorithm to find an
$\epsilon$-relative approximation in time $O(\log(1/\epsilon)^d)$ which is
polynomial for constant $d$. In summary,
for the 
$\ell_2$ norm  polynomial-time algorithms are known for strictly contracting maps; for the 
$\ell_\infty$ norm algorithms that are polynomial time for constant dimension are known.
For arbitrary $\ell_p$ norms, to the best of our knowledge, no polynomial-time
algorithms for constant dimension were known before this paper.


\section{\EOPL}
\label{sec:EOPL}

In this section, we define a new problem \EOPL, and 
 we show that this problem is polynomial-time equivalent to \EOML.
%
First we recall the definition of \EOML, which was first defined
in~\cite{hubavcek2017hardness}. 
It is close in spirit to the problem \EOL that
is used to define \PPAD~\cite{papadimitriou1994complexity}. 

\begin{definition}[\EOML~\cite{hubavcek2017hardness}]
Given circuits $S,P: \{0,1\}^n \rightarrow \{0,1\}^n$, and $V:\{0,1\}^n\rightarrow \{0,\dots, 2^n\}$ such that $P(0^n) =0^n\neq S(0^n)$ and $V(0^n)=1$, find a string $\xx \in \{0,1\}^n$ satisfying one of the following 
\begin{enumerate}[label=(T\arabic*)]
\itemsep0mm
\item either $S(P(\xx))\neq \xx \neq 0^n$ or $P(S(\xx))\neq \xx$,
\item $\xx\neq 0^n, V(\xx)=1$,
\item either $V(\xx)>0$ and $V(S(\xx))-V(\xx)\neq 1$, or $V(\xx)>1$ and $V(\xx)-V(P(\xx))\neq 1$. 
\end{enumerate}
\end{definition}
Intuitively, an \EOML is an \EOL instance that is also equipped with an
``odometer'' function. The circuits $P$ and $S$ implicitly define an
exponentially large graph in which each vertex has degree at most 2, just as in \EOL, and condition T1 says that the end of
every line (other than $0^n$) is a solution.
In particular, the string 
$0^n$ is guaranteed to be the end of a line, and so a solution can be found by
following the line that starts at $0^n$.
 The function $V$ is intended to help with this, by giving the number of steps
that a given string is from the start of the line. We have that $V(0^n) = 1$,
and that $V$ increases by exactly 1 for each step we make along the line.
Conditions T2 and T3 enforce this by saying that any violation of the property
is also a solution to the problem.

In \EOML, the requirement of incrementing $V$ by exactly one as we walk along the
line is quite restrictive. We define a new problem, \EOPL,  which is similar in
spirit to \EOML.
The key difference is that while $V$ is still required to be strictly
monotonically increasing along the line, the amount that it
increases in each step can be any amount.

\begin{definition}[\EOPL]
\label{def:EOPL}
Given Boolean circuits $S,P : \Set{0,1}^n \to \Set{0,1}^n$ such that $P(0^n) =0^n\neq S(0^n)$ and a Boolean circuit $V: \Set{0,1}^n \to \Set{0,1,\dotsc,2^m - 1}$ such that $V(0^n) = 0$ find one of the following:
\begin{enumerate}[label=(R\arabic*)]
\itemsep0mm
\item A point $x \in \Set{0,1}^n$ such that $S(P(x)) \neq x \neq 0^n$ or $P(S(x)) \neq x$.
\item A point $x \in \Set{0,1}^n$ such that $x \neq S(x)$, $P(S(x)) = x$, and $V(S(x)) - V(x) \leq 0$.
\end{enumerate}
\end{definition}


We define the complexity class \EOPLc as all problems that can be reduced
to \EOPL in polynomial time.
We also define a natural variant of \EOPLc, called \UniqueEOPLc which has a
\emph{single} potential line. Since there is no efficient way to verify this,
\UniqueEOPLc is a promise class.
We will show that the promise versions of \LCM and \PLCP are in \UniqueEOPLc.

\medskip

\noindent \textbf{Aldous' algorithm for problems in \EOPLc.}
In~\cite{GHHKMS18}, the problem ARRIVAL of deciding reachability for a directed
switching system was reduced to \EOML, and it was noted that a randomized 
algorithm by Aldous~\cite{Aldous83} provides the best-known algorithm for ARRIVAL.
Aldous' algorithm is actually very simple: it randomly samples a large number of
candidate solutions and then performs a local search from the best sampled
solution. Aldous' algorithm can actually solve any problem in \EOPLc.
%

\smallskip

\textbf{\EOML to \EOPL and back.}
As one might expect, the reduction to \EOPL is relatively easy, but we need to
be careful about handling the solutions that correspond to violations.
%
%
Full details of the reduction with
proofs are in Appendix~\ref{sec:EOMLtoEOPL}.

The reduction to \EOML is involved.
The basic idea is that, when the potential value jumps by more than~$1$, we
introduce a sequence of new vertices. 
We embed the potential into the vertex space, so that
if there is an edge from $\uu$ to $\uu'$ in the \EOPL
instance whose respective potentials are $p$ and $p'$ such that $p<p'$ then
we create edges $(u,p)\ra (u,p+1)\ra \dots \ra (u,p'-1)\ra (u,p')$. 
By increasing the vertex space in this way, we create many vertices that are
not on any line. To ensure that these vertices are not solutions we turn them
into self-loops.
The remaining details, including how we deal with solutions that correspond to violations,
are given in Appendix~\ref{sec:eopl2eoml}.
%
We obtain the following theorem.
\begin{theorem}
\label{thm:eoml2eopl}
\EOML and \EOPL are polynomial-time equivalent.
\end{theorem}
We note that our reduction implies that $\UniqueEOMLc = \UniqueEOPLc$.

\section{The P-matrix Linear Complementarity Problem}
\label{sec:PLCPtoEOPL}

In this section, we reduce \PLCP to \EOPL. Our reduction relies heavily 
on the application of Lemke's algorithm to P-matrix LCPs.
Let $[n]$ denote the set $\{1,\dots,n\}$.

\begin{definition}[LCP $(M, \qq)$]
\label{def:lcp}
Given a matrix $M \in \Real^{d \times d}$ and vector $\qq\in \Real^{d\times 1}$,
find a~{$\yy\in \Real^{d \times 1}$} s.t.:
\begin{equation}\label{eq:lcp}
M\yy\le \qq;\ \ \ \ \yy\ge 0;\ \ \ \ y_i(\qq - M\yy)_i =0,\ \forall i \in [n].
\end{equation}
\end{definition}
In general, deciding whether an LCP has a solution is \NP-complete~\cite{chung1989np},
but if $M$ is a P-matrix, as defined next, then the LCP $(M,\qq)$ has a \emph{unique} solution 
for all $\qq\in \Real^{d\times 1}$.
\begin{definition}[P-matrix]
\label{def:Pmatrix}
$M \in \Real^{d \times d}$ is called a P-matrix if every principle
minor of $M$ is positive.
\end{definition}
Checking if a matrix is a P-matrix is \coNP-complete~\cite{coxson1994p}.
Thus, Megiddo~\cite{megiddo1988note, megiddo1991total} defined the problem \PLCP, 
which avoids the need for a promise about $M$.
\begin{definition}[\PLCP and \PromisePLCP] \label{def:plcp} Given an LCP $(M, \qq)$, find either:
(\PLo) $\yy \in \Real^{n\times 1}$ that satisfies (\ref{eq:lcp}); 
or (\PLt) a non-positive principal minor of $M$.
For \PromisePLCP, we are promised that $M$ is a P-matrix and thus only seek a solution of type \PLo.
\end{definition}


\begin{figure}[tbp]
   \centering
\begin{center}
\scalebox{0.99}{ 
\begin{tikzpicture}

	\path[use as bounding box] (-3,-3) rectangle (3.7,3.9);
	\clip (-3,-3) rectangle (3.7,3.9);

	\tikzstyle{colvec} = [-{>[width=2mm,length=2mm]}, thick]

	\draw[draw=none, fill=shade1] (3,1.5)  -- (3.5,2) -- (1.3,3.9) --  (0,0) -- cycle {};
	\draw[draw=none, fill=shade2] (3,1.5) -- (3.7,1) -- (3.2, 0) -- (1.5,-2) -- (0,-1.5) -- (0,0) -- cycle {};
	\draw[draw=none, fill=shade3] (-1.75,0) -- (-1.3,3.6) -- (1.3,3.9) -- (0,0) -- cycle {};
	\draw[draw=none, fill=shade4] (-2.2,0) -- (-3,-2) -- (-2.5,-2.5) -- (-2,-3) -- (0,-1.5) -- (0,0) -- cycle {};

	\draw[colvec, darkpastelblue] (0,0) -- (2,1) {};
	\node[darkpastelblue] at (2.3,0.8) {$M_1$};
	\node[darkpastelblue] at (2.9,0.8) {$\begin{bmatrix} 2\\ 1\\ \end{bmatrix}$};
	\draw[colvec, darkpastelblue] (0,0) -- (1,3) {};
	\node[darkpastelblue] at (0.5,2.9) {$M_2$};
	\node[darkpastelblue] at (1.4,2.9) {$\begin{bmatrix} 1\\ 3\\ \end{bmatrix}$};
	\draw[colvec, darkpastelred] (0,0) -- (-1,0) {};
	\node[darkpastelred] at (-1,0.3) {$-e_1$};
	\draw[colvec, darkpastelred] (0,0) -- (0,-1) {};
	\node[darkpastelred] at (0.45,-0.85) {$-e_2$};

	\draw[-, darkpastelbrown, thick] (-2,-1) -- (1,1);
	\node[fill=darkpastelbrown, circle, inner sep=1.5pt] at (0.142,0.429) {};
	\node[fill=darkpastelbrown, circle, inner sep=1.5pt] at (-0.5,0) {};


	\node[fill=darkpastelmagenta, circle, inner sep=1.5pt] at (1,1) {};
	\node[darkpastelmagenta] at (1.35,1.55) {$-q = \begin{bmatrix} 1\\ 1\\ \end{bmatrix}$};

	\node[fill=darkpastelbrown, circle, inner sep=1.5pt] at (-2,-1) {};
	\node[darkpastelbrown] at (-2,-1.8) {$d = \begin{bmatrix} -2\\ -1\\ \end{bmatrix}$};


\end{tikzpicture}
}
\hskip -1.75cm
\scalebox{0.96}{ 
\tikzset{arrowstyle/.style={draw=black, fill=black, single arrow, minimum height=#1, minimum width=6ex, single arrow head extend=.4cm,}}
\newcommand{\tikzfancyarrow}[3]{\node [arrowstyle=1cm, rotate=-90] at (#1,#2) {};}

\newcommand{\zuparrow}[2]{
	\draw[->,line width=0.1mm] ([yshift=2mm]$(#1)!0.5!(#2)$) -- node[left] {\small $z$} ([yshift=5mm]$(#1)!0.5!(#2)$);
}
\newcommand{\zdnarrow}[2]{
	\draw[->,line width=0.1mm] ([yshift=5mm]$(#1)!0.5!(#2)$) -- node[left] {\small $z$} ([yshift=2mm]$(#1)!0.5!(#2)$);
}

\pgfdeclarelayer{background}
\pgfsetlayers{background,main}

\tikzstyle{vertex}=[draw,ultra thick,circle,fill=white,minimum size=10pt,inner sep=0pt]
\tikzstyle{red vertex} = [vertex, fill=pastelred]
\tikzstyle{blue vertex} = [vertex, fill=pastelblue]
\tikzstyle{edge} = [draw,thick,->]
\tikzstyle{weight} = []

\begin{tikzpicture}[scale=1,swap]

\path[use as bounding box] (-1.7,-3) rectangle (7.7,3);
\clip (-1.7,-3) rectangle (7.7,3);

\node (ray) at (-1.5,1) {};

\foreach \pos/\name in {
{(0,1)/start},
{(1,1)/1},
{(2,1)/2},
{(3,1)/3}, 
{(4,1)/4}, 
{(5,1)/5}, 
{(6,1)/6}, 
{(7,1)/end}}
	\node[vertex] (\name) at \pos {};

\path[edge] (ray) -- node[below,xshift=-1mm, align=center, fontscale=-1] {primary\\ ray} node () {} (start);

\zdnarrow{ray}{start}
\zdnarrow{start}{1}
\zdnarrow{1}{2}
\zuparrow{2}{3}
\zuparrow{3}{4}
\zuparrow{4}{5}
\zdnarrow{5}{6}
\zdnarrow{6}{end}



\foreach \source/ \dest /\weight in
	{
	start/1/,
	1/2/,
	2/3/,
	3/4/,
	4/5/,
	5/6/,
	6/end/
	}
	\path[edge] (\source) -- node[weight] {$\weight$} (\dest);


\node[red vertex] (10) at (-1,-1) [label=above:{$\mathbf 0$}] {};

\foreach \pos/\name in {
{(0,-1)/11},
{(1,-1)/12},
{(2,-1)/13}, 
{(3,-1)/14}, 
{(4,-1)/15},
{(5,-1)/16},
{(6,-1)/17},
{(7,-1)/18}}
	\node[vertex] (\name) at \pos {};

\draw[edge,->] (10) edge [loop left] node [below]  {\tiny $P$} ();
\draw[edge,->] (10) edge [bend left] node [above] {\tiny $S$} (11);
\draw[edge,->] (11) edge [bend left] node [below] {\tiny $P$} (10);

\draw[edge,->] (11) edge [bend left] node [above] {\tiny $S$} (12);
\draw[edge,->] (12)	edge [bend left] node [below] {\tiny $P$} (11);

\draw[edge,->] (12) edge [bend left] node [above] {\tiny $S$} (13);
\draw[edge,->] (13)	edge [bend left] node [below] {\tiny $P$} (12);
\draw[edge,->] (13) edge [loop right] node [above]  {\tiny $S$} ();

\draw[edge,->] (14) edge [loop below] node [right]  {\tiny $S,P$} ();
\draw[edge,->] (15) edge [loop below] node [right]  {\tiny $S,P$} ();

\draw[edge,->] (16) edge [loop left] node [below]  {\tiny $P$} ();
\draw[edge,->] (16) edge [bend left] node [above] {\tiny $S$} (17);
\draw[edge,->] (17)	edge [bend left] node [below] {\tiny $P$} (16);

\draw[edge,->] (17) edge [bend left] node [above] {\tiny $S$} (18);
\draw[edge,->] (18)	edge [bend left] node [below] {\tiny $P$} (17);
\draw[edge,->] (18) edge [loop right] node [above]  {\tiny $S$} ();

\node[blue vertex] at (18) {};
\node[blue vertex] at (16) {};
\node[blue vertex] at (13) {};

\draw[-{>[scale=2.5,
          length=2,
  		  width=8]}, line width=2mm] (3,0.4) -- (3,-0.4); 

\node[align=center] at (3,2.25) {Lemke path showing increase/decrease of $z$};
\node[align=center] at (3,-2.25) {\textsc{EndOfPotentialLine} instance}; 

\end{tikzpicture}
}
\end{center}
\caption{
Left: geometric view of Lemke's algorithm as inverting a piecewise linear map.
Right: construction of $S$ and $P$ for \EOPL instance from the Lemke
path, where the arrows on its edges indicate
whether $z$ increases or decreases along the edge.}
\label{fig:lemke}
\end{figure}
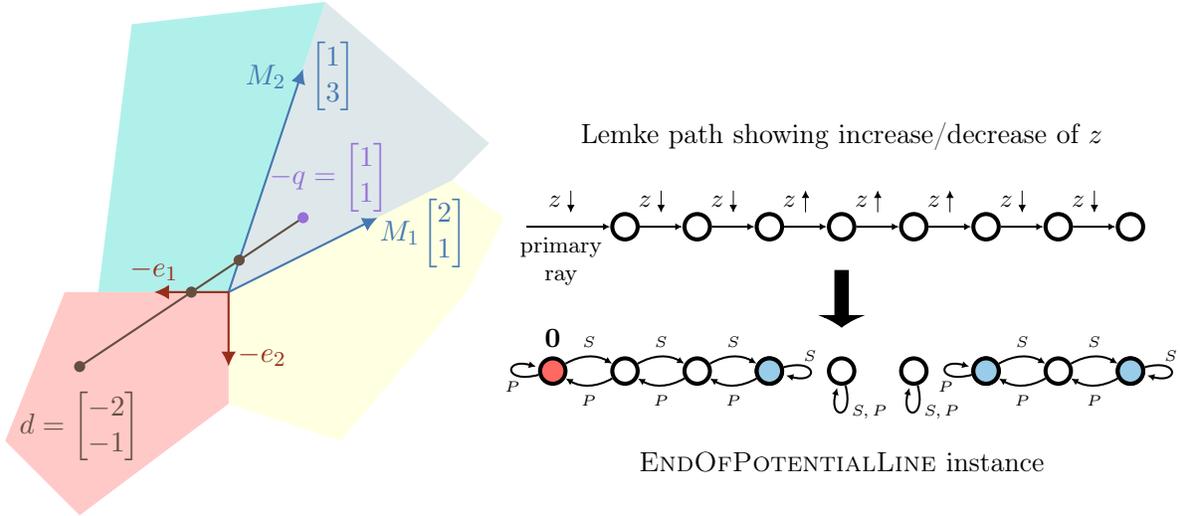   

\medskip

\noindent \textbf{Overview of \PLCP to \EOPL.}
Our reduction is based on Lemke's algorithm, explained in detail
in Section~\ref{sec:lemke}.
%
Lemke's algorithm introduces an extra variable $z$ and an extra column $d$, 
called a covering vector,
and follows a path along edges of the new LCP polyhedron 
based on a complementary pivot rule that maintains an almost-complementary solution.  
Geometrically, solving an LCP is equivalent to finding a \emph{complementary
cone}, corresponding to a subset of columns of $M$ and the complementary 
unit vectors, that contains $-q$. This is depicted on the left
in~Figure~\ref{fig:lemke}, which also shows Lemke's algorithm as inverting a
piecewise linear map along the line from $d$ to $-q$\footnote{Our figures should ideally be viewed in color.}. 
The algorithm pivots between the brown vertices at the
intersections of complementary cones and terminates at $-q$. The extra variable
$z$ can be normalized and then describes how far along the line from $d$ to $-q$
we are. P-matrix LCPs are exactly those where the complementary cones cover the
whole space with no overlap, and then $z$ decreases monotonically as Lemke's
algorithm proceeds.

The vertices along the Lemke path correspond to a subset of $[n]$ of 
basic variables. These subsets correspond to the bit strings describing
vertices in our \EOPL instance.
We use the following key properties of Lemke's algorithm as applied to a P-matrix:
\begin{enumerate}
\itemsep0mm
\item If Lemke's algorithm does not terminate with a LCP solution \PLo, it provides a \PLt solution.
\item For a P-matrix, the extra variable $z$ strictly decreases in each step of the algorithm.
\item Given a subset of $[n]$, we can efficiently decide if it corresponds to a vertex on a Lemke path.
\item By a result of Todd~\cite[Section 5]{todd1976orientation}, the Lemke path can be locally oriented.
\end{enumerate}

The starting point of the \EOPL instance corresponds to $d$. We will
then follow the line given to us by Lemke's algorithm, and we will use $z$ as
the potential function of this line. 
If we start with a \PLCP instance where $M$ is actually a P-matrix 
then this reduction will produce a single line from $d$ to the solution of the
\PLCP, and $z$ will monotonically decrease along this line. 
The main difficulty of the reduction is dealing with the case where $M$ is not a
$P$-matrix. This may cause Lemke's algorithm to terminate without an LCP solution.
Another issue is that, even when Lemke's algorithm does find a
solution, $z$ may not decrease monotonically along the line. 

In the former case, the first property above gives us a \PLt solution for the
\PLCP problem. In the latter case, we define any point on the line where $z$
increases to be a self-loop, breaking the line at these points.
Figure~\ref{fig:lemke} shows an example, where the two vertices at which $z$
increases are turned into self loops, thereby introducing two new solutions before
and after the break.
Both of these solutions give us a \PLt solution for the \PLCP instance. The 
full details of the reduction are
involved and appear in Appendix~\ref{sec:full_plcp_reduction}. It is worth
noting that, in the case where the input is actually a P-matrix, the resulting
\EOPL instance has a unique line. So, we have the following.

\begin{theorem}
\PLCP is in \EOPLc. \PromisePLCP is in \UniqueEOPLc.
\end{theorem}

Our reduction produces, for an LCP defined by an $n \times n$ matrix $M$, an
EOPL instance with $O(2^n)$ vertices. Thus, Aldous' algorithm can be applied to
give the following corollary.

\begin{corollary}
There is a randomized algorithm for \PLCP that runs in expected time $O(1.4143^n)$.
\end{corollary}

\section{Piecewise Linear Contraction Maps}
\label{sec:lcm2eopl}


\todo[inline]{State somewhere: The well-known Banach fixpoint theorem states
that all contraction maps have a \emph{unique} fixpoint~\cite{Banach1922}.}

We study contraction maps where $f$ is given as a \emph{\LinearFIXP circuit},
which is an arithmetic circuit comprised of $\max, \min,
+, -$, and $\mz$ (multiplication by a constant) gates~\cite{EtessamiY10}. Hence, a 
\LinearFIXP circuit defines a \emph{piecewise linear} function.

\begin{definition}[\LCM]
The input is a \LinearFIXP circuit computing $f : [0, 1]^d
\rightarrow [0,1]^d$, a constant $c \in (0, 1)$, and a $p \in \Natural \cup \{\infty\}$. 
It is promised that $f$ is a contraction map in the $\ell_p$
norm with contraction factor $c$. We are required to find the $x \in [0, 1]^d$
such that $f(x) = x$.
\end{definition}

\LCM is the \emph{promise} version of contraction. We reduce it to \EOPL. In the case
where the promise is not satisfied (i.e., $f$ is not actually
contracting), our reduction will detect this fact.
However, this does not allow us to reduce from the non-promise version, since in
the case where the \EOPL solution is not a fixpoint, we do not have an easy way
to produce two points that violate contraction, even though two such points must exist.

\smallskip

\noindent \textbf{UFEOPL.}
The first step in our reduction will be to reduce to a problem that we call
\UFEOPL. This is a version of \EOPL with two changes. 
Firstly, it is guaranteed (via a promise) that there is a unique line.
Secondly, the vertices on this line are only equipped with circuits that give
the next point on the line, and unlike EOPL, there is no circuit giving the
predecessor. Otherwise, the problem is unchanged. So the goal is still to find a
vertex on the line where the potential increases, or the end of the line.
Formally, we define the problem as follows.
\begin{definition}[\UFEOPL]
The input is a tuple $(C, S, V)$ of
Boolean circuits $C : \{0,1\}^n \to \{0,1\}$, $S : \{0,1\}^n \to \{0,1\}^n$,
and $V: \{0,1\}^n \to \{0,1,\dotsc,2^m - 1\}$, with the property that $C(0^n) =
1$ and $V(0^n) = 0$.
It is promised that, for every vertex $x \in \{0, 1\}^n$ such that $C(x) = 1$, there exists a $k$ such that $x = S^k(0^n)$.
We must output one of the following.
\begin{enumerate}[label=(U\arabic*)]
\itemsep1mm
\item A vertex $x \in \{0, 1\}^n$ such that $C(x) = 1$ and 
$V(S(x)) - V(x) \leq 0$.
\item 
A vertex $x \in \{0, 1\}^n$ such that $C(x) = 1$ and $C(S(x)) = 0$.
\end{enumerate}
\end{definition}
The circuits $S$ and $V$ give the successor and potential for each vertex on the
line, as usual. The extra circuit $C$ is used to determine whether a vertex is
on the line or not. This was not necessary for EOPL, because the successor and
predecessor circuits could be used to check this.
Observe that the promise guarantees
that every vertex on the line can be reached from the start point of the line,
so there is indeed only one line whenever the promise holds.

\smallskip

\noindent \textbf{Discretizing the problem. }
\begin{figure}
\begin{center}
\scalebox{0.9}{
\begin{tikzpicture}[scale=0.5]
\draw[thick,step=1,xshift=0.5cm,yshift=0.5cm] (0,0) grid (10,8); 
\foreach \x/\y in {1/2,2/4,3/6,4/5,5/4,6/3,7/3,8/3,9/5,10/6}
{
	\updncol{\x}{\y}
}
\end{tikzpicture}
}
\hskip 2cm
\scalebox{0.9}{
\begin{tikzpicture}[scale=0.5]
\draw[thick,step=1,xshift=0.5cm,yshift=0.5cm] (0,0) grid (10,8); 
\foreach \x/\y in {5/1,6/2,7/3,6/4,5/5,5/6,4/7,3/8}
{
	\leftrightcol{\x}{\y}
}
\end{tikzpicture}
}
\end{center}
\caption{Left: A direction function for the up/down dimension. Right: A
direction function for the left/right dimension.}
\label{fig:direction}
\end{figure}
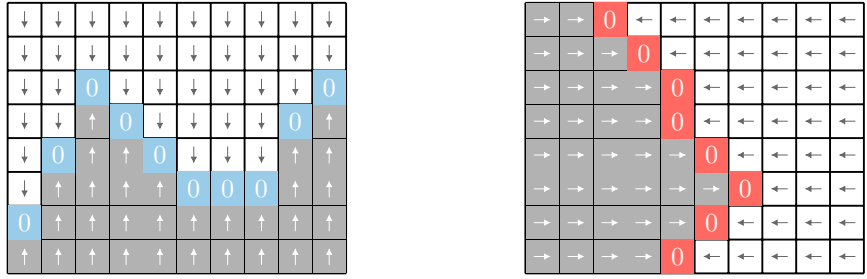
We discretize the space $[0,1]^d$ by defining a grid of points. This grid will
be represented by the set $P$ and each element of $P$ will be a tuple $p = (p_1,
p_2, \dots, p_d)$, where each $p_i$ is a rational whose 
denominator is polynomially sized. 

We also discretize the contraction map $f$ itself in the following way. For each
dimension $d$ we define a \emph{direction} function $D_i : P \rightarrow \{\up,
\down, \zero\}$, where for each point $p \in P$, we have:
if $f(p)_i > p$ then $D_i(p) = \up$,
if $f(p)_i < p$ then $D_i(p) = \down$, and
if $f(p)_i = p$ then $D_i(p) = \zero$.
In other words, the function $D_i$ simply outputs whether $f(p)$ moves up, down,
or not at all in dimension~$i$. So a fixpoint of $f$ is a point $p \in P$
such that $D_i(p) = \zero$ for all $i$.

We define the \DCM problem to be the problem of finding such a point, and we
show that \LCM can be reduced in polynomial time to \DCM. The key part of this
proof is that, since the input to \LCM is defined by a \LinearFIXP circuit, we
are able to find a suitably small grid such that there will actually exist a
point $p \in P$ with $D_i(p) = \zero$ for all $i$. The details of this are
deferred to Appendix~\ref{sec:discretizing}.

\smallskip

\noindent \textbf{A two-dimensional example.}
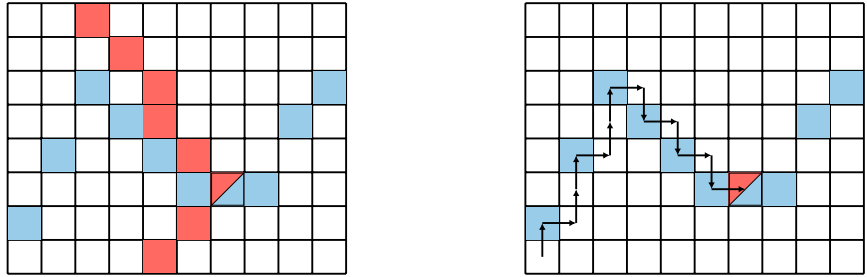
\begin{figure}
\begin{center}
\scalebox{0.9}{
\begin{tikzpicture}[scale=0.5]
\draw[thick,step=1,xshift=0.5cm,yshift=0.5cm] (0,0) grid (10,8); 
\foreach \x/\y in {1/2,2/4,3/6,4/5,5/4,6/3,7/3,8/3,9/5,10/6}
{
	\mysquare{\x}{\y}{pastelblue}
}
\foreach \x/\y in {5/1,6/2,6/4,5/5,5/6,4/7,3/8}
{
	\mysquare{\x}{\y}{pastelred}
}

\mytwocoloursquare{7}{3}

\end{tikzpicture}
}
\hskip 2cm
\scalebox{0.9}{
\begin{tikzpicture}[scale=0.5]
\draw[thick,step=1,xshift=0.5cm,yshift=0.5cm] (0,0) grid (10,8); 
\foreach \x/\y in {1/2,2/4,3/6,4/5,5/4,6/3,7/3,8/3,9/5,10/6}
{
	\mysquare{\x}{\y}{pastelblue}
	
}
\mytwocoloursquare{7}{3}

\patharrow{1}{1}{1}{2}
\patharrow{1}{2}{2}{2}
\patharrow{2}{2}{2}{3}
\patharrow{2}{3}{2}{4}
\patharrow{2}{4}{3}{4}
\patharrow{3}{4}{3}{5}
\patharrow{3}{5}{3}{6}
\patharrow{3}{6}{4}{6}
\patharrow{4}{6}{4}{5}
\patharrow{4}{5}{5}{5}
\patharrow{5}{5}{5}{4}
\patharrow{5}{4}{6}{4}
\patharrow{6}{4}{6}{3}
\patharrow{6}{3}{7}{3}
\end{tikzpicture}
}
\end{center}
\caption{Left: The red and blue surfaces. Right: the path that we follow.}
\label{fig:rbsurface}
\end{figure}
To illustrate our construction, we first give an example in two dimensions.
Figure~\ref{fig:direction} gives two direction functions for a two-dimensional
problem. The figure on the left shows a direction function for the up-down
dimension, which we will call dimension $1$ and illustrate using the color blue.
The figure on the right shows a direction function for the left-right dimension,
which we will call dimension $2$ and illustrate using the color red. Each square
in the figures represents a point in the discretized space, and the value of the
direction function is shown inside the box. 

On the left side in Figure~\ref{fig:rbsurface}, we have overlaid the \emph{surfaces} of these two
functions. The surface of a direction function $D_i$ is exactly the set of points $p
\in P$ such that $D_i(p) = \zero$. The fixpoint $p$ that we seek has $D_i(p)
= \zero$ for all dimensions $i$, and so it lies at the intersection of these
surfaces.

To reach this intersection, we will walk along the blue surface until we reach
the point that is on the red and blue surface. The path starts at the bottom
left-hand corner of the diagram, which corresponds to the point $(0, 0)$. It
initially finds the blue surface by following the blue direction function $D_1$,
which creates a path that walks up until the blue surface is found. 


Once we have found the point on the surface we then move one step in dimension
2. To do this, we follow the direction given by the red direction function
$D_2$, and so in our example, we move one step to the right. Once we have done
so, we then have to find the blue surface again, which again involves following
the direction of $D_1$. 
%
%
By repeatedly finding the blue surface, and the following the direction given by
$D_2$, we eventually arrive at the point that is on the red and blue surfaces,
which is the fixpoint that we seek. The line that we follow is shown in the
right-hand side of Figure~\ref{fig:rbsurface}.

We need two properties to make this approach work. Firstly, we must ensure that
every \emph{slice} of dimension 2 has a unique point in the blue surface, that
is, if we fix a coordinate $y$ in dimension $2$, then there is exactly one point
$(x, y)$ such that $D_1(x, y) = \zero$. Secondly, it must be the case that at
any point $p$ on the blue surface, the function $D_2(p)$ tells us the correct
direction to walk to find the fixpoint. We ensure that our \DCM instance
satisfies these properties, and the details are given in
Appendix~\ref{sec:discretizing}.

\smallskip

\noindent \textbf{The potential.}
How do we define a potential for this line? Observe that the dimension-two 
coordinates of the points on the line are weakly monotone, meaning that the line
never moves to the left. Furthermore, for any
dimension-two slice (meaning any slice in which the left/right coordinate is
fixed), the line is either monotonically increasing, or monotonically
decreasing, depending on the direction specified by $D_1$. So, if $k$ denotes
the maximum coordinate in either dimension, then the function
\begin{equation*}
V(p_1, p_2) = \begin{cases}
k \cdot p_2 + p_1 & \text{if $D_1(p_1, p_2) = \up$,} \\
k \cdot p_2 + (k-p_1) & \text{if $D_1(p_1, p_2) = \down$.} \\
\end{cases}
\end{equation*}
is a function that monotonically increases along the line.

\smallskip

\noindent \textbf{Uniqueness.}
Finally, we must argue that the line is unique. Here we must carefully define
what exactly a vertex on the line is, so that the circuit $C$ can correctly
identify the points that are on the line. In particular, when we are walking
towards the blue surface, we must make sure that we are either coming from the
point $(0, 0)$ at the start of the line, or that we have just visited a point on
the blue surface. For example, consider the third column of points from the left
in Figure~\ref{fig:rbsurface}. The first three points from the bottom of the
diagram should be identified as not being on the line, while the second three
points should be identified as being on the line.
To do this, we make the points on our line tuples $(q, p)$, where $p$ is the
current point on the line, and $q$ is either 
\begin{itemize}
\itemsep0.5mm
\item the symbol $\vblank$, indicating that we are the start of the line and
$p_2$ should be equal to $ 0$, or
\item a point $q$ that is on the blue surface, and satisfies $q_2 = p_2 - 1$. 
\end{itemize}
In the latter case, the point $q$ is the last point on the blue surface that was
visited by the line, and it can be used to determine whether $p$ is a valid
point. In our example in the third column, $q$ would be the point on the blue
surface in column $2$, and so the three points below $q$ in the third column can
be determined to not lie on the line, since their direction function points
upwards, and thus no line that visited $q$ can have arrived at these points.
Meanwhile, the points above $q$ do lie on the line, for exactly the same reason.

Once the line arrives at the next point on the blue surface, the point $q$ is
then overwritten with this new point. This step is the reason why only a
successor circuit can be given for the line, since the value that is overwritten
cannot easily be computed by a predecessor circuit.

\smallskip

\noindent \textbf{The full reduction.} 
Our reduction from \DCM to \UFEOPL generalizes the approach given above to $d$
dimensions.
We say that a point $p \in P$ is on the
\emph{$i$-surface} if $D_j(p) = \zero$ for all $j \le i$. In our two-dimensional
example we followed a line of points on the one-surface, in order to find a
point on the two-surface. In between any two points on the one-surface, we
followed a line of points on the zero-surface (every point is trivially on the
zero-surface).

Our line will visit a sequence of points on the $(d-1)$-surface in
order to find the point on the $d$-surface, which is the fixpoint. Between
any two points on the $(d-1)$-surface the line visits a sequence of points on
the $(d-2)$-surface, between any two points on the 
$(d-2)$-surface the line visits a sequence of points on the $(d-3)$-surface, and
so on. 

These points will satisfy a \emph{recursive monotonicity} property, which is a
generalization of the property that we saw in two-dimensions. The sequence of
points on the $(d-1)$-surface are monotonically increasing in the $d$th
coordinate. Between every pair of points on the $(d-1)$-surface, the sequence of
points on the $(d-2)$-surface will be either monotonically increasing or
monotonically decreasing in the $(d-1)$th coordinate, and so on. This will allow
us to define a potential function that generalizes the one that we saw in the
two-dimensional case. 

Full details of this construction are given in
Appendix~\ref{app:dcm2ufeopl}, where the following lemma is proved.

\begin{lemma}
\label{lem:dcm2ufeopl}
\DCM can be reduced in polynomial time to \UFEOPL.
\end{lemma}

\smallskip

\noindent \textbf{From \UFEOPL to \EOPL.}
The final step of the proof is to reduce \UFEOPL to \EOPL. Fortunately, we are
able to utilize prior work to perform this step of the reduction. We will
utilize the following problem that was introduced by Bitansky et
al~\cite{BPR15}. 

\begin{definition}[\SOVL~\cite{BPR15}]
The input to the problem consists of a starting vertex $x_s \in \{0, 1\}^n$, a
target integer $T \le 2^n$, and two boolean circuits $S : \{0, 1\}^n
\rightarrow \{0, 1\}^n$, $W : \{0, 1\}^n \times \{0, 1\}^n \rightarrow \{0, 1\}$.
It is promised that, for every vertex $x \in \{0, 1\}^n$, and every integer $i
\le T$, we have $W(x, t) = 1$ if and only if $x = S^{i-1}(x_s)$. The goal is
to find the vertex $x_f \in \{0, 1\}^n$ such that $W(x_f, T) = 1$.
\end{definition}

It was shown by Hub\'a\v{c}ek and Yogev~\cite{hubavcek2017hardness}  that \SOVL can be reduced in polynomial
time to \EOML, and hence also to \EOPL. 
So, to
complete our proof, we need to reduce \UFEOPL to \SOVL. We are able to do this
because we are guaranteed that the line in our instance is unique, and so we can
use the potential of that line to determine how far we are from the start of the
line. However, this does not directly give us the exact number of steps between
the start of the line, which we need to construct the circuit $W$. So, we first
reduce to a unique forward version of \EOML, using the same technique as we used
in Theorem~\ref{thm:eoml2eopl}, and then reduce that to \SOVL. The full
details of the proof can be found in Appendix~\ref{app:ufeopl2eopl}.

\begin{lemma}
\label{lem:ufeopl2eopl}
\UFEOPL can be reduced in polynomial time to \UniqueEOPLc. 
\end{lemma}

To summarise, we obtain the following theorem. 

\begin{theorem}
\label{thm:lcm}
\LCM is in \UniqueEOPLc.
\end{theorem}


\section{Algorithms for Contraction Maps}
\label{sec:algorithms}

\paragraph{\bf An algorithm for \LCM.}

The properties that we observed in our reduction from \LCM to \EOPL can also be
used to give polynomial time algorithms for the case where the number of
dimensions is constant. In our two-dimensional example, we relied on the fact
that each dimension-two slice has a unique point on the blue surface, and 
that the direction function at this point tells us the direction of the overall
fixpoint. 

This suggests that a nested binary search approach can be used to find the
fixpoint. The outer binary search will work on dimension-two coordinates, and
the inner binary search will work on dimension-one coordinates. For each fixed
dimension-two coordinate $y$, we can apply the inner binary search to find the
unique point $(x, y)$ that is on the blue surface. Once we have done so, $D_2(x,
y)$ tells us how to update the outer binary search to find a new candidate
coordinate $y'$. 

This can be generalized to $d$-dimensional instances, by
running $d$ nested instances of binary search. Doing so yields the following
theorem, whose proof appears in Appendix~\ref{subsec:exact_algo_details}.

\begin{theorem}
Given a $\LinearFIXP$ circuit $C$ encoding a contraction map $f : [0,1]^d\to
[0,1]^d$ with respect to any $\ell_p$ norm, there is an algorithm to find a
fixpoint of $f$ in time that is polynomial in $\size(C)$ and exponential in $d$.
\end{theorem}

\paragraph{\bf An algorithm for \CM.}

We are also able to generalize this to the more general \CM problem, where the
input is given as an arbitrary (non-linear) arithmetic circuit. Here
the key issue is that the fixpoint may not be rational, and so we must find a
suitably accurate approximate fixpoint. Our nested binary search
approach can be adapted to do this.

Since we now deal with approximate fixpoints, we must cut off each of our nested
binary search instances at an appropriate accuracy. Specifically, we must ensure
that the solution is accurate enough so that we can correctly update the outer
binary search. 
Choosing these cutoff points turns out to be quite involved, as we must choose
different cutoff points depending on both the norm and the level of recursion,
and moreover the $\ell_1$ case requires a separate proof.
The details of this are deferred to
Appendix~\ref{subsec:approx_algo_details}, where the following theorem is shown.

\begin{theorem}
\label{thm:alggeneral}
For a contraction map $f:[0,1]^d\to [0,1]^d$ under $\Norm{\cdot}_p$ for $2 \leq
p < \infty$, there is an algorithm to compute a point $v\in [0,1]^d$ such that
$\Norm{f(v) - v}_p < \eps$ in time $O(p^{d^2}\log^d(1/\eps)\log^d(p))$.
\end{theorem}

Actually, our algorithm treats the function as a black-box, and so it can be
applied to any contraction map, with Theorem~\ref{thm:alggeneral} giving the
number of queries that need to be made.


\section{Open Problems}

Many interesting open questions arise from our work.
In the case of finding a Nash equilibrium of a two-player game, which we now
know is \PPAD-complete~\cite{chen2009settling,daskalakis2009complexity}, the
definition of \PPAD was inspired by the path structure of the Lemke-Howson
algorithm.
Our definition of \EOPL is directly inspired by the path structure
of Lemke paths for P-matrix LCPs, as well as combining the canonical definitions
of \PPAD and \PLS.
Thus, the first conjecture we make is:

\begin{conjecture}
\PLCP is hard for \EOPLc and promise \PLCP is hard for \UniqueEOPLc.
\end{conjecture}

This is actually two conjectures, but many natural approaches for proving the
first result would almost automatically prove the second.
Similar to those conjectures:

\begin{conjecture}
\LCM is hard for \UniqueEOPLc.
\end{conjecture}

We defined \LCM as a promise problem in this paper.
The reason was that if we start with a non-contracting instance of \LCM 
then the \EOPL instance we produce will either give us a fixpoint or a short
proof that the instance was not contracting, but we do not know how to turn 
this proof into an actual explicit violation of contraction. 
Solving this issue is one interesting direction.

Another interesting question is whether the more general \CM is in \EOPL. 
\CM is defined by a general arithmetic circuit, and may have an
irrational fixpoint, so the problem is to find an approximate fixpoint. 
A natural approach is to approximate general circuits with \LinearFIXP circuits,
e.g., with interpolation, where the additional requirement would be to maintain
the contraction property.

Finally, we have shown that \EOPLc is contained in \CLS. We conjecture that:

\begin{conjecture}
\CLS = \EOPLc.
\end{conjecture}
It is not at all clear that this is true, and both possible outcomes would be
interesting. 
Given the recent result of Daskalakis, Tzamos, and
Zampetakis~\cite{DTZ17}, one approach to proving
this conjecture would be to show that \CM with a metric as input is in \EOPLc.
This seems challenging but plausible. If this conjecture is actually false,
then it is natural to ask which further problems known to be in \CLS also
lie in \EOPLc.

\todo[inline]{Do we want to add PL-Contraction to P-LCP, and vice-versa? Otherwise reads great. - Ruta.}
\newpage

\bibliography{paper}

\newpage
\appendix

\section{The Full Reductions and Proofs for Section~\ref{sec:EOPL}}
\label{app:eoml2eopl}

\subsection{\EOML to \EOPL}
\label{sec:EOMLtoEOPL}

Given an instance \CI of \EOML defined by circuits $S,P$ and $V$ on vertex
set $\{0,1\}^n$ we are going to create an instance $\CI'$ of \EOPL with circuits
$S',P'$, and $V'$ on vertex set $\{0,1\}^{(n+1)}$, i.e., we introduce one extra bit.  
This extra bit is essentially to take care of the difference in the value of potential 
at the starting point in \EOML and \EOPL, namely $1$ and $0$ respectively. 

Let $k=n+1$, then we create a potential function $V':\{0,1\}^k \rightarrow
\{0,\dots,2^k-1\}$. 
The idea is to make $0^k$ the starting point with potential zero as required,
and to make all other vertices with first bit $0$ be dummy vertices with self
loops. The real graph
will be embedded in vertices with first bit $1$, i.e., of type $(1,\uu)$. Here
by $(b,\uu)\in \{0,1\}^k$, where $b\in \{0,1\}$ and $\uu\in \{0,1\}^n$, we mean
a $k$ length bit string with first bit set to $b$ and for each $i\in[2:k]$ bit $i$ 
set to bit $u_i$. 

\medskip
\medskip

\noindent{\bf Procedure $V'(b,\uu)$:} If $b=0$ then Return $0$, otherwise Return $V(\uu)$. 
\medskip
\medskip

\noindent{\bf Procedure $S'(b,\uu)$:}
\vspace{-0.3cm}

\begin{enumerate}
\itemsep1mm
\item If $(b,\uu)=0^k$ then Return $(1,0^n)$
\item If $b=0$ and $\uu\neq 0^n$ then Return $(b,\uu)$ (creating self loop for dummy vertices)
\item If $b=1$ and $V(\uu)=0$ then Return $(b,\uu)$ (vertices with zero potentials have self loops)
\item If $b=1$ and $V(\uu)>0$ then Return $(b,S(\uu))$ (the rest follows $S$)
\end{enumerate}

\noindent{\bf Procedure $P'(b,\uu)$:}
\vspace{-0.3cm}

\begin{enumerate}
\itemsep1mm
\item If $(b,\uu)=0^k$ then Return $(b,\uu)$ (initial vertex points to itself in $P'$).
\item If $b=0$ and $\uu\neq 0^n$ then Return $(b,\uu)$ (creating self loop for dummy vertices)
\item If $b=1$ and $\uu=0^n$ then Return $0^k$ (to make $(0,0^n)\rightarrow (1,0^n)$ edge consistent)
\item If $b=1$ and $V(\uu)=0$ then Return $(b,\uu)$ (vertices with zero potentials have self loops)
\item If $b=1$ and $V(\uu)>0$ and $\uu \neq 0^n$ then Return $(b,P(\uu))$ (the rest follows $P$)
\end{enumerate}

Valid solutions of \EOML of type T2 and T3 requires the potential to be strictly greater than zero, while solutions of \EOPL may have zero potential. However, a solution of \EOPL can not be a self loop, so we've added self-loops around vertices with zero potential in the \EOPL instance.
By construction, the next lemma follows:
\begin{lemma}\label{lem:m2p-valid}
$S'$, $P'$, $V'$ are well defined and polynomial in the sizes of $S$, $P$, $V$ respectively. 
\end{lemma}

Our main theorem in this section is a consequence of the following three lemmas.

\begin{lemma}\label{lem:m2p-sl}
For any $\xx=(b,\uu)\in \{0,1\}^k$, $P'(\xx)=S'(\xx)=\xx$ (self loop) iff $\xx\neq 0^k$, and $b=0$ or $V(\uu)=0$.
\end{lemma}
\begin{proof}
This follows by the construction of $V'$, the second condition in $S'$ and $P'$, and third and fourth conditions in $S'$ and $P'$ respectively. 
\end{proof}

\begin{lemma}\label{lem:m2p-r1}
Let $\xx=(b,\uu)\in \{0,1\}^k$ be such that $S'(P'(\xx))\neq \xx \neq 0^k$ or $P'(S'(\xx))\neq \xx$ (an R1 type solution of \EOPL instance $\CI'$), then $\uu$ is a solution of \EOML instance $\CI$.
\end{lemma}
\begin{proof}
The proof requires a careful case analysis. 
By the first conditions in the descriptions of $S',P'$ and $V'$, we have $\xx \neq 0^k$. 
Further, since $\xx$ is not a self loop, Lemma \ref{lem:m2p-sl} implies $b=1$  and $V'(1,\uu)=V(\uu)>0$.
\medskip

\noindent{\em Case I.}
If $S'(P'(\xx))\neq \xx\neq 0^k$ then we will show that either $\uu$ is a genuine start of a line other than $0^n$ giving a T1 type solution of \EOML instance $\CI$, or there is some issue with the potential at $\uu$ giving either a T2 or T3 type solution of $\CI$. Since $S'(P'(1,0^n))=(1,0^n)$, $\uu \neq 0^n$. Thus if $S(P(\uu))\neq \uu$ then we get a T1 type solution of $\CI$ and proof follows. If $V(\uu)=1$ then we get a T2 solution of $\CI$ and proof follows. 

Otherwise, we have $S(P(\uu))=\uu$ and $V(\uu)>1$. Now since also $b=1$ $(1,\uu)$ is not a self loop (Lemma \ref{lem:m2p-sl}). 
Then it must be the case that $P'(1,\uu)=(1,P(\uu))$. However, $S'(1,P(\uu))\neq (1,\uu)$ even though $S(P(\uu))=\uu$. This happens only when $P(\uu)$ is a self loop because of $V(P(\uu))=0$ (third condition of $P'$).
Therefore, we have $V(\uu)-V(P(\uu))>1$ implying that $\uu$ is a T3 type solution of $\CI$. 
\medskip

\noindent{\em Case II.}
Similarly, if $P'(S'(\xx))\neq \xx$, then either $\uu$ is a genuine end of a line of $\CI$, or there is some issue with the potential at $\uu$. If $P(S(\uu))\neq \uu$ then we get T1 solution of $\CI$. Otherwise, $P(S(\uu))=\uu$ and $V(\uu)>0$. Now as $(b,\uu)$ is not a self loop and $V(\uu)>0$, it must be the case that $S'(b,\uu)=(1,S(\uu))$. However, $P'(1, S(\uu))\neq (b,\uu)$ even though $P(S(\uu))=\uu$. This happens only when $S(\uu)$ is a self loop because of $V(S(\uu))=0$. Therefore, we get $V(S(\uu))-V(\uu)<0$, i.e., $\uu$ is a type T3 solution of $\CI$. 
\end{proof}

\begin{lemma}\label{lem:m2p-r2}
Let $\xx=(b,\uu)\in \{0,1\}^k$ be an R2 type solution of the constructed \EOPL instance $\CI'$, then $\uu$ is a type T3 solution of \EOML instance~$\CI$.
\end{lemma}
\begin{proof}
Clearly, $\xx\neq 0^k$. Let $\yy = (b',\uu') = S'(\xx) \neq \xx$, and observe that $P(\yy) = \xx$. This also implies that $\yy$ is not a self loop, and hence $b=b'=1$ and $V(\uu)>0$ (Lemma \ref{lem:m2p-sl}). Further, $\yy = S'(1,\uu)=(1,S(\uu))$, hence $\uu'=S(\uu)$. Also, $V'(\xx)=V'(1,\uu)=V(\uu)$ and $V'(\yy)=V'(1,\uu')=V(\uu')$. 

Since $V'(\yy)-V'(\xx)\le 0$ we get $V(\uu')-V(\uu)\le 0 \Rightarrow V(S(\uu)) - V(\uu) \le 0\Rightarrow V(S(\uu)) - V(\uu)\neq 1$. Given that $V(\uu)>0$, $\uu$ gives a type T3 solution of \EOML.
\end{proof}

\begin{theorem}\label{thm:m2p}
An instance of \EOML can be reduced to an instance of \EOPL in linear time such that a solution of the former can be constructed in a linear time from the solution of the latter. 
\end{theorem}

\subsection{\EOPL to \EOML}
\label{sec:eopl2eoml}

In this section we give a linear time reduction from an instance $\CI$ of \EOPL to an instance $\CI'$ of \EOML. Let the given \EOPL instance $\CI$ be defined on vertex set $\{0,1\}^n$ and with procedures $S,P$ and $V$, where $V:\{0,1\}^n\rightarrow \{0,\dots,2^m-1\}$. 
\medskip

\noindent{\bf Valid Edge.} We call an edge $\uu \rightarrow \vv$ valid if $\vv=S(\uu)$ and $\uu=P(\vv)$. 
\medskip

We construct an \EOML instance $\CI'$ on $\{0,1\}^k$ vertices where $k=n+m$. 
Let $S',P'$ and $V'$ denotes the procedures for $\CI'$ instance. 
The idea is to capture value $V(\xx)$ of the potential in the $m$ least significant bits of vertex description itself, so that it can be gradually increased or decreased on valid edges. For vertices with irrelevant values of these least $m$ significant bits we will create self loops. Invalid edges will also become self loops, e.g., if $\yy=S(\xx)$ but $P(\yy)\neq \xx$ then set $S'(\xx,.)=(\xx,.)$. We will see how these can not introduce new solutions. 

In order to ensure $V'(0^k)=1$, the $V(S(0^n))=1$ case needs to be discarded. For
this, we first do some initial checks to see if the given instance $\CI$ is not
trivial.  If the input \EOPL instance is trivial, in the sense that either
$0^n$ or $S(0^n)$ is a solution, then we can just return it.

\begin{lemma}
\label{lem:valid-edges}
If $0^n$ or $S(0^n)$ are not solutions of \EOPL instance $\CI$ then $0^n
\rightarrow S(0^n) \rightarrow S(S(0^n))$ are valid edges, and $V(S(S(0^n))\ge 2$. 
\end{lemma}

\begin{proof}
Since both $0^n$ and $S(0^n)$ are not solutions, we have
	$V(0^n)<V(S(0^n))<V(S(S(0^n)))$, $P(S(0^n))=0^n$, and for $\uu = S(0^n)$,
	$S(P(\uu))=\uu$ and $P(S(\uu))=\uu$. In other words, $0^n \rightarrow S(0^n)
	\rightarrow S(S(0^n))$ are valid edges, and since $V(0^n)=0$, we have
	$V(S(S(0^n))\ge 2$. 
\end{proof}

Let us assume now on that $0^n$ and $S(0^n)$ are not solutions of $\CI$, and
then by Lemma \ref{lem:valid-edges}, we have $0^n \rightarrow S(0^n) \rightarrow
S(S(0^n))$ are valid edges, and $V(S(S(0^n))\ge 2$. We can avoid the need to check
whether $V(S(0))$ is one all together, by making $0^n$ point directly to
$S(S(0^n))$ and make $S(0^n)$ a dummy vertex. 

We first construct $S'$ and $P'$, and then construct $V'$ which will give
value zero to all self loops, and use the least significant $m$ bits to give a
value to all other vertices.
Before describing $S'$ and $P'$ formally, we first describe the underlying
principles. Recall that in $\CI$ vertex set is $\{0,1\}^n$ and possible potential values are $\{0,\dots,2^m-1\}$, while in $\CI'$ vertex set is $\{0,1\}^k$ where $k=m+n$. 
We will denote a vertex of $\CI'$ by a tuple $(\uu,\pi)$, where $\uu \in
\{0,1\}^n$ and $\pi\in \{0,\dots,2^m-1\}$. 
Here when we say that we introduce an {\em edge $\xx\rightarrow \yy$} we mean
that we introduce a valid edge from $\xx$ to $\yy$, i.e., $\yy=S'(\xx)$ and $\xx=P(\yy)$. 
\begin{itemize}
\item Vertices of the form $(S(0^n),\pi)$ for any $\pi \in \{0,1\}^m$ and the vertex $(0^n,1)$ are
dummies and hence have self loops.
\item If $V(S(S(0^n))=2$ then we introduce an edge $(0^n,0)\rightarrow(S(S(0^n)),2)$, otherwise 
\begin{itemize}
\item for $p=V(S(S(0^n))$, we introduce the edges $(0^n,0)\ra (0^n,2)\ra (0^n, 3)\dots (0^n,p-1)\ra (S(S(0^n)),p)$.
\end{itemize}
\item If $\uu \ra \uu'$ valid edge in $\CI$ then let $p=V(\uu)$ and $p'=V(\uu')$
\begin{itemize}
\item If $p=p'$ then we introduce the edge $(\uu,p)\ra (\uu',p')$. 
\item If $p<p'$ then we introduce the edges $(\uu,p)\ra (\uu,p+1)\ra \dots\ra (\uu,p'-1)\ra (\uu',p')$.
\item If $p>p'$ then we introduce the edges $(\uu,p)\ra (\uu,p-1)\ra \dots\ra (\uu,p'+1)\ra (\uu',p')$.
\end{itemize}
\item If $\uu\neq 0^n$ is the start of a path, i.e., $S(P(\uu))\neq \uu$, then
make $(\uu,V(\uu))$ start of a path by ensuring $P'(\uu,V(\uu))=(\uu,V(\uu))$.
\item If $\uu$ is the end of a path, i.e., $P(S(\uu))\neq \uu$, then make
$(\uu,V(\uu))$ end of a path by ensuring $S'(\uu,V(\uu))=(\uu,V(\uu))$.
\end{itemize}

Last two bullets above remove singleton solutions from the system by making them
self loops. However, this can not kill all the solutions since there is a path
starting at $0^n$, which has to end somewhere. Further, note that this entire process ensures that no new start or end of a paths are introduced. 
\medskip
\medskip

\noindent{\bf Procedure $S'(\uu,\pi)$.} 
\vspace{-0.2cm}

\begin{enumerate}
\itemsep1mm
\item If ($\uu=0^n$ and $\pi=1$) or $\uu=S(0^n)$ then Return $(\uu,\pi)$. 
\item If $(\uu,\pi)=0^k$, then let $\uu'=S(S(0^n))$ and $p'=V(\uu')$. 
\begin{enumerate}
\itemsep1mm
\item If $p'=2$ then Return $(\uu',2)$ else Return $(0^n,2)$.
\end{enumerate}
\item If $\uu=0^n$ then
\begin{enumerate}
\itemsep1mm
\item If $2\le \pi<p'-1$ then Return $(0^n,\pi+1)$.
\item If $\pi=p'-1$ then Return $(S(S(0^n)),p')$.
\item If $\pi\ge p'$ then Return $(\uu,\pi)$.
\end{enumerate}
\item Let $\uu'=S(\uu)$, $p'=V(\uu')$, and $p=V(\uu)$. 
\item \label{itm:Sfive} If $P(\uu')\neq \uu$ or $\uu'=\uu$ then Return $(\uu,\pi)$
\item If $\pi=p=p'$ or ($\pi=p$ and $p'=p+1$) or $(\pi=p$ and $p'=p-1$) then Return $(\uu',p')$.
\item If $\pi<p\le p'$ or $p\le p'\le \pi$ or $\pi>p\ge p'$ or $p\ge p'\ge \pi$ then Return $(\uu,\pi)$
\item If $p<p'$, then if $p\le \pi<p'-1$ then Return $(\uu,\pi+1)$. If $\pi=p'-1$ then Return $(\uu',p')$.
\item If $p>p'$, then if $p \ge \pi>p'+1$ then Return $(\uu,\pi-1)$. If $\pi=p'+1$ then Return $(\uu',p')$.
\end{enumerate}
\medskip

\noindent{\bf Procedure $P'(\uu,\pi)$.} 
\vspace{-0.2cm}

\begin{enumerate}
\itemsep1mm
\item If ($\uu=0^n$ and $\pi=1$) or $\uu=S(0^n)$ then Return $(\uu,\pi)$. 
\item If $\uu=0^n$, then 
\begin{enumerate}
\itemsep1mm
\item If $\pi=0$ then Return $0^k$.
\item If $\pi<V(S(S(0^n)))$ and $\pi\notin \{1,2\}$ then Return $(0^n,\pi-1)$.
\item If $\pi<V(S(S(0^n)))$ and $\pi=2$ then Return $0^k$.
\end{enumerate}
\item If $\uu=S(S(0^n))$ and $\pi=V(S(S(0^n))$ then 
\begin{enumerate}
\itemsep1mm
\item If $\pi=2$ then Return $(0^n,0)$, else Return $(0^n,\pi-1)$. 
\end{enumerate}
\item If $\pi=V(\uu)$ then 
\begin{enumerate}
\itemsep1mm
\item Let $\uu'=P(\uu)$, $p'=V(\uu')$, and $p=V(\uu)$. 
\item If $S(\uu')\neq \uu$ or $\uu'=\uu$ then Return $(\uu,\pi)$
\item If $p=p'$ then Return $(\uu',p')$ 
\item If $p'<p$ then Return $(\uu',p-1)$ else Return $(\uu',p+1)$
\end{enumerate}
\item Else \% when $\pi \neq V(\uu)$
\begin{enumerate}
\itemsep1mm
\item Let $\uu'=S(\uu)$, $p'=V(\uu')$, and $p=V(\uu)$
\item If $P(\uu')\neq \uu$ or $\uu'=\uu$ then Return $(\uu,\pi)$
\item If $p'=p$ or $\pi<p< p'$ or $p<p'\le \pi$ or $\pi>p> p'$ or $p>p'\ge \pi$ then Return $(\uu,\pi)$
\item If $p<p'$, then If $p<\pi\le p'-1$ then Return $(\uu,\pi-1)$. 
\item If $p>p'$, then if $p> \pi\ge p'+1$ then Return $(\uu,\pi+1)$. 
\end{enumerate}
\end{enumerate}

As mentioned before, the intuition for the potential function procedure $V'$ is to return zero for self loops, return $1$ for $0^k$, and return the number specified by the lowest $m$ bits for the rest. 
\medskip
\medskip

\noindent{\bf Procedure $V'(\uu,\pi)$.} Let $\xx=(\uu,\pi)$ for notational convenience.
\vspace{-0.2cm}

\begin{enumerate}
\itemsep1mm
\item If $\xx=0^k$, then Return $1$. 
\item If $S'(\xx) = \xx$ and $P'(\xx)=\xx$ then Return $0$.
\item If $S'(\xx) \neq \xx$ or $P'(\xx)\neq \xx$ then Return $\pi$.
\end{enumerate}

The fact that procedures $S'$, $P'$ and $V'$ give a valid \EOML instance follows from construction.
\begin{lemma}\label{lem:p2m-valid}
Procedures $S'$, $P'$ and $V'$ gives a valid \EOML instance on vertex set $\{0,1\}^k$, where $k=m+n$ and $V':\{0,1\}^k\ra \{0,\dots, 2^k-1\}$.
\end{lemma}

The next three lemmas shows how to construct a solution of \EOPL instance $\CI$ from a type T1, T2, or T3 solution of constructed \EOML instance $\CI'$.
The basic idea for next lemma, which handles type T1 solutions, is that we never create spurious end or start of a path. 
\begin{lemma}\label{lem:p2m-t1}
Let $\xx=(\uu,\pi)$ be a type T1 solution of constructed \EOML instance $\CI'$. Then $\uu$ is a type R1 solution of the given \EOPL instance $\CI$.
\end{lemma}

\begin{proof}
Let $\Delta=2^m-1$.
In $\CI'$, clearly $(0^n,\pi)$ for any $\pi \in {1,\dots, \Delta}$ is not a start or end of a path, and $(0^n,0)$ is not an end of a path. Therefore, $\uu\neq 0^n$. Since $(S(0^n),\pi), \forall \pi\in \{0,\dots,\Delta\}$ are self loops, $\uu \neq S(0^n)$.

If to the contrary, $S(P(\uu))=\uu$ and $P(S(\uu))=\uu$. If $S(\uu)=\uu=P(\uu)$ then $(\uu,\pi),\ \forall \pi\in\{0,\dots,\Delta\}$ are self loops, a contradiction. 
\todo{I don't understand this line. Basically, the point found can't have
corresponded to a self-loop because it would then have been a self-loop too.}

For the remaining cases, let $P'(S'(\xx))\neq \xx$, and let $\uu'=S(\uu)$. 
\todo{There must have been a edge in the original if there is a successor edge
in the new instance}. There is a valid edge from $\uu$ to $\uu'$ in $\CI$. Then
we will create valid edges from $(\uu,V(\uu))$ to $(S(\uu),V(S(\uu))$ with
appropriately changing second coordinates. The rest of $(\uu,.)$ are self loops,
a contradiction. 

Similar argument follows for the case when $S'(P'(\xx))\neq \xx$. 
\end{proof}

The basic idea behind the next lemma is that a T2 type solution in $\CI'$ has
potential $1$. Therefore, it is surely not a self loop. Then it is either an end of a path or near an end of a path, or else near a potential violation. 

\begin{lemma}\label{lem:p2m-t2}
Let $\xx=(\uu,\pi)$ be a type T2 solution of $\CI'$. Either $\uu \neq 0^n$ is start of a path in $\CI$ (type R1 solution), or $P(\uu)$ is an R1 or R2 type solution in $\CI$, or $P(P(\uu))$ is an R2 type solution in $\CI$.
\end{lemma}

\begin{proof}
Clearly $\uu \neq 0^n$, and $\xx$ is not a self loop, i.e., it is not a dummy vertex with irrelevant value of $\pi$. Further, $\pi=1$. If $\uu$ is a start or end of a path in $\CI$ then done. 

Otherwise, if $V(P(\uu))>\pi$ then we have $V(\uu)\le \pi$ and hence $V(\uu)-V(P(\uu))\le 0$ giving $P(\uu)$ as an R2 type solution of $\CI$. 
If $V(P(\uu))<\pi=1$ then $V(P(\uu))=0$. Since potential can not go below zero, either $P(\uu)$ is an end of a path, or for $\uu''=P(P(\uu))$ and $\uu'=P(\uu)$ we have $\uu'=S(\uu'')$ and $V(\uu')-V(\uu'')\le 0$, giving $\uu''$ as a type R2 solution of $\CI$.
\end{proof}

At a type T3 solution of $\CI'$ potential is strictly positive, hence these solutions are not self loops. If they correspond to potential violation in $\CI$ then we get a type R2 solution. But this may not be the case, if we made $S'$ or $P'$ self pointing due to end or start of a path respectively. In that case, we get a type R1 solution. The next lemma formalizes this intuition. 

\begin{lemma}\label{lem:p2m-t3}
Let $\xx=(\uu,\pi)$ be a type T3 solution of $\CI'$. If $\xx$ is a start or end of a path in $\CI'$ then $\uu$ gives a type R1 solution in $\CI$. Otherwise $\uu$ gives a type R2 solution of $\CI$.
\end{lemma}

\begin{proof}
Since $V'(\xx)>0$, it is not a self loop and hence is not dummy, and $\uu\neq 0^n$. If $\uu$ is start or end of a path then $\uu$ is a type R1 solution of $\CI$. Otherwise, there are valid incoming and outgoing edges at $\uu$, therefore so at $\xx$. 

If $V((S(\xx))-V(\xx)\neq 1$, then since potential either remains the same or increases or decreases exactly by one on edges of $\CI'$, it must be the case that $V(S(\xx))-V(\xx)\le 0$. This is possible only when $V(S(\uu))\le V(\uu)$. Since $\uu$ is not an end of a path we do have $S(\uu)\neq \uu$ and $P(S(\uu))=\uu$. Thus, $\uu$ is a type T2 solution of $\CI$.

If $V((\xx)-V(P(\xx))\neq 1$, then by the same argument we get that for $(\uu'',\pi'')=P(\uu)$, $\uu''$ is a type R2 solution of $\CI$. 
\end{proof}

Our main theorem follows using Lemmas \ref{lem:p2m-valid}, \ref{lem:p2m-t1}, \ref{lem:p2m-t2}, and \ref{lem:p2m-t3}.

\begin{theorem}\label{thm:p2m}
An instance of \EOPL can be reduced to an instance of \EOML in polynomial time such that a solution of the former can be constructed in a linear time from the solution of the latter. 
\end{theorem}

\section{The Full Reduction and Proofs for Section~\ref{sec:PLCPtoEOPL}}

\subsection{Lemke's algorithm}
\label{sec:lemke}

The explanation of Lemke's algorithm in this section is taken from \cite{GMSV}.
The problem is interesting only when $\qq \not \geq 0$, since otherwise $\yy = 0$ is a trivial solution. Let us introduce
slack variables $\ps$ to obtain the following equivalent formulation:
\begin{equation} \label{eq:b} \MM \yy  + \ps = \pq, \ \ \ \  \yy \geq 0, \ \ \ \ \ps \geq 0 \ \ \ \ \mbox{and} \ \ \ \ y_is_i = 0,\ \forall i\in[d].  \end{equation}

Let $\CQ$ be the polyhedron in $2d$ dimensional space defined by the first three conditions; we will assume that $\CQ$ is
non-degenerate (just for simplicity of exposition; this will not matter for our reduction).  
Under this condition, any solution to (\ref{eq:b}) will be a vertex of $\CQ$, since it must satisfy $2d$
equalities. Note that the set of solutions may be disconnected.
The ingenious idea of Lemke was to introduce a new variable and consider the system:
\begin{equation} \label{eq:c} \MM \yy  + \ps -z \one  = \pq, \ \ \ \  \yy \geq 0, \ \ \ \ \ps \geq 0, \ \ \ \  z \geq 0  \ \
\ \ \mbox{and} \ \ \ \ y_is_i = 0,\ \forall i\in[d].  \end{equation}
The next lemma follows by construction of (\ref{eq:c}).
\begin{lemma}\label{lem:lemke1}
Given $(\MM,\qq)$, $(\yy,\ps,z)$ satisfies \eqref{eq:c} with $z=0$ iff $\yy$ satisfies~\eqref{eq:lcp}.
\end{lemma}
Let $\CPol$ be the polyhedron in $2d + 1$ dimensional space defined by the first four conditions of \eqref{eq:c}, i.e.,
\begin{equation}\label{eq:cp}
\CPol = \{ (\yy,\ps, z) \ |\ \MM \yy  + \ps -z \one  = \pq, \ \ \ \yy \geq 0, \ \ \ \ps \geq 0, \ \ \  z \geq 0\};
\end{equation}
for now, we will assume that $\CPol$ is {\em non-degenerate}.  

Since any solution to (\ref{eq:c}) must still satisfy $2d$ equalities in $\CPol$, the set of solutions, say
$S$, will be a subset of the one-skeleton of $\CPol$, i.e., it will consist of edges and vertices of $\CPol$.  Any solution to
the original system (\ref{eq:b}) must satisfy the additional condition $z = 0$ and hence will be a vertex of $\CPol$.

Now $S$ turns out to have some nice properties. Any point of $S$ is {\em fully labeled} in the sense that for each $i$, $y_i
= 0$ or $s_i = 0$.  We will say that a point of $S$ {\em has duplicate label i} if $y_i = 0$ and $s_i = 0$ are both satisfied
at this point. Clearly, such a point will be a vertex of $\CPol$ and it will have only one duplicate label.  Since there are
exactly two ways of relaxing this duplicate label, this vertex must have exactly two edges of $S$ incident at it.  Clearly, a
solution to the original system (i.e., satisfying $z = 0$) will be a vertex of $\CPol$ that does not have a duplicate label.  On
relaxing $z=0$, we get the unique edge of $S$ incident at this vertex.

As a result of these observations, we can conclude that $S$ consists of paths and cycles.  Of these paths, Lemke's algorithm
explores a special one.  An unbounded edge of $S$ such that the vertex of $\CPol$ it is incident on has $z > 0$ is called a
{\em ray}.  Among the rays, one is special -- the one on which $\yy = 0$. This is called the {\em primary ray} and the rest
are called {\em secondary rays}. Now Lemke's algorithm explores, via pivoting, the path starting with the primary ray. This
path must end either in a vertex satisfying $z = 0$, i.e., a solution to the original system, or a secondary ray. In the
latter case, the algorithm is unsuccessful in finding a solution to the original system; in particular, the original system
may not have a solution.  
We give the full pseudo-code for Lemke's algorithm in Table~\ref{tab:lemke}.



\begin{table}[!htb]
\caption{Lemke's Complementary Pivot Algorithm}\label{tab:lemke}
\begin{tabular}{|l|}
\hline
\hspace{5pt} {\bf If} $\qq\ge 0$ {\bf then} {\bf Return} $\yy\leftarrow \zeros$ \\
\hspace{5pt} $\yy\leftarrow 0, z\leftarrow |\min_{i \in [d]} q_i|, \ps=\qq+z\ones$\\
\hspace{5pt} $i\leftarrow $ duplicate label at vertex $(\yy,\ps,z)$ in $\CPol$. $flag\leftarrow 1$ \\
\hspace{5pt} {\bf While} $z>0$ {\bf do}\\
\hspace{10pt} {\bf If} $flag=1$ {\bf then} set $(\yy',\ps',z')\leftarrow $ vertex obtained by relaxing $y_i=0$ at $(\yy,\ps,z)$ in $\CPol$\\
\hspace{10pt} {\bf Else} set $(\yy',\ps',z')\leftarrow $ vertex obtained by relaxing $s_i=0$ at $(\yy,\ps,z)$ in $\CPol$\\
\hspace{10pt} {\bf If} $z>0$ {\bf then}\\
\hspace{15pt} $i \leftarrow $ duplicate label at $(\yy',\ps',z')$\\
\hspace{15pt} {\bf If} $v_i>0$ and $v'_i=0$ {\bf then} $flag\leftarrow 1$. {\bf Else} $flag\leftarrow 0$\\
\hspace{15pt} $(\yy,\ps,z)\leftarrow(\yy',\ps',z')$\\
\hspace{5pt} End {\bf While} \\
\hspace{5pt} {\bf Return} $\yy$\\
\hline
\end{tabular}
\end{table}

\newpage

\subsection{The reduction from \PLCP to \EOPL}
\label{sec:full_plcp_reduction}

In this section, we give a polynomial-time reduction from \PLCP to \EOPL.

It is well known that if matrix $M$ is a P-matrix (\PLCP), then $z$ strictly
decreases on the path traced by Lemke's algorithm \cite{cottle2009linear}.
Furthermore, by a result of Todd~\cite[Section 5]{todd1976orientation}, paths traced by
complementary pivot rule can be locally oriented.  Based on these two facts, 
we now derive a polynomial-time reduction from \PLCP to \EOPL.

Let $\CI=(M,\qq)$ be a given \PLCP instance, and let $\CL$ be the length of the 
bit representation of $M$ and $\qq$. 
We will reduce $\CI$ to an \EOPL instance $\CE$ in time $\poly(\CL)$. 
According to Definition~\ref{def:EOPL}, the instance $\CE$ is defined 
by its vertex set $\vert$, and procedures $S$ (successor), $P$ (predecessor) and $\pot$ (potential). 
Next we define each of these. 

As discussed in Section \ref{sec:lemke} the linear constraints of (\ref{eq:c})
on which Lemke's algorithm operates forms a polyhedron $\CPol$ given in
(\ref{eq:cp}). We assume that $\CPol$ is non-degenerate. This is without
loss of generality since, a typical way to ensure this is by perturbing $\qq$ so
that configurations of solution vertices remain unchanged
\cite{cottle2009linear}, and since $M$ is unchanged the LCP is still \PLCP. 

Lemke's algorithm traces a path on feasible points of (\ref{eq:c}) which is on
$1$-skeleton of $\CPol$ starting at $(\yy^0,\ps^0,z^0)$, where:
\begin{equation}\label{eq:v0}
\yy^0=0,\ \ \ \ \ z^0= |\min_{i \in [d]} q_i|,\ \ \ \ \  \ps^0=\qq+z\ones
\end{equation}
We want to capture
vertex solutions of (\ref{eq:c}) as vertices in \EOPL instance $\CE$. To
differentiate we will sometimes call the latter {\em configurations}. Vertex
solutions of (\ref{eq:c}) are exactly the vertices of polyhedron $\CPol$ with
either $y_i=0$ or $s_i=0$ for each $i\in [d]$. Vertices of (\ref{eq:c}) with
$z=0$ are our final solutions (Lemma \ref{lem:lemke1}). While each of its {\em
non-solution} vertex has a duplicate label. Thus, a vertex of this path can be
uniquely identified by which of $y_i=0$ and $s_i=0$ hold for each $i$ and its
duplicate label. This gives us a representation for vertices in the \EOPL
instance $\CE$. 

\medskip

\noindent{\bf \EOPL Instance $\CE$.}
\begin{itemize}
\item Vertex set $\vert=\{0,1\}^n$ where $n = 2d$. 
\item Procedures $S$ and $P$ as defined in Tables \ref{tab:S} and \ref{tab:P} respectively
\item Potential function $\pot:\vert \rightarrow \{0,1,\dots, 2^m-1\}$ defined in Table \ref{tab:F} for $m=\lceil ln(2\Delta^3)\rceil$, 
	  where $$\Delta=(n! \cdot I_{max}^{2d+1})+1$$ 
	  and $I_{max} = \max\{\max_{i,j\in [d]} M(i,j),\ \max_{i\in [d]} |q_i|\}$. 
\end{itemize}

For any vertex $\uu\in \vert$, the first $d$ bits of $\uu$ represent
which of the two inequalities, namely $y_i\ge 0$ and $s_i\ge 0$, are tight for
each $i \in [d]$. A valid setting of the second set of $d$ bits will have 
at most one non-zero bit -- if none is one then $z=0$, otherwise the location of one bit indicates the duplicate label. 
Thus, there are many invalid configurations, namely
those with more than one non-zero bit in the second set of $d$ bits. 
These are dummies that we will handle separately, and we define a procedure 
$\isvalid$ to identify non-dummy vertices in Table \ref{tab:iv}.
To go between ``valid'' vertices of $\CE$ and corresponding vertices of the Lemke polytope
$\CPol$ of LCP $\CI$, we define procedures $\eti$ and $\ite$ in Table
\ref{tab:ei}.

\begin{table}[!hbt]
\caption{Procedure \isvalid(\uu)}\label{tab:iv}
\begin{tabular}{|l|}
\hline
\hspace{5pt} {\bf If} $\uu=0^{n}$ {\bf then} {\bf Return} 1\\
\hspace{5pt} {\bf Else} let $\tau = (u_{(d+1)}+\dots+u_{2d})$\\
\hspace{15pt} {\bf If} $\tau> 1$ {\bf then} {\bf Return} 0\\
\hspace{15pt} Let $S\leftarrow \emptyset$. \% set of tight inequalities. \\
\hspace{15pt} {\bf If} $\tau = 0$ {\bf then} $S=S\cup \{ z=0\}$. \\
\hspace{15pt} {\bf Else}\\
\hspace{30pt} Set $l\leftarrow $ index of the non-zero coordinate in vector $(u_{(d+1)},\dots,u_{2d})$. \\
\hspace{30pt} Set $S=\{y_l=0, s_l=0\}$.\\
\hspace{15pt} {\bf For} each $i$ from $1$ to $d$ {\bf do} \\
\hspace{30pt} {\bf If} $u_i=0$ {\bf then} $S=S\cup \{y_i=0\}$, {\bf Else} $S=S\cup \{s_i=0\}$\\
\hspace{15pt} Let $A$ be a matrix formed by lhs of equalities $M\yy+\ps -\ones z=\qq$ and that of set $S$\\
\hspace{15pt} Let $\bb$ be the corresponding rhs, namely $\bb=[\qq; \zeros_{d\times 1}]$.\\
\hspace{15pt} Let $(\yy',\ps',z') \leftarrow \bb * A^{-1}$\\
\hspace{15pt} {\bf If} $(\yy',\ps',z') \in \CPol$ {\bf then} {\bf Return} 1, {\bf Else} {\bf Return} 0 \\
\hline
\end{tabular}
\end{table}

\begin{table}[!hbt]
\caption{Procedures $\ite(\uu)$ and $\eti(\yy,\ps,z)$}\label{tab:ei}
\begin{tabular}{|l|}
\hline
\begin{tabular}{l}
$\ite(\yy,\ps,z)$ \\ \hline
\hspace{5pt} {\bf If} $\exists i \in [d]$ s.t. $y_i * s_i \neq 0$ {\bf then} {\bf Return} $(\zeros_{(2d-2)\times 1};1;1)$ \% Invalid \\
\hspace{5pt} Set $\uu\leftarrow \zeros_{2d\times 1}$. Let $DL=\{i\in [d]\ |\ y_i=0\mbox{ and } s_i=0\}$.\\
\hspace{5pt} {\bf If} $|DL|>1$ {\bf then} {\bf Return} $(\zeros_{(2d-2)\times 1};1;1)$ \%In valid \\
\hspace{5pt} {\bf If} $|DL|=1$ {\bf then} for $i\in DL$, set $u_i\leftarrow 1$\\
\hspace{5pt} {\bf For} each $i\in [d]$ {\bf If} $s_i=0$ {\bf then} set $u_{d+i}\leftarrow 1$\\
\hspace{5pt} {\bf Return} $\uu$
\end{tabular}
\\ \hline
\begin{tabular}{l}
$\eti(\uu)$  \\ \hline
\hspace{5pt} {\bf If} $\uu=0^n$ {\bf then} {\bf Return} $(\zeros_{d \times 1}, \qq+z^0+1, z^0+1)$ \% This case will never happen\\
\hspace{5pt} {\bf If} \isvalid(\uu)=0 {\bf then} {\bf Return} $\zeros_{(2d+1) \times 1}$\\
\hspace{5pt} Let $\tau = (u_{(d+1)}+\dots+u_{2d})$\\
\hspace{5pt} Let $S\leftarrow \emptyset$. \% set of tight inequalities. \\
\hspace{5pt} {\bf If} $\tau = 0$ {\bf then} $S=S\cup \{ z=0\}$. \\
\hspace{5pt} {\bf Else}\\
\hspace{15pt} Set $l\leftarrow $ index of non-zero coordinate in vector $(u_{(d+1)},\dots,u_{2d})$. \\
\hspace{15pt} Set $S=\{y_l=0, s_l=0\}$.\\
\hspace{5pt} {\bf For} each $i$ from $1$ to $d$ {\bf do} \\
\hspace{15pt} {\bf If} $u_i=0$ {\bf then} $S=S\cup \{y_i=0\}$, {\bf Else} $S=S\cup \{s_i=0\}$\\
\hspace{5pt} Let $A$ be a matrix formed by lhs of equalities $M\yy+\ps -\ones z=\qq$ and that of set $S$\\
\hspace{5pt} Let $\bb$ be the corresponding rhs, namely $\bb=[\qq; \zeros_{d\times 1}]$.\\
\hspace{5pt} {\bf Return} $\bb * A^{-1}$\\
\end{tabular}\\
\hline
\end{tabular}
\end{table}

By construction of $\isvalid$, $\eti$ and $\ite$, the next lemma follows.

\begin{lemma}\label{lem:vert}
If $\isvalid(\uu)=1$ then $\uu=\ite(\eti(\uu))$, and the corresponding vertex $(\yy,\ps,z)\in \eti(\uu)$ of $\CPol$ is feasible in (\ref{eq:c}). If $(\yy,\ps,z)$ is a feasible vertex of (\ref{eq:c}) then $\uu=\ite(\yy,\ps,z)$ is a valid configuration, {\em i.e.,} $\isvalid(\uu)=1$.
\end{lemma}

\begin{proof}
The only thing that can go wrong is that the matrix $A$ generated in $\isvalid$ and $\eti$ procedures are singular, or the set of double labels $DL$ generated in $\ite$ has more than one elements. 
Each of these are possible only when more than $2d+1$ equalities of $\CPol$ hold at the corresponding point $(\yy,\ps,z)$, violating non-degeneracy assumption. 
\end{proof}

The main idea behind procedures $S$ and $P$, given in Tables \ref{tab:S} and
\ref{tab:P} respectively, is the following (also see Figure \ref{fig:lemke}):
Make dummy configurations in $\vert$ to point to themselves with cycles of
length one, so that they can never be solutions. 
The starting vertex $0^n \in \vert$ points to the configuration that corresponds
to the first vertex of the Lemke path, namely $\uu^0=\ite(\yy^0,\ps^0,z^0)$. 
Precisely, $S(0^n)=\uu^0$, $P(\uu^0)=0^n$ and $P(0^n)=0^n$ (start of
a path). 

For the remaining cases, let $\uu\in \vert$ have corresponding representation
$\xx=(\yy,\ps,z)\in \CPol$, and suppose $\xx$ has a duplicate label. As one
traverses a Lemke path for a P-LCPs, the value of $z$ monotonically decreases.
So, for $S(\uu)$ we compute the adjacent vertex $\xx'=(\yy',\ps',z')$ of $\xx$
on Lemke path such that the edge goes from $\xx$ to $\xx'$, and if the $z'<z$,
as expected, then we point $S(\uu)$ to configuration of $\xx'$ namely
$\ite(\xx')$. Otherwise, we let $S(\uu)=\uu$. Similarly, for $P(\uu)$, we find
$\xx'$ such that edge is from $\xx'$ to $\xx$, and then we let $P(\uu)$ be
$\ite(\xx')$ if $z'>z$ as expected, otherwise $P(\uu)=\uu$. 

For the case when $\xx$ does not have a duplicate label, then we have $z=0$. This is
handled separately since such a vertex has exactly one incident edge on the Lemke
path, namely the one obtained by relaxing $z=0$. According to the direction of 
this edge, we do similar process as before. For example, if the edge goes from 
$\xx$ to $\xx'$, then, if $z'<z$, we set $S(\uu)=\ite(\xx')$ else $S(\uu)=\uu$,
and we always set $P(\uu)=\uu$.  In case the edge goes from $\xx'$ to $\xx$, we
always set $S(\uu)=\uu$, and we set $P(\uu)$ depending on whether or not $z'>z$.

%
\medskip

The potential function $\pot$, formally defined in Table \ref{tab:F},
gives a value of zero to dummy vertices and the starting vertex $0^n$. To all
other vertices, essentially it is $((z^0-z) * \Delta^2)+1$. Since value of $z$
starts at $z^0$ and keeps decreasing on the Lemke path this value will keep
increasing starting from zero at the starting vertex $0^n$. Multiplication by
$\Delta^2$ will ensure that if $z_1>z_2$ then the corresponding potential values 
will differ by at least one. This is because, since $z_1$ and $z_2$ are 
coordinates of two vertices of polytope $\CPol$, their maximum value is $\Delta$
and their denominator is also bounded above by $\Delta$. Hence $z_1-z_2\le
1/\Delta^2$ (Lemma \ref{lem:pot}).  

To show correctness of the reduction we need to show two things: $(i)$ All the
procedures are well-defined and polynomial time. $(ii)$ We can construct a
solution of $\CI$ from a solution of $\CE$ in polynomial time. 

\begin{table}
\begin{minipage}{0.73\textwidth}
\caption{Successor Procedure $S(\uu)$}\label{tab:S}
\begin{tabular}{|l|}
\hline
\hspace{0pt}{\bf If} $\isvalid(\uu)=0$ {\bf then} {\bf Return} $\uu$\\
\hspace{0pt}{\bf If} $\uu=0^n$ {\bf then} {\bf Return} $\ite(\yy^0,\ps^0,z^0)$\\
\hspace{0pt}$\xx=(\yy,\ps,z) \leftarrow \eti(\uu)$\\
\hspace{0pt}{\bf If} $z=0$ {\bf then} \\
\hspace{5pt} $\xx^1\leftarrow$ vertex obtained by relaxing $z=0$ at $\xx$ in $\CPol$. \\
\hspace{5pt} {\bf If} Todd \cite{todd1976orientation} prescribes edge from $\xx$ to $\xx^1$ \\
\hspace{10pt} {\bf then} set $\xx'\leftarrow \xx^1$. {\bf Else Return} $\uu$ \\
\hspace{0pt}{\bf Else} set $l\leftarrow $ duplicate label at $\xx$\\
\hspace{5pt} $\xx^1\leftarrow $ vertex obtained by relaxing $y_l=0$ at $\xx$ in $\CPol$ \\
\hspace{5pt} $\xx^2\leftarrow $ vertex obtained by relaxing $s_l=0$ at $\xx$ in $\CPol$ \\
\hspace{5pt} {\bf If} Todd \cite{todd1976orientation} prescribes edge from $\xx$ to $\xx^1$ \\
\hspace{10pt} {\bf then} $\xx'=\xx^1$ \\
\hspace{5pt} {\bf Else} $\xx'=\xx^2$\\
\hspace{0pt}Let $\xx'$ be $(\yy',\ps',z')$. \\
\hspace{0pt}{\bf If} $z>z'$ {\bf then} {\bf Return} $\ite(\xx')$. {\bf Else} {\bf Return} $\uu$.\\
\hline
\end{tabular}
\end{minipage}%
\hspace{-1cm}
\begin{minipage}{0.23\textwidth}
\caption{Potential Value $\pot(\uu)$}\label{tab:F}
\begin{tabular}{|l|}
\hline
\hspace{0pt} {\bf If} $\isvalid(\uu)=0$ \\
\hspace{5pt} {\bf then} {\bf Return} $0$\\
\hspace{0pt} {\bf If} $\uu=0^n$\\
\hspace{5pt}  {\bf then} {\bf Return} $0$\\
\hspace{0pt} $(\yy,\ps,z) \leftarrow \eti(\uu)$\\
\hspace{0pt} {\bf Return} $\lfloor \Delta^2*(\Delta -z)\rfloor$\\
\hline
\end{tabular}
\end{minipage}
\end{table}

\begin{table}[!htb]
\caption{Predecessor Procedure $P(\uu)$}\label{tab:P}
\begin{tabular}{|l|}
\hline
\hspace{0pt} {\bf If} $\isvalid(\uu)=0$ {\bf then} {\bf Return} $\uu$\\
\hspace{0pt} {\bf If} $\uu=0^n$ {\bf then} {\bf Return} $\uu$\\
\hspace{0pt} $(\yy,\ps,z) \leftarrow \eti(\uu)$\\
\hspace{0pt} {\bf If} $(\yy,\ps,z)=(\yy^0,\ps^0,z^0)$ {\bf then} {\bf Return} $0^n$\\
\hspace{0pt} {\bf If} $z=0$ {\bf then} \\
\hspace{5pt} $\xx^1\leftarrow$ vertex obtained by relaxing $z=0$ at $\xx$ in $\CPol$. \\
\hspace{5pt} {\bf If} Todd \cite{todd1976orientation} prescribes edge from $\xx^1$ to $\xx$ {\bf then} set $\xx'\leftarrow \xx^1$. {\bf Else Return} $\uu$\\
\hspace{0pt} {\bf Else}\\
\hspace{5pt} $l\leftarrow $ duplicate label at $\xx$\\
\hspace{5pt} $\xx^1\leftarrow $ vertex obtained by relaxing $y_l=0$ at $\xx$ in $\CPol$ \\
\hspace{5pt} $\xx^2\leftarrow $ vertex obtained by relaxing $s_l=0$ at $\xx$ in $\CPol$ \\
\hspace{5pt} {\bf If} Todd \cite{todd1976orientation} prescribes edge from $\xx^1$ to $\xx$ {\bf then} $\xx'=\xx^1$ {\bf Else} $\xx'=\xx^2$\\
\hspace{0pt} Let $\xx'$ be $(\yy',\ps',z')$. {\bf If} $z<z'$ {\bf then} {\bf Return} $\ite(\xx')$. {\bf Else} {\bf Return} $\uu$.\\
\hline
\end{tabular}
\end{table}

\begin{lemma}\label{lem:PSF}
Functions $P$, $S$ and $\pot$ of instance $\CE$ are well defined, making $\CE$ a valid \EOPL instance. 
\end{lemma}

\begin{proof}
Since all three procedures are polynomial-time in $\CL$, they can be defined
by $\poly(\CL)$-sized Boolean circuits. Furthermore, for any $\uu \in \vert$,
we have that $S(\uu),P(\uu) \in \vert$. For~$\pot$, 
since the value of $z \in [0,\ \Delta-1]$, we
have $0\le \Delta^2(\Delta-z)\le \Delta^3$. Therefore, $\pot(\uu)$ is an
integer that is at most $2 \cdot \Delta^3$ and hence is in set $\{0,\dots, 2^m-1\}$. 
\end{proof}

There are two possible types of solutions of an \EOPL instance. One indicates
the beginning or end of a line, and the other is a vertex with locally optimal
potential (that does not point to itself). 
First we show that the latter case never arise. For this, we need the
next lemma, which shows that potential differences in two adjacent
configurations adheres to differences in the value of $z$ at corresponding
vertices.

\begin{lemma}\label{lem:pot}
Let $\uu \neq \uu'$ be two valid configurations, i.e.,
	$\isvalid(\uu)=\isvalid(\uu')=1$, and let $(\yy,\ps,z)$ and $(\yy',\ps',z')$
	be the corresponding vertices in $\CPol$. Then the following holds: $(i)$
	$\pot(\uu)=\pot(\uu')$ iff $z=z'$. $(ii)$ $\pot(\uu)>\pot(\uu')$ iff $z<z'$.
\end{lemma}

\begin{proof}
Among the valid configurations all except $\zeros$ has positive $\pot$ value. Therefore, wlog let $\uu,\uu'\neq \zeros$. For these we have $\pot(\uu)=\lfloor \Delta^2*(\Delta -z)\rfloor$, and $\pot(\uu')=\lfloor \Delta^2*(\Delta -z')\rfloor$. 

Note that since both $z$ and $z'$ are coordinates of vertices of $\CPol$, whose description has highest coefficient of $\max\{\max_{i,j\in [d]} M(i,j),\max_{i\in [d]} |q_i|\}$, and therefore their numerator and denominator both are bounded above by $\Delta$. Therefore, if $z< z'$ then we have 
\[
z'-z\ge \frac{1}{\Delta^2} \Rightarrow ((\Delta-z) - (\Delta - z'))*\Delta^2 \ge 1 \Rightarrow \pot(\uu)-\pot(\uu') \ge 1.
\]

For $(i)$, if $z=z'$ then clearly $\pot(\uu)=\pot(\uu')$, and from the above argument it also follows that if $\pot(\uu)= \pot(\uu')$ then it can not be the case that $z\neq z'$. Similarly for $(ii)$, if $\pot(\uu)>\pot(\uu')$ then clearly, $z'>z$, and from the above argument it follows that if $z'>z$ then it can not be the case that $\pot(\uu')\ge \pot(\uu)$. 
\end{proof}

Using the above lemma, we will next show that instance $\CE$ has no local maximizer. 

\begin{lemma}\label{lem:t}
Let $\uu,\vv \in \vert$ s.t. $\uu\neq \vv$, $\vv=S(\uu)$, and $\uu=P(\vv)$. Then $\pot(\uu)< \pot(\vv)$.
\end{lemma}
\begin{proof}
Let $\xx=(\yy,\ps,z)$ and $\xx'=(\yy',\ps',z')$ be the vertices in polyhedron $\CPol$ corresponding to $\uu$ and $\vv$ respectively. From the construction of $\vv=S(\uu)$ implies that $z'<z$. Therefore, using Lemma \ref{lem:pot} it follows that $\pot(\vv)<\pot(\uu)$.
\end{proof}

Due to Lemma \ref{lem:t} the only type of solutions available in $\CE$ is where $S(P(\uu))\neq \uu$ and $P(S(\uu))\neq \uu$. Next two lemmas shows how to construct solutions of $\CI$ from these. 

\begin{lemma}\label{lem:t1}
Let $\uu \in \vert$, $\uu \neq 0^n$. 
If $P(S(\uu))\neq \uu$ or $S(P(\uu))\neq \uu$, then $\isvalid(\uu)=1$, and for $(\yy,\ps,z)=\eti(\uu)$ if $z=0$ then $\yy$ is a $\PLo$ type solution of \PLCP instance $\CI=(M,\qq)$. 
\end{lemma}
\begin{proof}
By construction, if $\isvalid(\uu) = 0$, then $S(P(\uu))=\uu$ and $P(S(\uu))=\uu$, therefore $\isvalid(\uu)=0$ when $\uu$ has a predecessor or successor different from $\uu$.
Given this, from Lemma \ref{lem:vert} we know that $(\yy,\ps,z)$ is a feasible vertex in (\ref{eq:c}). Therefore, if $z=0$ then using Lemma \ref{lem:lemke1} we have a solution of the LCP (\ref{eq:lcp}), {\em i.e.,} a type $\PLo$ solution of our \PLCP instance $\CI=(\MM,\qq)$.
%
%
\end{proof}

\begin{lemma}\label{lem:t2}
Let $\uu \in \vert$, $\uu \neq 0^n$ such that $P(S(\uu))\neq \uu$ or $S(P(\uu))\neq \uu$, and let $\xx=(\yy,\ps,z)=\eti(\uu)$. 
If $z\neq 0$ then $\xx$ has a duplicate label, say $l$. And for directions $\sigma_1$ and $\sigma_2$ obtained by relaxing $y_l=0$ and $s_l=0$ respectively at $\xx$, we have $\sigma_1(z)*\sigma_2(z)\ge 0$, where $\sigma_i(z)$ is the coordinate corresponding to $z$. 
\end{lemma}
\begin{proof}
From Lemma \ref{lem:t1} we know that $\isvalid(\uu)=1$, and therefore from Lemma \ref{lem:vert}, $\xx$ is a feasible vertex in (\ref{eq:c}).
From the last line of Tables \ref{tab:S} and \ref{tab:P} observe that $S(\uu)$ points to the configuration of vertex next to $\xx$ on Lemke's path only if it has lower $z$ value otherwise it gives back $\uu$, and similarly $P(\uu)$ points to the previous only if value of $z$ increases.


First consider the case when $P(S(\uu))\neq \uu$. Let $\vv=S(\uu)$ and corresponding vertex in $\CPol$ be $(\yy',\ps',z')=\eti(\vv)$. 
If $\vv\neq \uu$, then from the above observation we know that $z'>z$, and in that
case again by construction of $P$ we will have $P(\vv)=\uu$, contradicting
$P(S(\uu))\neq \uu$. Therefore, it must be the case that $\vv=\uu$.
Since $z\neq 0$ this happens only when the next vertex on Lemke path after $\xx$ has
higher value of $z$ (by above observation). As a consequence of $\vv=\uu$, we also have $P(\uu)\neq \uu$. By construction of $P$ this implies for 
$(\yy'',\ps'',z'')=\eti(P(\uu))$, $z''>z$. Putting both together we get 
increase in $z$ when we relax $y_l=0$ as well as when we relax $s_l=0$ at
$\xx$.

For the second case $S(P(\uu))\neq \uu$ similar argument gives that value of $z$ decreases when we relax $y_l=0$ as well as when we relax $s_l=0$ at
$\xx$. The proof follows.
\end{proof}

Finally, we are ready to prove our main result of this section using Lemmas
\ref{lem:t}, \ref{lem:t1} and \ref{lem:t2}. Together with Lemma \ref{lem:t2},
we will use the fact that on Lemke path $z$ monotonically decreases if $M$ is a
P-matrix or else we get a witness that $M$ is not a
P-matrix~\cite{cottle2009linear}. 

\begin{theorem}
\PLCP reduces to \EOPL in polynomial-time. 
\end{theorem}
\begin{proof}
	Given an instance of $\CI=(\MM,\qq)$ of \PLCP, where $M\in \Real^{d\times d}$ and $\qq\in \Real^{d\times 1}$ reduce it to an instance $\CE$ of \EOPL as described above with vertex set $\vert=\{0,1\}^{2d}$ and procedures $S$, $P$ and $\pot$ as given in Table \ref{tab:S}, \ref{tab:P}, and \ref{tab:F} respectively.

Among solutions of \EOPL instance $\CE$, there is no local potential maximizer,
	i.e., $\uu\neq \vv$ such that $\vv=S(\uu)$, $\uu=P(\vv)$ and $\pot(\uu)>\pot(\vv)$
	due to Lemma \ref{lem:t}. We get a solution $\uu \neq 0$ such that either
	$S(P(\uu))\neq \uu$ or $P(S(\uu))\neq \uu$, then by Lemma \ref{lem:t1} it is
	valid configuration and has a corresponding vertex $\xx=(\yy,\ps,z)$ in
	$\CPol$. Again by Lemma~\ref{lem:t1} if $z=0$ then $\yy$ is a $\PLo$ type solution
	of our \PLCP instance $\CI$. On the other hand, if $z>0$ then from Lemma
	\ref{lem:t2} we get that on both the two adjacent edges to $\xx$ on Lemke
	path the value of $z$ either increases or deceases. This gives us a minor of $M$
	which is non-positive~\cite{cottle2009linear}, 
	i.e., a \PLt type solution of the \PLCP instance~\CI.
\end{proof}

\section{The Full Reduction and Proofs for Section~\ref{sec:lcm2eopl}}

\subsection{Slice Restrictions of Contraction Maps}
\label{sec:slice}

Before we begin the reduction, we first fix some notation.
Our algorithm and reduction for contraction maps will make heavy use of the
concept of a \emph{slice restriction} of a contraction map, which we describe
here. First, for any $d\in \N$, we define the set of \emph{slices} $\Slice_d =
\Paren{[0,1]\cup \Set{\blank}}^d$ to be vectors of length $d$ each component of
which is either a number in $[0,1]$ or the special symbol $\blank$ which
indicates that corresponding component is free to vary. With each slice $\ss\in
\Slice_d$ we associate a hyperplane $H(\ss) = \Setbar{x\in \R^d}{x_i = s_i\
\text{for}\ s_i \neq \blank}$. We define the set of fixed coordinates of a slice
$\ss\in \Slice_d$, $\fixed(\ss) = \Setbar{i\in [d]}{s_i \neq \blank}$ and the set
of free coordinates, $\free(\ss) = [d] \setminus \fixed(s)$. A slice for which $\Card{\free(\ss)} = i$ will be called an \emph{$i$-slice}.

We can define the \emph{slice restriction} of a function $f : [0,1]^d \to [0,1]^d$ with respect to a slice $\ss\in \Slice_d$, denoted $\restr{f}{\ss}$, to be the function obtained by fixing the coordinates $\fixed(\ss)$ according to $\ss$, and keeping the coordinates of $\free(\ss)$ as arguments. To simplify usage of $\restr{f}{\ss}$ we'll formally treat $\restr{f}{\ss}$ as a function with $d$ arguments, where the coordinates in $\fixed(\ss)$ are ignored. Thus, we define $\restr{f}{\ss}:[0,1]^d\to [0,1]^d$ by
\[ \restr{f}{\ss}(x) = f(y) \quad\text{where}\ y_i = \begin{cases} s_i &\ \text{if $i \in \fixed(\ss)$}\\ x_i&\ \text{if $i \in \free(\ss)$.}\end{cases} \]

Let $\free(\ss) = \Set{i_1,\dotsc, i_k}$. We'll also introduce a variant of $\restr{f}{\ss}$ when we want to consider the slice restriction as a lower-dimensional function, $\Restr{f}{\ss} : [0,1]^d \to [0,1]^{\Card{\free(\ss)}}$ defined by
\[ \Restr{f}{\ss}(x) = \Paren{\Paren{\restr{f}{\ss}(x)}_{i_1}, \dotsc, \Paren{\restr{f}{\ss}(x)}_{i_k}}\text{.} \]

We can also define slice restrictions for vectors in the natural way:
\[ \Paren{\restr{x}{\ss}}_i = \begin{cases} s_i&\ \text{if $s_i \neq \blank$}\\ x_i&\ \text{otherwise.} \end{cases}\]
Finally, we'll use $\Restr{x}{\ss}$ to denote projection of $x$ onto the coordinates in $\free(\ss)$:
\[ \Restr{x}{\ss} = \Paren{x_{i_1},\dotsc, x_{i_k}}\text{.} \]

We now extend the definition of a contraction map to a slice restriction of a function in the obvious way. We say that $\restr{\tf}{\ss}$ is a contraction map with respect to a norm $\Norm{\cdot}$ with Lipschitz constant $c$ if for any $x,y\in [0,1]^d$ we have
\[ \Norm{\Restr{f}{\ss}(x) - \Restr{f}{\ss}(y)} \leq c \Norm{\Restr{x}{\ss} - \Restr{y}{\ss}}\text{.} \]

Slice restrictions will prove immensely useful through the following observations:

\begin{lemma}\label{lem:cm1}
Let $f:[0,1]^d\to [0,1]^d$ be a contraction map with respect to $\Norm{\cdot}_p$ with Lipschitz constant $c \in (0,1)$. Then for any slice $\ss \in \Slice_d$, $\restr{f}{\ss}$ is also a contraction map with respect to $\Norm{\cdot}_p$ with Lipschitz constant $c$.
\end{lemma}
\begin{proof}
For any two vectors $x,y \in [0,1]^{d}$ we have
\begin{align*}
  \Norm{\restr{\tf}{\ss}(x) - \restr{\tf}{\ss}(y)}_p &\leq \Norm{f(\restr{x}{\ss}) - f(\restr{y}{\ss})}_p\\
                                       &< c \Norm{\restr{x}{\ss} - \restr{y}{\ss}}_p\\
                                       &= c \Norm{\Restr{x}{\ss} - \Restr{y}{\ss}}_p
\end{align*}
\end{proof}

Since slice restrictions of contraction maps are themselves contraction maps in the sense defined above, they have unique fixpoints, up to the coordinates of the argument which are fixed by the slice and thus ignored. We'll nevertheless refer to the \emph{unique fixpoint of a slice restriction of a contraction map}, which is the unique point $x\in [0,1]^d$ such that
\[ \Restr{f}{\ss}(x) = \Restr{x}{\ss}\quad\text{and}\quad x = \restr{x}{\ss}\text{.} \]

\begin{lemma}\label{lem:cm2}
Let $f : [0,1]^d \to [0,1]^d$ be a contraction map with respect to $\Norm{\cdot}_p$ with Lipschitz constant $c\in (0,1)$. Let $\ss, \ss' \in \Slice_d$ be such that $\fixed(\ss') = \fixed(\ss) \cup \Set{i}$ and $s_j = s'_j$ for all $j \in \fixed(\ss)$. Let $x, y \in [0,1]^d$ be the unique fixpoints of $\Restr{f}{\ss}$ and $\Restr{f}{\ss'}$, respectively. Then $(x_i - y_i)(f(y)_i - y_i) > 0$. 
\end{lemma} 
\begin{proof}
We'll prove this by contradiction. Without loss of generality, assume towards a contradiction that $x_i \leq y_i$ and that $f(y)_i > y_i$. Then we have
\begin{align*}
  \Norm{f(y) - f(x)}_p^p
  &= \Norm{\Paren{f(y)_1,\dotsc, f(y)_d} - \Paren{f(x)_1, \dotsc, f(x)_d}}_p^p \\
  &= \sum_{j\in \fixed(\ss)} \Abs{s_j - s_j}^p + \sum_{j\in \free(\ss')} \Abs{y_j - x_j}^p + \Abs{f(y)_i - x_i}\\
  &> \sum_{i\in \fixed(\ss)} \Abs{s_j - s_j}^p + \sum_{j\in \free(\ss')} \Abs{y_j - x_j}^p + \Abs{y_i - x_i}\\
  &= \Norm{y - x}_p^p
\end{align*} which contradicts the fact that $f$ is a contraction map. The lemma follows.
\end{proof}

\subsection{Discretizing the problem.}
\label{sec:discretizing}

We will turn \LCM into a discrete problem by overlaying a grid of points on
the $[0, 1]^d$ cube. For each integer $k$, let $I_k = \{0, 1/k, 2/k, \dots,
k/k\}$ be the discretization of the interval $[0, 1]$ into points of the form
$x/k$. Given a tuple $(k_1, k_2, \dots, k_d)$ of values, where $k_i$ specifies
the desired grid width for dimension $k$, we define the set of points $P =
I_{k_1} \times I_{k_2} \times \dots \times I_{k_d}$. Observe that every element
of $P$ is a point $p \in [0,1]^d$ where $p_i$ is a rational with denominator
$k_i$.

We will frequently refer to subsets of $P$ in which some of the dimensions have
been fixed, and so we specialize the set $\Slice_d$ that was defined in
Section~\ref{sec:slice} for this task. Throughout this section we take
$\DSlice_d \subset \Slice_d$ to
be the subset of possible slices that align with the grid defined by $P$. More
formally, an element $s \in \Slice_d$ is in $\DSlice_d$ if and only if whenever $s_i \ne
\blank$, we have that $s_i = x/k_i$.
For every slice $s \in \DSlice_d$, we define the set of points $P_s \subseteq P$, to be
the subset of points that lie on the slice~$s$.


%

Having discretized the space, we now also discretize the function $f$.
A \emph{direction function} for a dimension $i \in \{1, 2, \dots, n\}$ is a
function $D :  P \rightarrow \{\up, \down, \zero\}$ that satisfies the following
property.
Let $p^1$ and $p^2$ be two points in $P$ that differ only in dimension $i$, and
suppose that $p^1_i > p^2_i$. Both of the following conditions must hold.
\begin{itemize}
\item If $D(p^1) \in \{\up, \zero\}$  then $D(p^2) = \up$.
\item If $D(p^2) \in \{\down, \zero\}$  then $D(p^1) = \down$.
\end{itemize}

Figure~\ref{fig:direction} illustrates two direction functions for a
two-dimensional problem. The figure on the left shows a direction function for
the up-down dimension, while the figure on the right shows a direction function
for the left-right dimension. Each square in the figures represent a point in
the discretized space, and the value of the direction function is shown inside
the box. 


\smallskip

\noindent \textbf{The discrete contraction problem.}
Suppose that we have a set $\mathcal{D}$ of direction functions, such that for
each dimension $i$, there a function $D_i \in \mathcal{D}$ that is a direction
function for dimension $i$. We say that a point $p \in P$ is a \emph{fixpoint} of 
$\mathcal{D}$ if $D_i(p) = \zero$ for all $i$. Furthermore, for each
slice $s \in \DSlice_d$, we say that a point $p \in P_s$ is a \emph{fixpoint of $s$}
if $D_i(p) = \zero$ for all $i$ for which $s_i = \blank$.

We say that a slice $s \in \DSlice_d$ is an \emph{$i$-slice} if, for every $j \le i$ we
have that $s_j = \blank$, and for every $j > i$ we have $s_j \ne \blank$.

\begin{definition}[Discrete Contraction Map]
We say that $\mathcal{D}$ is a discrete contraction map if, for every
$i$-slice $s$, the following conditions hold.
\begin{enumerate}
\itemsep1mm
\item There is a unique fixpoint of $s$.

\item Let $s' \in \DSlice_d$ be a sub-slice of $s$ where some coordinate $i$ for which
$s_i = \blank$ has been fixed
to a value, and all other coordinates are unchanged. If $q$ is the unique fixpoint of $s$, and $p$ is the unique
fixpoint of $s'$, then
\begin{itemize}
\itemsep1mm
\item if $p_i < q_i$, then $D_i(p) = \up$, and
\item if $p_i > q_i$, then $D_i(p) = \down$.
\end{itemize}
\end{enumerate}
\end{definition}
The first condition ensures that every slice has a unique fixpoint, and is
the discrete analogue of the property proved in Lemma~\ref{lem:cm1}. Note that taking $s = (\blank, \blank, \dots, \blank)$ in the above implies that
$\mathcal{D}$ has a unique fixpoint.
The second condition is a technical
condition that we will use in our reduction, and is the discrete analogue of the
property proved in Lemma~\ref{lem:cm2}.

The discrete contraction problem is to find the fixpoint of $\mathcal{D}$.

\begin{definition}[\DCM]
The input to the problem is a set $\mathcal{D}$ of direction functions, each of which is defined by a
boolean circuit. It is promised that each function in $\mathcal{D}$ is a
direction function, and that $\mathcal{D}$ satisfies the properties of a
discrete contraction map. The task is to find the point $p \in P$ that is the
fixpoint of $\mathcal{D}$.
\end{definition}

\smallskip

\noindent \textbf{From \LCM to \DCM.}
We can show that \LCM can be reduced, in polynomial time, to \DCM. To do so, we
take an instance of contraction, defined by the function $f$, and produce a set
$\mathcal{D}_f$. 
For each
dimension $i \le n$ we define the function $D_i \in \mathcal{D}_f$ so that, for
each point $p \in P$, we have:
\begin{itemize}
\item if $f(p)_i > p$ then $D_i(p) = \up$,
\item if $f(p)_i < p$ then $D_i(p) = \down$, and
\item if $f(p)_i = p$ then $D_i(p) = \zero$.
\end{itemize}
In other words, the function $D_i$ simply checks whether $f(p)$ moves up, down,
or not at all in dimension $i$. 
In the rest of this section, we will provide a proof for the following lemma.

\begin{lemma}
\label{lem:lcm2dcm}
If $f$ is a contraction map, then $\mathcal{D}_f$ is a discrete contraction map.
\end{lemma}

Intuitively, the two required properties are consequences of
Lemmas~\ref{lem:cm1} and~\ref{lem:cm2}. The main difficulty of the proof is
finding suitable values for the tuple $(k_1, k_2, \dots, k_d)$ so that every
slice actually has a unique fixpoint in the set $P$. We prove that such a
tuple exists, and we rely on the fact that $f$ is defined by a \LinearFIXP
circuit, and so has rational fixpoints of polynomial bit-length. 

\smallskip

\noindent \textbf{\LinearFIXP circuits and LCPs.}
The goal of this proof is to find a tuple $(k_1, k_2, \dots, k_d)$ such that the
lemma holds. 
We utilize a lemma, which was shown in~\cite{mehta}, which proves that that
every \LinearFIXP circuit can be transformed, in polynomial time, into an LCP,
such that the solution of the LCP captures the fixpoint of the circuit, and
gives a bound on the bit-length of the numbers used in the LCP. 

Throughout this proof, we will use the function $b(x)$ to denote, for each
rational $x$, the
bit-length of the representation of $x$, which is the bit-length needed to
represent the numerator and denominator of $x$. When we apply $b$ to a matrix or
a vector, we take the maximum of the bit-lengths of the elements of that matrix
or vector. 

\begin{lemma}[\cite{mehta}]
\label{lem:ruta}
Let $C$ be a \LinearFIXP circuit, and let $n$ be the number of $\max$ and $\min$ gates
used in $C$. We can produce, in polynomial time, an LCP
defined by an $n \times n$ matrix $M_C$ and $n$-dimensional vector $\qq_C$ every fixpoint of $C$ is a
solution of the LCP defined by $M$ and $\bf{q}$ and vice versa.
Furthermore, if $\size(C)$ denotes 
$$(\#\mbox{inputs } + \#\mbox{gates } + \mbox{ number of bits needed to represent the constants used}),$$ then 
$b(M_C)$ and $b(\qq_C)$ are both at most $O(n\times\size(C))$.
\end{lemma}

Crucially, if $C'$ denotes a circuit where one of the inputs of $C$ is fixed to
be some number $x$, then $b(M_{C'}) \le b(M_{C})$ and $b(\qq_{C'}) \le b(x) +
b(M_{C}) + b(\qq_{C})$. In other words, the bit-length of $M_{C'}$ does not
depend on $x$, and is in fact at most the bit-length of $M_{C}$.


\smallskip

\noindent \textbf{Bounding the bit-length of a solution of an LCP.}
We now prove two technical lemmas about the bit-length of any solution to the
LCP. We begin with the following lemma regarding the bit-length of a matrix
inverse.

\begin{lemma}
\label{lem:Ainverse}
Let $A \in \Rational^{n\times n}$ be an square matrix of full rank, and let the largest absolute value of an entry of $A$ be denoted by $B$,
i.e., $|A_{ij}| \le B$ for all $i,j$.
Let $p/q$ for $p,q \in \Integer$ denote an arbitrary entry of $A^{-1}$. Then we have $\max(p,q) \le B^n n^{n/2}$.
\end{lemma}
\begin{proof}
We have that $A^{-1} = \frac{1}{\det(A)}(\text{adjoint}(A))$. Entries of $\text{adjoint}(A)$ are $\pm$ times the determinant
of a square sub-matrix of $A$ of size $n-1$, which this also has a bound of $B$ on the absolute value of any entry.
It is a well-known corollary of Hadamard's inequality that $|\det(A)| \le B^n n^{n/2}$. Applying this bound to entries of 
$\text{adjoint}(A)$ for $p$ and $\det(A)$ for $q$ gives the result.
\end{proof}


We now use this to prove the following bound on the bit-length of a fixpoint
of an LCP.

\begin{lemma}
\label{lem:lcpsize}
Let $M$ and $\qq$ be an LCP, and $\yy$ be the solution to the LCP. We have that
$b(\yy) \le n \cdot \log n + 3n \cdot b(M) + b(\qq)$.
\end{lemma}
\begin{proof}



We first note that if an LCP has a solution, then it has a vertex solution.
Let $\yy$ be a solution of the LCP defined by $M=[M_1,\ldots,M_n]$ and $\qq$, and 
let $\ww = \qq - M \cdot \yy$ so that $I\cdot\ww + M\cdot\yy = \qq$ where 
$I$ is the identity matrix $[e_1,\ldots,e_n]$.
By definition, for each $i$, we either have
that $\yy_i = 0$ or that $\ww_i = (\qq - M \cdot \yy)_i = 0$. 
Let $\alpha = \{i \ |\ \yy_i = 0 \}$. 
We define the matrix $A = [A_1,\ldots,A_n]$ from unit vectors and the columns of $M$ as follows:
$$
A_i = \begin{cases}
	e_i,& \textbf{ if } i \in \alpha,\\
	M_i,& \textbf{ if } i \notin \alpha.
\end{cases}
$$ 
If the vertex is non-degenerate, then whenever $\yy_i=0$ we have $\ww_i>0$, and
therefore $A$ is guaranteed to be invertible. In case of a degenerate vertex,
note that we take $e_i$ for the $i$th column in $A$ even when $\yy_i$ and
$\ww_i$ both are zero.
This will ensure that $A$ is invertible.

Then, the LCP conditions, i.e., the complementary condition, the equation
$I\cdot\ww + M\cdot\yy = \qq$, and the non-negativity of $\yy$ and $\ww$, are 
equivalent to the existence of a non-negative vector $\xx$ that satisfies
$$A \cdot \xx = \qq\ .$$ 
Then $\yy_i = \xx_i$ for $i \notin \alpha$ and $\yy_i = 0$ for $i \in \alpha$,
so we have $b(\yy) \le b(\xx)$.
Also note that we have $b(A) \le b(M)$, since the entries in columns that take the
value of $e_i$ have constant bit-length.

We must transform $A$ into an integer matrix in order to apply
Lemma~\ref{lem:Ainverse}. Let $l$ denote the least common multiple of the
entries in $A$. Note that $l \le n^2 \cdot 2^{b(A)}$ and hence $b(l) \le b(A) +
2\log(n)$. Our matrix equation above can be rewritten as $l \cdot A \cdot \xx =
l \cdot q$, and note that $(l \cdot A)$ is an integer matrix. Hence we have $\xx
= (l \cdot A)^{-1} \cdot l \cdot \qq$.

If $B$ denotes the maximum entry of $(l \cdot A)$, then $B \le l \cdot 2^{b(A)}$.
So by Lemma~\ref{lem:Ainverse} the maximum entry of 
$(l \cdot A)^{-1}$ is 
\begin{equation*}
(l \cdot 2^{b(A)})^n \cdot n^{n/2} \le 
(n^2 \cdot 2^{2 b(A)})^n \cdot n^{n/2} 
\end{equation*}
Each entry of $\yy$ consists of the sum of $n$ entries of 
$(l \cdot A)^{-1}$ each of which is multiplied by an entry of
$\qq$, and finally the sum is multiplied by $l$. So we get the following bound
on the bit-length of $\qq$.
\begin{align*}
b(\yy) &\le \log \left( n \cdot 2^{b((l \cdot A)^{-1})} \cdot 
2^{b(\qq)} \cdot 2^{b(l)} \right) \\
& \le \log \left( n \cdot \left( (n^2 \cdot 2^{2 b(A)})^n \cdot n^{n/2} \right) \cdot 2^{b(\qq)} \cdot \left(n^2 \cdot 2^{b(A)} \right) \right) \\
& = \log \left( n^{n/2 + 5} \cdot 2^{(2n + 1) b(A)} \cdot 2^{b(\qq)} \cdot  \right) \\
& \le  n \cdot \log n + 3n \cdot b(A) + b(\qq).
\end{align*}





\end{proof}

\smallskip

\noindent \textbf{Fixing the grid size.} 
We shall fix the grid size iteratively, starting with $k_d$, and working
downwards. 
Note that there is exactly 
one $d$-slice, which is $s = (\blank, \blank, \dots, \blank)$.
We fix $k_d = 2^{n \cdot \log n + 3n \cdot b(M_C) + b(\qq_C)}$, and recall that the set
$I_{k_d}$ contains all rationals with denominators at most $k_d$.
Lemma~\ref{lem:lcpsize} implies that the fixpoint of $s$, which is a solution
of the LCP defined by $M_C$ and $\qq_C$, is a rational with denominator at most 
$k_d$.

Now let $s = (\blank, \blank, \dots, \blank, x)$ be a $(d-1)$-slice where $x$ is
a member of $I_{k_d}$, meaning that $x$ is a point with denominator $k_d$.
Recall that if we fix the $d$th input of $C$ to be $x$,
then we obtain a smaller circuit $C'$, and we have that $b(M_{C'}) \le b(M_C)$.
However, $x$ will be added to some elements of $q$, and so $b(\qq_{C'}) \le
b(k_d)$. If $\exp(x) = 2^x$, then we fix 
\begin{align*}
k_{d-1} &= \exp \left( n \cdot \log n + 3n \cdot b(M_{C'}) + b(\qq_{C'}) \right)\\
&\le \exp \left( 2 n \cdot \log n + 6n \cdot b(M_{C}) + b(\qq_{C}) \right).
\end{align*}
Again, Lemma~\ref{lem:lcpsize} implies that the fixpoint of $s$ is a rational
with denominator at most $k_{d-1}$.

Repeating the above argument for all $i \le d-1$ leads us to set
\begin{equation*}
k_i = \exp \left( (d-i+1) n \cdot \log n + 3 \cdot (d-i+1) \cdot n \cdot
b(M_{C}) + b(\qq_{C}) \right). 
\end{equation*}
By the same reasoning, we get that for each $i$-slice $s$, the unique fixpoint 
of $s$ is a rational with denominator at most $k_i$.

Observe that the maximum bit-length of the numbers $k_i$ is attained by the
number $k_1$, which has bit-length at most
\begin{equation*}
d \cdot n \cdot \log n + 3 d \cdot n \cdot b(M_{C}) + b(\qq_{C}). 
\end{equation*}
By Lemma~\ref{lem:ruta}, this is polynomial in the size of $C$.

\smallskip

\textbf{Completing the proof of Lemma~\ref{lem:lcm2dcm}.} 
Now that we have fixed the grid size, we must show that the two properties of a
discrete contraction map hold. 

The first property requires that each $i$-slice
has a unique fixpoint in the grid. We have specifically chosen the grid to
ensure that each $i$-slice has a fixpoint in the grid, and if the grid had
two fixpoints for a given $i$-slice, then we would immediately obtain a
contradiction using Lemma~\ref{lem:cm1}.
The second property follows immediately from Lemma~\ref{lem:cm2}.




\subsection{Proof of Lemma~\ref{lem:dcm2ufeopl}}
\label{app:dcm2ufeopl}

We will formally define our line by induction over the dimensions. We will start
by defining a line that specifies the sequence of points on the $(n-1)$-surface.
In each step of the induction, we will assume that we have defined the visited 
points on the $j$-surface for all $j > i$, and we will specify the points on
the $i$-surface that are visited by the line.

For each step of our reduction, we will define a \emph{partial} \UFEOPL instance
$L_i = (C_i, S_i, V_i)$ which captures the points of the full line that are
on the $i$-surface. The reason that we call them partial instances, is because the
successor function $S_i$ that we define will not actually point to the next
vertex in the line, since it is not computable in polynomial time. It will
instead give the next point on the full line. In all other respects, the
instance $L_i$ is a valid instance of \UFEOPL, so all circuits will
be polynomial-time computable, the potential will be monotonically increasing,
and the line will end at the fixpoint of the discrete contraction instance.
More formally, the instance $L_i$ will satisfy the following
conditions, which will serve as our inductive hypothesis.
\begin{enumerate}
\itemsep1mm
\item There will a polynomial-time circuit $\point$, which given a vertex $v$ of
$L_i$, will return the point in $P$ that corresponds to $v$.
\label{ind:start}

\item For every vertex $v$ on the line, the point $p = \point(v)$ is on the
$i$-surface.

\item For every vertex $v$ on the line, let $u$ denote the next vertex on the
line. If $p = \point(S_i(v))$ and $q = \point(u)$, then $p_j = q_j$ for all $j \ge i$.
 
\item For the first vertex $v$ on the line, the point $p = \point(v)$ has $p_j = 0$ for all $j \ge i$.

\item For the last vertex $v$ on the line, the point $p = \point(v)$ is a
solution to the discrete contraction problem.

\item The potential function is monotonically increasing along the line.
\label{ind:end}
\end{enumerate}
Condition 3 above specifies the behavior of the successor circuit. It specifies
that, if we ignore dimensions $1$ through $i-1$, then the line is connected.

\smallskip

\noindent \textbf{The base case.}
In the base-case, we define the instance $L_{d-1} = (C_{d-1}, S_{d-1}, V_{d-1})$
in the following way. For each slice $s_x = (\blank, \blank, \dots, \blank, x)
\in \DSlice_d$, Property 1 of a discrete contraction map ensures that there is a unique
point $p \in P_s$ that is a fixpoint of $s$, and so is on the
$(d-1)$-surface. Our line will consist of these points. It will start at the
unique fixpoint of $s_0$, it will then move to the unique fixpoint of
$s_1$, and then to the unique fixpoint of $s_2$, and so on until it reaches
the unique fixpoint of $\mathcal{D}$. Obviously, when we are at $s_1$, we
cannot easily compute $s_2$, but our conditions do not require us to do so. We
are simply required to produce a point that agrees with $s_2$ in dimension $d$,
which we can trivially do by moving one step positively in dimension $d$ from
$s_1$.

We can use Property 2 of a discrete contraction map to determine whether a point
on the $(d-1)$-surface is on this line. This property states that, if $q$ is the
fixpoint of $\mathcal{D}$, then $D_{d}(p) = \up$ if and only if $q_d > p_d$.
So the points on our line will be exactly the points $p$ on the $(d-1)$-surface
for which $D_{d}(p) \in \{\up, \zero\}$.

Formally, we define $L_{d-1}$ in the following way. Each vertex of 
$L_{d-1}$ will be a point $p \in P$, and so we can trivially define the circuit
$\point(p) = p$.
\begin{itemize}
\item For each point $p \in P$ we have that $C_{d-1}(p) = 1$ if and only if both of
the following conditions hold.
\begin{itemize}
\item $p$ is on the $(d-1)$-surface.
\item $D_{d}(p) \in \{\up, \zero\}$.
\end{itemize}

\item For each point $p \in P$ for which $C_{d-1}(p) = 1$, we define
$S_{d-1}(p)$ as follows.
\begin{itemize}
\item If $p$ is $D_{d}(p) = \zero$, then $p$ is the end of the line.
\item Otherwise, $S_{d-1}(p) = p'$, where $p'_d = p_d+1$, and $p'_i = p_i$ for all $i < d$.
\end{itemize}

\item For each point $p \in P$ for which $C_{d-1}(p) = 1$, we define $V_{d-1}(p) = p_d$.
\end{itemize}
The successor function gives, for each point $p$ that is on the line, the point
directly above $p$ in dimension $d$. This point agrees with the next point on
the line (which is the unique fixpoint of the slice defined by $p_d+1$) in
dimension $d$, and so satisfies our definition. The potential function simply
uses coordinate $d$ of each point, which is monotonically increasing along the
line.

\begin{lemma}
The instance $(C_{d-1}, S_{d-1}, V_{d-1})$ satisfies Conditions~\ref{ind:start}
through~\ref{ind:end} of the inductive hypothesis.
\end{lemma}
\begin{proof}
We must check that all conditions of the inductive hypothesis hold. 

\begin{itemize}
\item Condition 1 requires that the function $\point$ is computable in
polynomial time. This is trivially true. 

\item Condition 2 requires that every point $p$ is on the $(d-1)$-surface, which
is enforced by the circuit $C_{d-1}$. 

\item Condition 3 requires that if $v$ is a vertex on the line, and $u$ is the
next vertex on the line, then the point $p = \point(S_{d-1}(v))$ agrees with $q =
\point(u)$ in dimension $d$. This is true, since the $d$th coordinate of
$S_{d-1}(p)$ is $p_d + 1$, and this agrees with the $d$th coordinate of $q$.

\item Condition 4 requires that the first point on the line has zero in the
$d$th coordinate. This is true by definition. 

\item Condition 5 requires that the last vertex on the line is a solution to the
discrete contraction problem. This is true, since by definition, the line ends
at a point $p$ that is on the $(d-1)$-surface that satisfies $D_d(p) = \zero$.
Hence, the point satisfies $D_i(p) = \zero$ for all $i$, and so it is a solution
to the discrete contraction problem.

\item Condition 6 requires that the potential function monotonically increases
along the line. This is true because the potential function of each point is the
final coordinate of that point. Since the line 
passes through each slice $s_x$ until it finds the solution, and since it always
moves from the slice $s_x$ to the slice $s_{x+1}$, we have that the final
coordinate of the points on the line monotonically increases.
\end{itemize}
\end{proof}

\smallskip

\noindent \textbf{The inductive step.}
Suppose that we have an instance $L_{i+1} = (C_{i+1}, S_{i+1}, V_{i+1})$ that satisfies
the inductive hypothesis. We will now describe how to build the instance
$L_i = (C_{i}, S_{i}, V_{i})$.

Let $v$ and $u$ be two adjacent vertices on the line defined by $L_{i+1}$. By
the inductive hypothesis, the point $p = \point(S_{i+1}(v))$ and the point $q =
\point(u)$ agree on dimensions $i+1$ through $n$, but may not agree on dimension
$i$. When we construct $L_i$, we will fill in this gap, by introducing a new
sequence of points that join $p$ with $q$.

Each of these new points will be in the $i$-surface. Let $s_j$ define the slice
$(\blank, \blank, \dots, j, p_{i+1}, p_{i+2}, \dots, p_d)$, ie., the slice where
coordinate $i$ is equal to $j$, and coordinates $i+1$ through $d$ agree with
$p$ (and, by the inductive hypothesis, $q$).
Let $r^j$ be the unique fixpoint of $s_j$, which by definition is a point on
the $i$-surface. If $D_{i}(p) = \up$, then our 
line will go through the sequence $r^{p_{i}}, r^{p_{i}+1}, \dots, r^{q_i}$,
and if $D_{i}(p) = \down$, then our line
will go through the sequence $r^{p_{i}}, r^{p_{i}-1}, \dots, r^{q_i}$. This is
analogous to how we linked two points on the one-surface by a sequence of points
on the zero-surface in our two-dimensional example.

We also introduce a new sequence of points at the start of the line. Let
$v_\text{init}$ be the first vertex of the line $L_{i+1}$. By the inductive
hypothesis, we have that $\point(v_\text{init})$ is the point on the
$i+1$-surface that is zero in dimensions $i+1$
through $d$. We create a new sequence of points starting at the point on the
$i$-surface that
is zero in dimensions $i$ through $d$, and ending 
at $\point(v_\text{init})$. This line is constructed in the same way as the
other lines, by taking the unique $i$-witness of each slice $(\blank, \blank,
\dots, j, 0, 0, \dots, 0)$, where $j$ appears in coordinate $i$ of the slice.

As we saw in the two-dimensional example, we need to remember the point $p$ in
order to determine whether any given point $r^j$ is on this line. So, a vertex
of $L_i$ will consist of a pair $(v, r)$, where $v$ is either
\begin{itemize}
\item $v$ is a vertex of $L_{i+1}$, or
\item the special symbol $\vblank$,
\end{itemize}
and $r \in P$ is a point. For each vertex $(v, r)$ we define $\point(v, r) = r$, which
can clearly be computed in polynomial time. The symbol $\vblank$ will be used in
the initial portion of the line, where we have not yet arrived at the first
vertex of $L_{i+1}$.

The circuit $C_i(v, r)$ is defined as follows.
\begin{itemize}
\item If $v \ne \vblank$, then the circuit returns $1$ if and only if all of the
following hold. Let $u = \point(S_{i+1}(v))$.
\begin{itemize}
\item $C_{i+1}(v) = 1$.
\item $r$ is on the $i$-surface.
\item $r$ and $u$ agree on coordinates $i+1$ through $d$.
\item If $D_{i}(u) = \up$ then $D_{i}(r) \in \{\up, \zero\}$ and $r_i \ge u_i$.
\item If $D_{i}(u) = \down$ then $D_{i}(r) \in \{\down, \zero\}$ and $r_i \le u_i$.
\end{itemize}

\item If $v = \vblank$, then the circuit returns $1$ if and only if all of the
following hold.
\begin{itemize}
\item $r_j = 0$ for all $j > i$.
\item $r$ is on the $i$-surface.
\item $D_{i}(r) \in \{\up, \zero\}$.
\end{itemize}
\end{itemize}
The first bullet specifies the line between $\point(v)$ and
$\point(S_{i+1}(v))$, while the second bullet specifies the line between the
initial vertex of $L_i$, and the initial vertex of $L_{i+1}$.

For each pair $(v, r)$ such that $C_i(v, r) = 1$, the circuit $S_i(v, r)$ is
defined in the following way.
\begin{itemize}
\item If $D_i(r) = \up$, then the circuit returns the vertex $(v, r')$ such that
$r'_i =
r_i + 1$, and $r'_j = r_j$ for all $j \ne i$.
\item If $D_i(r) = \down$, then the circuit returns the vertex $(v, r')$ such
that $r'_i
= r_i - 1$, and $r'_j = r_j$ for all $j \ne i$.
\item If $D_i(r) = \zero$, then the circuit performs the following operation.
\begin{itemize}
\item First it modifies $v$ so that $\point(v) = r$. This is valid because $r$
is an $i$-witness with $D_i(r) = \zero$, and therefore $r$ is an $i+1$-witness.
Let $u$ denote this new vertex of $L_i$.
\item Then it computes $r' = S_{i+1}(u)$.
\item Finally, it outputs the vertex $(u, r')$.
\end{itemize}
\end{itemize}
The first two items simply follow the direction function $D_i$. The third item
above is more complex. Once we arrive at the point with 
$D_i(r) = \zero$, we have found the $i+1$-witness that we are looking for, and
this is the next point on the line $L_{i+1}$. So we use this fact to ask
$S_{i+1}$ for the next vertex of~$L_{i+1}$, and begin following the next line.

For each pair $(v, r)$ such that $C_i(v, r) = 1$, the circuit $V_i(v, r)$ is
defined in the following way. 
\begin{itemize}
\item If $v \ne \vblank$ and $D_i(r) = \up$, then $V_i(v, r) = (k_i+1) \cdot
(V_{i+1}(v) + 1) + r_i$.
\item If $v \ne \vblank$ and $D_i(r) = \down$, then $V_i(v, r) = (k_i+1) \cdot
(V_{i+1}(v) + 1) + k_i - r_i$.
\item If $v = \vblank$, then $V_i(v, r) = r_i$.
\end{itemize}
This definition generalizes the definition that we gave in our two-dimensional
example.

\begin{lemma}
The instance $(C_{i}, S_{i}, V_{i})$ satisfies Conditions~\ref{ind:start}
through~\ref{ind:end} of the inductive hypothesis.
\end{lemma}
\begin{proof}
We prove each condition in turn.
\begin{itemize}
\item 
Condition 1 requires that the circuit $\point$ is computable in polynomial time,
which is clearly true.

\item 
Condition 2 requires that every point on the line is on the $i$-surface. The
definition of $C_i$ ensures that this is true.

\item 
Let $(v, r)$ and $(u, t)$ be two adjacent vertices on the line. Condition 3
requires that we have that $r' = \point(S_i(v, r))$ agrees with $t$ on all
dimensions $j \ge i$. This follows because, in all three cases in the definition
of $S_i$, we have that $r'$ is a member of the slice $s_t = (\blank, \blank,
\dots, \blank, t_{i+1}, t_{i+1}, \dots, t_d)$, and $(u, t)$ is the unique
$i$-witness of the slice $s_t$. 

\item 
Condition 4 requires that the first vertex $v$ on the line have $\point(v)_j =
0$ for all $j \ge i$. This is true by definition.

\item 
Condition 5 requires that the last vertex on the line is the solution to the
discrete contraction problem. This is implied by the inductive hypothesis,
because $L_i$ ends at the last vertex of $L_{i+1}$.

\item 
Condition 6 requires that the potential function is monotonically increasing
along the line. To see this, recall that $k_i$ is the grid width of dimension
$i$. In the definition of $V_i$, We multiply $V_{i+1}(v)$ by  $k_i + 1$ to
ensure that there are at least $k_i$ gaps in the potential function between each
vertex of $L_{i+1}$. We fill these gaps using the $i$th coordinate of $r_i$,
which is either monotonically increasing (in the first case) or monotonically
decreasing (in the second case). We also shift the potentials $V_{i+1}(v)$ by
$1$ to give enough space for the initial path before the first vertex of
$L_{i+1}$.
\end{itemize}
\end{proof}

\smallskip

\noindent \textbf{Completing the proof.} 
To complete the proof of Lemma~\ref{lem:dcm2ufeopl}, it suffices to note that
the instance $(C_1, S_1, V_1)$ is a (non-partial) instance of \UFEOPL.
The inductive hypothesis guarantees that the line is connected, unique, and ends
at the unique fixpoint of $\mathcal{D}$. Furthermore, the potential is
monotonically increasing along the line.

\subsection{Proof of Lemma~\ref{lem:ufeopl2eopl}}
\label{app:ufeopl2eopl}

\noindent \textbf{From \UFEOPL to unique forward \EOML.}
We first reduce our \UFEOPL instance to a unique forward version of \EOML.
To avoid introducing another distinct line following problem, we can define
unique forward \EOML to be a \UFEOPL instance where $V(S(v)) = V(v) + 1$ for
every vertex $v$ that is on the line. For the reduction from \UFEOPL, we can use
exactly the same proof that we used to reduce \EOPL to \EOML. 

Recall that the reduction does the following operation. Let $v$ be a vertex in
the \UFEOPL instance, and let $u = S(v)$. If $V(u) \ne V(v) + 1$, that is if the
potential does not increase by exactly one between $u$ and $v$, then let $p =
V(u) - V(v)$ be the difference in potentials between the two vertices. If $p$ is
negative then we are the end of the line. Otherwise, we introduce $p - 1$ new
vertices between $u$ and $v$, each of which increases the potential by exactly
$1$.

Intuitively, there is no need to use the predecessor circuit in order to perform
this operation. This can be confirmed by inspecting
the proof of Theorem~\ref{thm:eoml2eopl}, and noting that the predecessor
circuit is not used in a way that is relevant for us. Specifically, in the
procedure defining the successor function $S'$, the predecessor circuit is only
used in Step~\ref{itm:Sfive}, where it is used to check that $P(S(v)) = v$, and
hence to
determine whether $v$ is the end of the line. In \UFEOPL, we can use $C$ to
perform this task, and so we do not need to use the predecessor circuit at all.

Note also that, as long as the promise for \UFEOPL is indeed true, then the
unique forward \EOML instance will contain exactly one line.

\smallskip

\noindent \textbf{From unique forward \EOML to \SOVL.}
Let $V$ be the set of vertices in the unique forward \EOML instance, and let
$(C, S, V)$ be the circuits used to define the instance. Let $p$ be a number
that is strictly larger than the largest possible potential that $V$ can
produce.

The set of vertices in the \SOVL instance will be pairs $(v, i)$, where $v \in
V$ is a vertex, and $i \le p$ is either an integer, or the special symbol
$\vblank$. The idea is that each vertex on the line will be represented by the
vertex $(v, \vblank)$. If $v$ is not the end of the line, then we move to
$(S(v), \vblank)$.

However, we must make sure that the end of the line is at a known distance from
the start of the line, and we use the second element to ensure this.  Once we
arrive at the vertex $u$ that is the end of the line, we then move to $(u, 1)$,
and then to $(u, 2)$, and so on, until we reach $(u, p - V(u))$. This ensures
that the end of the line is exactly $p$ steps from the start of the line.

We now formally define the \SOVL instance. The start vertex will be the start
vertex of the unique forward \EOML instance, and the target integer will be $p$.
The circuit $S'$ is defined for a vertex $(v, i)$ so that
\begin{itemize}
\item If $C(v) = 1$ and $C(S(v)) = 1$, meaning that $v$ is not the end of the line, and $i = \vblank$,
then the circuit returns $(S(v), \vblank)$. 
\item If $C(v) = 1$ and $C(S(v)) = 0$, meaning that $v$ is the end of the line, then 
\begin{itemize}
\item If $i = \vblank$ then the circuit returns $(v, 1)$.
\item If $i \ne \vblank$ and $i \le p$, then the circuit returns $(v, i+1)$.
\end{itemize}
\end{itemize}
In all other cases the vertex is not on the line, so the value returned by
$S'$ is irrelevant.
The circuit $W'$ is defined for a vertex $(v, i)$ and an integer $k$ so that
\begin{itemize}
\item If $C(v) = 1$ and $i = \vblank$, then the circuit returns $1$ if and only
and $V(v) = k$.
\item If $C(v) = 1$, $i \ne \vblank$, and $C(S(v)) = 0$, then the circuit
returns $1$ if and only if $k = V(v) + i$ and $V(v) + i \le p$.
\end{itemize}
In all other cases the circuit returns $0$.
Note that, since the potential increases by exactly one at each step of our
line, the circuit $W$ does indeed correctly identify the vertices on the line by
their distance from the start vertex.

Note that if the unique forward \EOML instance has a unique line, then the
promise of the \SOVL problem is indeed satisfied.

\smallskip

\noindent \textbf{Completing the proof of Lemma~\ref{lem:ufeopl2eopl}.} 
To complete the proof, it suffices to apply the result of Hub\'{a}\v{c}ek and
Yogev~\cite{hubavcek2017hardness}, which shows that \SOVL can be reduced to
\EOML using a technique of Bitansky, Paneth and Rosen~\cite{BPR15}. 
Moreover, it can be verified that, so long as the promise of the \SOVL
instance is satisfied, then the proof of Hub\'{a}\v{c}ek and Yogev produces an
\EOML instance that has exactly one line. Therefore, we have reduced \UFEOPL to
unique \EOPL.


\section{The Algorithms and Proofs for Section~\ref{sec:algorithms}}
\label{sec:algorithm_details}


In this section, we provide an exact algorithm for solving \LCM,
which is a promise problem guaranteed to have a rational fixpoint
of polynomial bit complexity.
Then we extend this algorithm to find an approximate fixpoint of general 
contraction maps for which there may not be an exact solution of
polynomial bit complexity.
Our algorithms work for any $\ell_p$ norm with $p \in \Natural$, and are polynomial for constant dimension $d$.
These are the first such algorithms for $p \ne 2$. 
Such algorithms were so far only known for the $\ell_2$ and $\ell_\infty$
norms \cite{HuangKhachSik99,Sik01,ShellSik03}\footnote{Our approach does not cover the $\ell_\infty$
norm, as that would require more work and not give a new result.}

\subsection{Overview: algorithm to find a fixed-point of \LCM}


The algorithm does a nested binary search using Lemmas \ref{lem:cm1} and
\ref{lem:cm2} to find fixpoints of slices with increasing numbers of free
coordinates. We illustrate the algorithm in two dimensions in
Figure~\ref{fig:exact_algo}. The algorithm is recursive. To find the eventual
fixpoint in $d$ dimensions we fix a single coordinate $s_1$, find the
unique $(d-1)$-dimensional fixpoint of $\Restr{f}{\ss}$, the
$(d-1)$-dimensional contraction map obtained by fixing the first coordinate of
the input to $f$ to be $s_1$. Let $x$ the unique fixpoint of $\Restr{f}{\ss}$
where $x_1 = s_1$. If $f(x_1) > s_1$, then the $d$-dimensional fixpoint $x^*$ of
$f$ has $x^*_1 > s_1$, and if $f(x_1) < s_1$, then $x^*_1 < s_1$ (Lemma \ref{lem:cm2}). We can thus do
a binary search for the value of $x^*_1$. Once we've found $x^*_1$, we can
recursively find the $(d-1)$-dimensional fixpoint of $\Restr{f}{\ss}$ where $s_1
= x_1$. The resulting solution will be the $d$-dimensional fixpoint. At each
step in the recursive procedure, we do a binary search for the value of one
coordinate of the fixpoint at the slice determined by all the coordinates
already fixed. For piecewise-linear functions, we know that all fixpoints are
rational with bounded bit complexity (as discussed in Section \ref{sec:discretizing}), so we can find each coordinate exactly.

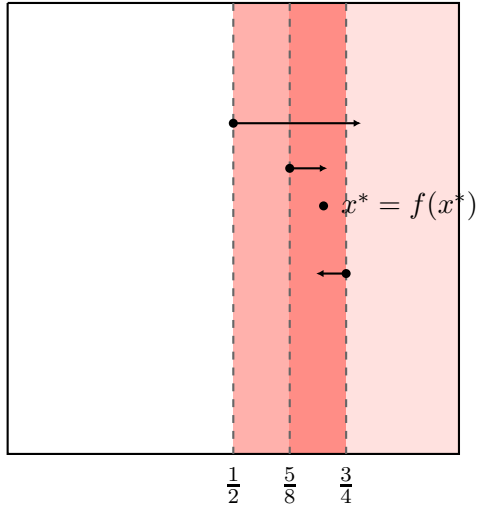
\begin{figure}[h!]
  \centering
  \begin{center}
    \begin{tikzpicture}
  \pgfmathsetmacro{\size}{3}
  \pgfmathsetmacro{\rad}{0.6}
  \coordinate (bl) at (-\size, -\size);
  \coordinate (tl) at (-\size, \size);
  \coordinate (tr) at (\size, \size);
  \coordinate (br) at (\size, -\size);
  \coordinate (first) at (0, 1.4);
  \coordinate (second) at ($(\size/2, -0.6)$);
  \coordinate (third) at ($(\size/4, 0.8)$); 
  \draw[thick] (tl) -- (tr) -- (br) -- (bl) -- (tl);
  \draw[thick, dashed, darkgray] (0, \size) -- (0, -\size);
  \draw[thick, dashed, darkgray] ($(\size/2, \size)$) -- ($(\size/2, -\size)$);
  \draw[thick, dashed, darkgray] ($(\size/4, \size)$) -- ($(\size/4, -\size)$);
    
  \fill (first) circle [radius=\rad mm];
  \fill (second) circle [radius=\rad mm];
  \fill (third) circle [radius=\rad mm];

  \coordinate (fourth) at ($(0, 0)$);
  \coordinate (fifth) at ($(0, \size/2)$);
  \coordinate (sixth) at ($(0, 3*\size/4)$);
  \coordinate (actual) at ($(1.2, 0.3)$);

  \coordinate (b1) at (0,-\size);
  \coordinate (b2) at (\size/2,-\size);
  \coordinate (b3) at (\size/4,-\size);


  \fill (actual) circle [radius=\rad mm];
  \node [right=1mm of actual] {$x^* = f(x^*)$};

  \node [below=0.5mm of b1] {$\frac{1}{2}$};
  \node [below=0.5mm of b2] {$\frac{3}{4}$};
  \node [below=0.5mm of b3] {$\frac{5}{8}$};

  \draw[->, thick] (first.center) -- ++(1.7, 0.0);
  \draw[->, thick] (second.center) -- ++(-0.4, 0.0);
  \draw[->, thick] (third.center) -- ++(0.5, 0.0);

  \begin{pgfonlayer}{background}
    \fill[pastelred, opacity=0.2] (0, -\size) rectangle (\size, \size);
    \fill[pastelred, opacity=0.4] ($(0, -\size)$) rectangle ($(\size/2, \size)$);
    \fill[pastelred, opacity=0.5] ($(\size/4, -\size)$) rectangle ($(\size/2, \size)$);
  \end{pgfonlayer}{background}
\end{tikzpicture}
  \end{center}
  \caption{An illustration of the algorithm to find a fixpoint of a piecewise-linear contraction map in two dimensions. The algorithm begins by finding a fixpoint along the slice with $x_1 = 1/2$. The fixpoint along that slice points to the right, so we next find a fixpoint along the slice with $x_1 = 3/4$. The fixpoint along that slice points to the left, so we find the fixpoint along $x_1 = 5/8$. We successively find fixpoints of one-dimensional slices, and then use those to do a binary search for the two-dimensional fixpoint. The red regions are the successive regions considered by the binary search, where each successive step in the binary search results in a darker region.} 
  \label{fig:exact_algo}
\end{figure}

Using this algorithm we obtain the following theorem.

\begin{theorem}
  Given a $\LinearFIXP$ circuit $C$ encoding a contraction map $f : [0,1]^d\to [0,1]^d$ with respect to any $\ell_p$ norm, there is an algorithm to find a fixpoint of $f$ in time $O(L^d)$ where $L = \poly(\Card{C})$ is an upper bound on the bit-length of the fixpoint of $f$.
\end{theorem}

\todo[inline]{Is it worth separating out the results under the heading ``Fixing the grid size'' into a more general form so that we can state a lemma like: Given a $\LinearFIXP$ circuit $C$ computing a contraction map $f$, all fixpoints of slice restrictions with rational coordinates having denominator at most $2^L$ are rational with denominator at most $2^L$ for some $L = \poly(\Card{C})$?}

The full details of the algorithm can be found in Appendix~\ref{subsec:exact_algo_details}.

\subsection{Overview: algorithm to find an approximate fixed-point of \CM}

Here we generalize our algorithm to find an approximate fixpoint of an arbitrary function
given by an arithmetic circuit, i.e., our algorithm solves \CM, which is 
specified by a circuit $f$ that represents the contraction map,\footnote{The algorithm works even if $f$ is given as an arbitrary black-box, as long as it is guaranteed to be a contraction map.} a
$p$-norm, and $\eps$. Again, let $d$ denote the dimension of the problem, i.e. the number
of inputs (and outputs) of $f$.
Let $x^*$ denote the unique exact fixpoint for the contraction map $f$.
We seek an approximate fixpoint, i.e., a point for which $\Norm{f(x)-x}_p \leq \eps$. 

We do the same recursive binary search as in the algorithm above, but
at each step of the algorithm instead of finding an
exact fixpoint, we will only find an approximate fixpoint of $\Restr{f}{\ss}$. The difficulty in this case will come from the fact that Lemma~\ref{lem:cm2} does not apply to approximate fixpoints. Consider the example illustrated in Figure~\ref{fig:approx_algo}. In this example, $y$ is the unique fixpoint of the slice restriction along the gray dashed line. By Lemma~\ref{lem:cm2}, $(f(y)_1 - y_1)(x^*_1 - y_1) \geq 0$ so if we find $y$, we can observe that $f(y)_1 > y_1$ and recurse on the right side of the figure, in the region labeled $\mathcal{R}$. If we try to use the same algorithm but where we only find approximate fixopints at each step, we'll run into trouble. In this case, if we found $z$ instead of $y$, we would observe that $f(z)_1 < z_1$ and conclude that $x^*_1 < z_1$, which is incorrect. As a result, we would limit our search to the region labeled $\mathcal{L}$, and wouldn't be able to find $x^*$. 
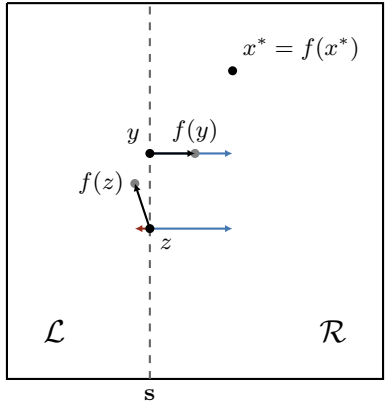
\begin{figure}[h!]
  \centering
  \begin{center}
    \begin{tikzpicture}
  \pgfmathsetmacro{\size}{2.5}
  \pgfmathsetmacro{\rad}{0.6}
  \coordinate (apxc) at (-0.6, -0.5);
  \coordinate (actualc) at (-0.6, 0.5);
  \coordinate (globalc) at (0.5, 1.6);
  \coordinate (apxc end) at ($(apxc) +(-0.2, 0.6)$);
  \coordinate (actualc end) at ($(actualc) +(0.6, 0)$);
  \coordinate (bl) at (-\size, -\size);
  \coordinate (tl) at (-\size, \size);
  \coordinate (tr) at (\size, \size);
  \coordinate (br) at (\size, -\size);
  \coordinate (slice top) at (-0.6, \size);
  \coordinate (slice bottom) at (-0.6, -\size);
    
  \draw[thick] (tl) -- (tr) -- (br) -- (bl) -- (tl);
  \draw[thick, dashed, darkgray] (slice bottom) -- (slice top);
  
  \fill (apxc) circle [radius=\rad mm] node (apx) {}; 
  \fill (actualc) circle [radius=\rad mm] node (actual) {}; 
  \fill[gray] (apxc end) circle [radius=\rad mm, fill=darkgray] node (apx end) {}; 
  \fill[gray] (actualc end) circle [radius=\rad mm, fill=darkgray] node (actual end) {}; 
  \fill (globalc) circle [radius=\rad mm] node (global) {}; 
    
  \draw[->, thick] (apx.center) -- (apxc end) node (apx end) [left] {\footnotesize $f(z)$};
  \draw[->, thick] (actual.center) -- (actualc end) node (actual end) [above] {\footnotesize $f(y)$}; 
  \node [above right] at (global) {\footnotesize $x^* = f(x^*)$};
  \node [above left] at (actual) {\footnotesize $y$};
  \node [below right] at (apx) {\footnotesize $z$};
  \node [below] at (slice bottom) {\footnotesize $\ss$};
  \begin{pgfonlayer}{background}
    \node[above right=0.5cm of bl] (L) {\large $\mathcal{L}$}; 
    \node[above left=0.5cm of br] (R) {\large $\mathcal{R}$}; 
    \draw[->, thick, darkpastelblue] (apx.center) -- ++(1.1, 0.0);
    \draw[->, thick, darkpastelblue] (actual.center) -- ++(1.1, 0);
    \draw[->, thick, darkpastelred] (apx.center) -- ++(-0.2, 0.0);
  \end{pgfonlayer}
\end{tikzpicture}
  \end{center}
  \caption{A step in the recursive binary search. Here, $x^*$ is the fixpoint for the original function, $y$ is the fixpoint for the slice restriction $\Restr{f}{\ss}$ along the dashed gray line, and $z$ is an approximate fixpoint to the slice restriction.}
  \label{fig:approx_algo}
\end{figure}

When looking for an approximate fixpoint, we'll have to choose a different precision $\eps_i$ for each level of the recursion so that either the point $x$ returned by the $i$th recursive call to our algorithm satisfies $\Abs{f(x)_i - x_i} > \eps_i$ and we can rely on it for pivoting in the binary search, or $\Abs{f(x)_i - x_i} \leq \eps_i$ and we can return $x$ as an approximate fixpoint to the recursive call one level up. Each different $\ell_p$ norm will require a different choice of $\Paren{\eps_i}_{i=1}^d$.

Using this idea we are able to obtain the following results:

\begin{theorem}
For a contraction map $f:[0,1]^d \to [0,1]^d$ with respect to the $\ell_1$ norm, there is an algorithm compute a point $v\in [0,1]^d$ such that $\Norm{f(v) - v}_1 < \eps$ in time $O(d^d\log(1/\eps))$. 
\end{theorem}

\begin{theorem}
  For a contraction map $f:[0,1]^d\to [0,1]^d$ under $\Norm{\cdot}_p$ for $2 \leq p < \infty$, there is an algorithm to compute a point $v\in [0,1]^d$ such that $\Norm{f(v) - v}_p < \eps$ in time $O(p^{d^2}\log^d(1/\eps)\log^d(p))$.
\end{theorem}

The full details of the algorithms can be found in Appendix~\ref{subsec:approx_algo_details}.

\subsection{Details: finding a fixed-point of \LCM}
\label{subsec:exact_algo_details}
Suppose $f$ is piecewise-linear; it follows that the coordinates of the unique fixpoints of $\Restr{f}{\ss}$ will be rational numbers with bounded denominators. Consider the values of $k_i$'s computed in Section \ref{sec:discretizing}. 
The analysis in the section tells us that if we consider a slice $\ss$ fixing any $i$ coordinates to numbers with denominators $k_d,\dots,k_{d-i+1}$ then the
unique fixed-point of $\Restr{f}{\ss}$ will have coordinates with denominator bounded by $k_{d-i}$. Furthermore, $k_1$ is the largest among all $k_i$'s and its bit-length is bounded by polynomial in the size of the circuit $C$ representing the given \LCM instance.

Let $\Slice_d(\hh)$ be the set of slices with fixed coordinates having rational values with denominator $h_i$ in the $i$th coordinate:
\[ \ss \in \Slice_d(\hh) \iff \ss \in \Slice_d\ \text{and}\ s_i \in \Set{0/h_i,1/h_i,\dotsc,h_i/h_i},\ \forall i \in \fixed(\ss)\text{.} \] 
From the above discussion, define for each $$i\in\{1,\dots, d\}, \ \ \ L_i=\ln(k_{d-i+1}).$$ We will design an algorithm assuming that 
upper bound of $2^{L_i}$ on the $i$th coordinate denominators of fixpoints for any slice restriction $\Restr{f}{\ss}$ where $\ss \in \Slice_d(2^{L_1},\dots, 2^{L_d})$. 

\begin{algorithm}[t]
	\caption{Algorithm for \LCM}\label{alg}
	
	\begin{algorithmic}[1]
          \State Input: A slice $\ss \in \Slice_d(2^{L_1},\dots,2^{L_d})$ where $\fixed(\ss) = [k]$ for some $k \leq d$.
          \State Output: A point $y$ 
		such that $\Restr{f}{\ss}(y) = \Restr{y}{\ss}$ and $y = \restr{y}{\ss}$.
          \Function{FindFP}{$\ss$}
              \State Let $k = \Card{\fixed(\ss)}$.
              \If{$k= d$} \Return $\ss$ \EndIf
              \State Set $k \gets k + 1$. Set $\tt^l \gets \ss$, $\tt^h \gets \ss$.
              \State Set $t^l_k \gets 0$, $t^h_k \gets 1$. 
              \State Set $v^l \leftarrow \FindFP(\tt^l)$, and  $v^h \leftarrow \FindFP(\tt^h)$.
              \If{$f(v^l)_k = t^l_k$} \Return $v^l$. \EndIf
              \If{$f(v^h)_k = t^h_k$} \Return $v^h$. \EndIf
              \State Set $\tt \gets \ss$
                  \While{$t^h_k - t^l_k > \frac{1}{2^{(L_k+1)}}$} 
                      \State Set $t_k \leftarrow \frac{(t^h_k + t^l_k)}{2}$.
                      \State Set $v \leftarrow\FindFP(\tt)$. \label{alg:vv}
                      \If{$f(v)_k = t_k$}  \Return $v$. \EndIf
                      \If{$f(v)_k > t_k$} Set $t^l_k \leftarrow t_k$
                      \Else\ Set $t^h_k \leftarrow t_k$. \EndIf
                  \EndWhile
              \State $t_k \leftarrow$ unique number in $(t^l_k, t^h_k)$ with denominator at most $2^{L_k}$. 
              \State \Return $\FindFP(\tt)$.
              \EndFunction 
\end{algorithmic}
\end{algorithm}

\smallskip

\noindent \textbf{Analysis.}
Given a $k < d$ and $\ss \in \Slice_d(2^{L_1},\dots,2^{L_d})$ with $\fixed(\ss) = [k]$, from Lemma~\ref{lem:cm1} we know that the restricted function $\Restr{f}{\ss}$ is also contracting. We will show that Algorithm~\ref{alg} computes a fixpoint of $\Restr{f}{\ss}$. 
Since the algorithm is recursive, we will prove its correctness by induction. The next lemma establishes the base case of induction and follows by design of the algorithm and is equivalent to finding a fixpoint of a one dimensional function.

\begin{lemma}\label{lem:cm-ind1}
If $\Card{\fixed(\ss)}=d-1$ then $\FindFP(\ss)$ returns a $v$ such that $v_i = s_i$ for $i \in \fixed(\ss)$ and $f(v)_d = v_d$.  
\end{lemma}

Now for the inductive step, assuming $\FindFP$ can compute a fixpoint of any $q$-dimensional contraction map, we will show that it can compute one for $(q+1)$-dimensional contraction map.

\begin{lemma}\label{lem:cm-ind2}
Fix some $\ss \in \Slice_d(2^{L_1},\dots,2^{L_d})$ with $\fixed(\ss) = [k-1]$ for some $k-1 < d$, and let $\tt = (s_1,s_2,\dotsc,s_{k-1}, t_k, \blank,\dotsc,\blank)$ for some $t_k \in \Set{0/2^{L_k},1/2^{L_k},\dotsc,2^{L_k}/2^{L_k}}$. If $\FindFP(\tt)$ returns the unique fixpoint of the restricted function $\Restr{f}{\tt}$, then $\FindFP(\ss)$ returns the unique fixpoint of the function $\Restr{f}{\ss}$.
\end{lemma}
\begin{proof}
Since $f$ is a contraction map, so is $\Restr{f}{\ss}$ due to Lemma~\ref{lem:cm1}. The induction hypothesis ensures that for any value of $t_k\in \Set{0/2^{L_1}, 1/2^{L_k}, \dotsc, 2^{L_k}/2^{L_k}}$ if $v = \FindFP(\tt)$ then $v$ is the unique fixpoint of the function $\Restr{f}{\tt}$, i.e., $f(v)_j = t_j$ for $j \in \free(\tt)$ and $v = \restr{v}{t}$. Now, if $t_k=0$ then $f(v)_k \geq 0 = t_k = v_k$ and if $t_k=1$ then $f(v)_k \leq 1 = t_k = v_k$. If either is an equality then $v$ is the unique fixpoint of $\Restr{f}{\ss}$ as well and the lemma follows. 

Otherwise, we know that the $k$th coordinate of the fixpoint of $\Restr{f}{\ss}$, which we'll call $t^*_k$, is between $t^l_k=0$ and $t^h_k=1$.
Note that when $t_k$ is set to $t^*_k$, the vector $v = \FindFP(\tt)$ will be the unique fixpoint of $\Restr{f}{\ss}$, since both $\Restr{f}{\tt}$ and $\Restr{f}{\ss}$ have unique fixpoints. Thus, it suffices to show that $t_k$ will eventually be set to $t^*_k$ during the execution of the algorithm. 

The while loop of Algorithm~\ref{alg} does a binary search between $t^h_k$ and $t^l_k$ to find $t^*_k$, while keeping track of the fixpoints of $\Restr{f}{\tt}$. We first observe that after line~\ref{alg:vv} of each execution of the loop in Algorithm~\ref{alg}, $v$ is the unique fixpoint of $\Restr{f}{\tt}$. Therefore, whenever $f(v)_k=t_k$ is satisfied, we will return the unique fixpoint of $\Restr{f}{\ss}$ and the lemma follows. 

The binary search maintains the invariant that if we let $v^l = \FindFP(t^l)$ and $v^h = \FindFP(t^h)$ we have $f(v^l)_k > v^l_k$, and $f(v^h)_k < v^h_k$. By Lemma~\ref{lem:cm2}, this invariant ensures that $t^*_k$ satisfies
\[ t^l_k < t^*_k < t^h_k \] at all times. Therefore, binary search either returns the desired fixpoint or ends with $t^l_k$ and $t^h_k$ such that $(t^h_k-t^l_k)\leq 1/2^{L_k+1}$ and $t^*_k \in [t^l_k, t^h_k]$. By the assumption we know that $t^*_k$ is a rational number with denominator at most $2^{L_k}$. Since there can be at most one such number in $[l_k, h_k]$, $t^*_k$ can be uniquely identified. And therefore the last line of Algorithm~\ref{alg} returns a fixpoint of $\Restr{f}{\ss}$. 
\end{proof}

Applying induction using Lemma~\ref{lem:cm-ind1} as a base-case and Lemma~\ref{lem:cm-ind2} as an inductive step, the next theorem follows.

\begin{theorem}
$\FindFP(\blank,\blank,\dotsc,\blank)$ returns the unique fixpoint of $f$ in time $O(L^d)$ where $L=\max_k L_k$ is polynomial in the size of the input instance. 
\end{theorem}

\subsection{Details: finding an approximate fixed-point of \CM}
\label{subsec:approx_algo_details}
Each different $\ell_p$ norm will require a different choice of $\Paren{\eps_i}_{i=1}^d$. There are two distinct cases to consider for $\Paren{\eps_i}_{i=1}^d$:
\begin{itemize}
\item For the $p=1$, we choose $\eps_i = \eps/2^{2i}$.
\item For $2 \leq p < \infty$, we choose $\eps_i = \eps^{p^i} p^{-2\sum_{j=0}^i p^j}$.
\end{itemize}

We now prove lemmas showing that these choices of $\Paren{\eps_i}_{i=1}^d$ allow us to make progress at each step of the algorithm. In both of the following lemmas, let $\ss \in \Slice_d$ with $\fixed(\ss) = [k-1]$ for $k-1 < d$.

\begin{lemma}\label{lem:approx_cm-l1}
  Let $f$ be a contraction map with respect to $\Norm{\cdot}_1$, and let $\eps_i = \eps/2^i$. Let $x^*$ be the unique fixpoint of $\Restr{f}{\ss}$. Given any point $v \in [0,1]^d$ with $\Abs{f(v)_i - v_i} \leq \eps_i$ for all $i > k$, either $\Abs{f(v)_k - v_k} \leq \eps_k$ or $(f(v)_k - v_k) (x^*_k - v_k) > 0$.
\end{lemma}

\begin{proof}
  Assume that $v$ satisfies $\Abs{f(v)_i - v_i} \leq \eps_i$ for all $i > k$, and $\Abs{f(v)_k - v_k} > \eps_k$. Without loss of generality, let $f(v)_k - v_k > \eps_k$. Assume towards a contradiction that $x^*_k \leq v_k$. Then
  \begin{align*}
    \Norm{\Restr{f}{\ss}(v) - \Restr{x^*}{\ss}}_1
    &= \sum_{j=k}^d \Abs{f(v)_j - x^*_j}\\
    &> \eps_k + \Abs{v_k - x^*_k} + \sum_{j=k+1}^d \Abs{f(v)_j - x^*_j}\\
    &\geq \eps_k + \Abs{v_k - x^*_k} + \sum_{j=k+1}^d \Brack{\Abs{v_j - x^*_j} - \eps_j}\\
    &= \eps_k - \sum_{j=k+1}^d \eps_j + \Norm{\Restr{v}{\ss} - \Restr{x^*}{\ss}}_1\\
    &> \Norm{\Restr{v}{\ss} - \Restr{x^*}{\ss}}_1\text{.}
  \end{align*}
  The last line uses the fact that $\eps/2^{2k} \geq \sum_{j=k+1}^d \eps/2^{2j}$ for all $k \leq d$. We now have $\Norm{\Restr{f}{\ss}(v) - \Restr{x^*}{\ss}}_1 \geq \Norm{\Restr{v}{\ss} - \Restr{x^*}{\ss}}_1$ which contradicts $\Restr{f}{\ss}$ being a contraction map. This completes the proof.
\end{proof}

\begin{lemma}\label{lem:approx_cm-lp}
  Let $f$ be a contraction map with respect to $\Norm{\cdot}_p$, and let $\eps_i = \eps^{p^j} p^{-\sum_{j=0}^i p^j}$. Let $x^*$ be the unique fixpoint of $\Restr{f}{\ss}$. Given any point $v \in [0,1]^d$ with $\Abs{f(v)_i - v_i} \leq \eps_i$ for all $i > k$, either $\Abs{f(v)_k - v_k} \leq \eps_k$ or $(f(v)_k - v_k) (x^*_k - v_k) > 0$.
\end{lemma}

\begin{proof}
  Assume that $v$ satisfies $\Abs{f(v)_i - v_i} \leq \eps_i$ for all $i > k$, and $\Abs{f(v)_k - v_k} > \eps_k$. Without loss of generality, let $f(v)_k - v_k > \eps_k$. Assume towards a contradiction that $x^*_k \leq v_k$. Then
  \begin{align}
    \Norm{\Restr{f}{\ss}(v) - \Restr{x^*}{\ss}}_p^p
    &= \sum_{j=k}^d \Abs{f(v)_j - x^*_j}^p\\
    &> \eps_k^p + \Abs{v_k - x^*_k}^p + \sum_{j=k+1}^d \Abs{f(v)_j - x^*_j}^p\\
    &\geq \eps_k^p + \Abs{v_k - x^*_k}^p + \sum_{j=k+1}^d \Paren{{\Abs{v_j - x^*_j} - \eps_j}}^p\\
    &= \eps_k^p + \sum_{j=k}^d \Abs{v_k - x^*_k}^p - \sum_{j=k+1}^d \Brack{\Abs{v_j - x^*_j}^p - \Paren{{\Abs{v_j - x^*_j} - \eps_j}}^p}\\
    &= \Norm{\Restr{v}{\ss} - \Restr{x^*}{\ss}}_p^p + \eps_k^p - \sum_{j=k+1}^d \Brack{\Abs{v_j - x^*_j}^p - \Paren{{\Abs{v_j - x^*_j} - \eps_j}}^p} \\
    &\geq \Norm{\Restr{v}{\ss} - \Restr{x^*}{\ss}}_p^p + \eps_k^p - \sum_{j=k+1}^d p\eps_j \label{approx_ub_step}\\
    &> \Norm{\Restr{v}{\ss} - \Restr{x^*}{\ss}}_p^p \text{.} \label{approx_eps_step}\\
  \end{align}
  Line \ref{approx_ub_step} uses the fact that $a^p - (a - b)^p \leq 1 - (1 - b)^p$ for $p > 1$, $a,b\in (0,1)$. Line \ref{approx_eps_step} follows from our choice of $\Paren{\eps_i}_{i=1}^d$; we have
  \begin{align*}
    \sum_{j=k+1}^d p\eps_j
    &= \sum_{j=k+1}^d p\eps^{p^{j}} p^{-2\sum_{j=0}^i p^i + p}\\
    &< \sum_{j=k+1}^d \eps^{p^{k+1}} p^{-2\sum_{j=1}^i p^i - p}\\
    &< \eps^{p^{k+1}} p^{-2\sum_{j=1}^{k+1}p^i - p} \Brack{\sum_{j=1}^{k+1} p^{-2p}}\\
    &< \eps_{k}^p\text{.}
  \end{align*}

Again we have a contradiction since $\Norm{\Restr{f}{\ss}(v) - \Restr{x^*}{\ss}}_1 \geq \Norm{\Restr{v}{\ss} - \Restr{x^*}{\ss}}_1$ and $\Restr{f}{\ss}$ is a contraction map. The lemma follows.
\end{proof}

We will also need to establish one more lemma for each case:

\begin{lemma}\label{lem:approx_cm-l1_final}
  Let $f$ be a contraction map with respect to $\Norm{\cdot}_1$ and let $\eps_i = \eps/2^{2i}$. Let $x^*$ be the unique fixpoint of $\Restr{f}{\ss}$ and let $t^h_k, t^l_k \in [0,1]$ be such that $t^h_k - t^l_k < \eps_k/2$ and the unique fixpoint $x^*$ of $\Restr{f}{\ss}$ satisfies
  \[ t^l_k < x^* < t^h_k\text{.} \] Let $v^l, v^h \in [0,1]^d$ be such that $v^l = \restr{v^l}{\ss}$ and $v^l_k = t^l_k$, and $v^h = \restr{v^h}{\ss}$ and $v^h_k = t^h_k$. Furthermore, assume that $\Abs{f(v^h)_i - v^h_i} \leq \eps_i$ and $\Abs{f(v^l)_i - v^l_i} \leq \eps_i$ for all $i > k$. Then either $\Abs{f(v^h)_k - v^h_k} \leq \eps_k$ or $\Abs{f(v^l)_k - v^l_k} \leq \eps_k$.
\end{lemma}

\begin{proof}
  We will prove this by contradiction. Assume that $\Abs{f(v^h)_k - v^h_k} > \eps_k$ and $\Abs{f(v^l)_k - v^l_k} > \eps_k$. We will show that $\Norm{\Restr{f}{\ss}(v^h) - \Restr{f}{\ss}(v^l)}_1 - \Norm{\Restr{v^h}{\ss} - \Restr{v^l}{\ss}}_1 > 0$, which contradicts $\Restr{f}{\ss}$ being a contraction map. We have
  \begin{align*}
    \Norm{\Restr{f}{\ss}(v^h) - \Restr{f}{\ss}(v^l)}_1 - \Norm{\Restr{v^h}{\ss} - \Restr{v^l}{\ss}}_1
    &= \sum_{j=k}^d \Brack{\Abs{f(v^h)_j - f(v^l)_j} - \Abs{v^h_j - v^l_j}}\\
    &= \Brack{\Abs{f(v^h)_k - f(v^l)_k} - \Abs{v^h_k - v^l_k}} + \sum_{j=k+1}^d \Brack{\Abs{f(v^h)_j - f(v^l)_j} - \Abs{v^h_j - v^l_j}}\\
    &> 3\eps/2^{2k+1} - \sum_{j=k+1}^d 2\eps/2^{2j}\\
    &> 0,
    \end{align*} completing the proof.
\end{proof}
\begin{lemma}\label{lem:approx_cm-lp_final}
  Let $f$ be a contraction map with respect to $\Norm{\cdot}_p$ and let $\eps_i = \eps^{p^i} p^{-\sum_{j=0}^i p^i}$. Let $x^*$ be the unique fixpoint of $\Restr{f}{\ss}$ and let $t^h_k, t^l_k \in [0,1]$ be such that $t^h_k - t^l_k < \eps_k$ and the unique fixpoint $x^*$ of $\Restr{f}{\ss}$ satisfies
  \[ t^l_k < x^* < t^h_k\text{.} \] Let $v^l, v^h \in [0,1]^d$ be such that $v^l = \restr{v^l}{\ss}$ and $v^l_k = t^l_k$, and $v^h = \restr{v^h}{\ss}$ and $v^h_k = t^h_k$. Furthermore, assume that $\Abs{f(v^h)_i - v^h_i} \leq \eps_i$ and $\Abs{f(v^l)_i - v^l_i} \leq \eps_i$ for all $i > k$. Then either $\Abs{f(v^h)_k - v^h_k} \leq \eps_k$ or $\Abs{f(v^l)_k - v^l_k} \leq \eps_k$.
\end{lemma}

\begin{proof}
  We will prove this by contradiction. Assume that $\Abs{f(v^h)_k - v^h_k} > \eps_k$ and $\Abs{f(v^l)_k - v^l_k} > \eps_k$. We will show that $\Norm{\Restr{f}{\ss}(v^h) - \Restr{f}{\ss}(v^l)}_p^p - \Norm{\Restr{v^h}{\ss} - \Restr{v^l}{\ss}}_p^p > 0$, which contradicts $\Restr{f}{\ss}$ being a contraction map. We have
  \begin{align}
    \Norm{\Restr{f}{\ss}(v^h) - \Restr{f}{\ss}(v^l)}_p^p - \Norm{\Restr{v^h}{\ss} - \Restr{v^l}{\ss}}_p^p
    &= \sum_{j=k}^d \Brack{\Abs{f(v^h)_j - f(v^l)_j}^p - \Abs{v^h_j - v^l_j}^p}\\
    &= \Brack{\Abs{f(v^h)_k - f(v^l)_k}^p - \Abs{v^h_k - v^l_k}^p} \\
    &\quad\quad+ \sum_{j=k+1}^d \Brack{\Abs{f(v^h)_j - f(v^l)_j}^p - \Abs{v^h_j - v^l_j}^p} \nonumber\\
    &\geq \Brack{\Abs{f(v^h)_k - f(v^l)_k} - \Abs{v^h_k - v^l_k}}^p\\
    &\quad\quad+ \sum_{j=k+1}^d \Brack{\Paren{\Abs{v^h_j - v^l_j} - 2\eps_j}^p - \Abs{v^h_j - v^l_j}^p} \nonumber\\
    &\geq \Paren{3\eps_k/2}^p - \sum_{j=k+1}^d 2p\eps_j \label{approx_ub_step2}\\
    &> \eps_k^p - \sum_{j=k+1}^d p\eps_j \label{approx_3/2_step2}\\
    &> 0\text{.} \label{approx_eps_step2}
    \end{align} Here, line \ref{approx_ub_step2} mirrors line \ref{approx_ub_step} from Lemma~\ref{lem:approx_cm-lp} (which holds for $\eps_i < 1/2$, which is the case here). Line \ref{approx_3/2_step2} uses the fact that $3/2^p > 2$ for $p > 2$. Finally, line \ref{approx_eps_step2} uses the same argument as in line \ref{approx_eps_step} in Lemma~\ref{lem:approx_cm-lp}.
\end{proof}

We now proceed to prove the correctness of our algorithm. 

\begin{algorithm}[t]
	\caption{Algorithm for \CM}\label{alg2}
	
	\begin{algorithmic}[1]
          \State Input: A slice $\ss \in \Slice_d$ such that $\fixed(\ss) = [k]$ for some $k\leq d$.
          \State Output: A point $v\in [0,1]^d$ such that $\Abs{f(v)_i - v_i} \leq \eps_i,\ \forall i \in \free(\ss)$ and $v = \restr{v}{\ss}$.
          \Function{ApproxFindFP}{$\ss$}
              \State Let $k = \Card{\fixed(\ss)}$.
              \If{$k= d$} \Return $\ss$ \EndIf
              \State Set $k \gets k + 1$. Set $\tt^l \gets \ss$, $\tt^h \gets \ss$.
              \State Set $t^l_k \gets 0$, $t^h_k \gets 1$. 
              \State Set $v^l \leftarrow \ApproxFindFP(\tt^l)$, and  $v^h \leftarrow \ApproxFindFP(\tt^h)$.
              \If{$\Abs{f(v^l)_k - v^l_k} \leq \eps_k$} \Return $v^l$. \EndIf
              \If{$\Abs{f(v^h)_k - v^h_k} \leq \eps_k$} \Return $v^h$. \EndIf
              \State Set $\tt \gets \ss$
                  \While{$t^h_k - t^l_k > \eps_k$}
                      \State Set $t_k \leftarrow \frac{(t^h_k + t^l_k)}{2}$.
                      \State Set $v \leftarrow\ApproxFindFP(\tt)$. 
                      \If{$\Abs{f(v)_k - v_k} \leq \eps_k$}  \Return $v$. \EndIf
                      \If{$f(v)_k > t_k$} Set $t^l_k \leftarrow t_k$
                      \Else\ Set $t^h_k \leftarrow t_k$. \EndIf
                  \EndWhile
              \State $t_k \leftarrow \frac{(t^h_k + t^l_k)}{2}$.
              \State \Return $\ApproxFindFP(\tt)$. \label{alg2:final_output}
              \EndFunction 
\end{algorithmic}
\end{algorithm}

\smallskip

\noindent \textbf{Analysis.} 
We will show that for any norm $\Norm{\cdot}$ and $\eps$-sequence for $\Norm{\cdot}$, Algorithm~\ref{alg2} returns a point $v\in [0,1]^d$ such that $\Norm{f(v) - v} \leq \eps$. To do this, we will show that for any $k < d$ and slice $\ss \in \Slice_d$ with $\fixed(\ss) = [k]$, the point $v$ returned by $\ApproxFindFP(\ss)$ will satisfy
\[ \Abs{f(v)_i - v_i} \leq \eps_i,\ \forall i\in \free(\ss)\text{.} \] Since Algorithm~\ref{alg2} is recursive, our proof will be by induction. The next lemma establishes the base case of the induction and follows by design of the algorithm.

\begin{lemma}\label{lem:approx-cm-ind1}
If $\Card{\fixed(\ss)}=d-1$ then $\FindFP(\ss)$ returns a $v$ such that $\Abs{f(v)_d - v_d} \leq \eps_d$ and $v = \restr{v}{\ss}$.
\end{lemma}

Now for the inductive step, we show that if $\ApproxFindFP$ can compute an approximate fixpoint of any $q$-dimensional contraction map, we will show that it can compute one for the $(q+1)$-dimensional contraction map instances it is given.

\begin{lemma}\label{lem:approx-cm-ind2}
  Fix some $\ss \in \Slice_d$ with $\fixed(\ss) = [k-1]$ for some $k-1 < d$, and let $\tt = (s_1,s_2,\dotsc,s_{k-1}, t_k, \blank,\dotsc,\blank)$ for some $t_k \in [0,1]$. If $\ApproxFindFP(\tt)$ returns a point $v$ such that
  \[ \Abs{f(v)_i - v_i} \leq \eps_i,\ \forall i \in \free(\tt) \] and $v = \restr{v}{\tt}$, then $\ApproxFindFP(\ss)$ returns a point such that \[ \Abs{f(v)_i - v_i} \leq \eps_i,\ \forall i \in \free(\ss)\text{.} \]
\end{lemma}
\begin{proof}
  Since $f$ is a contraction map, so is $\Restr{f}{\ss}$ due to Lemma~\ref{lem:cm1}. We assume that for any value of $t_k\in [0,1]$ if $v = \FindFP(\tt)$ then $v$ satisfies
  \[ \Abs{f(v)_i - v_i} \leq \eps_i,\ \forall i \in \free(\tt) \] and $v= \restr{v}{\tt}$.

  We first observe that after the first two calls to $\ApproxFindFP(\tt^l)$ and $\ApproxFindFP(\tt^h)$, if $\Abs{f(v^h)_k - v^h_k} \leq \eps_k$ or $\Abs{f(v^l)_k - v^l_k}\leq \eps_k$, we return $v^h$ or $v^l$, respectively, so the output of $\ApproxFindFP(\ss)$ satisfies the requirements of the lemma.

  Moreover, in every subsequent call to $\ApproxFindFP(\tt)$, if the output $v$ satisfies $\Abs{f(v)_k - v_k} \leq \eps_k$, we return $v$, so we'll assume in what follows that the value returned by the algorithm is from line \ref{alg2:final_output}.

  In what follows let $t^*_k$ be the value of the $k$th coordinate of the unique fixpoint of $\restr{f}{\ss}$.

The binary search in the while loop maintains the invariant that if we let $v^l = \FindFP(t^l)$ and $v^h = \FindFP(t^h)$ we have $f(v^l)_k - v^l_k > \eps_k$, and $f(v^h)_k - v^h_k < -\eps_k$. By Lemmas \ref{lem:approx_cm-l1} and \ref{lem:approx_cm-lp}, this invariant ensures that $t^*_k$ satisfies
\[ t^l_k < t^*_k < t^h_k \] at all times.

Now we just need to argue that when we get to $t^h_k - t^l_k < \eps_k$, the point $v$ returned by the final call to $\ApproxFindFP(\tt)$ will satisfy the required conditions. To see this we appeal to Lemmas \ref{lem:approx_cm-l1_final} and \ref{lem:approx_cm-lp_final}, noting that $t^h_k - t_k < \eps_k / 2$ and $t_k - t^l_k < \eps_k / 2$ and both $\Abs{f(v^h)_k - v^h_k}, \Abs{f(v^l)_k - v^l_k} > \eps_k$. It then follows that $\Abs{f(v)_k - v_k} \leq \eps$ since either $t^l_k < t^*_k < t_k$ or $t_k < t^*_k < t^h_k$.
\end{proof}

Applying induction using Lemma~\ref{lem:approx-cm-ind1} as a base-case and Lemma~\ref{lem:approx-cm-ind2} as an inductive step, we obtain the following theorems:

\begin{theorem}
For a contraction map with respect to the $\ell_1$ norm, $\ApproxFindFP(\blank,\blank,\dotsc,\blank)$ returns a point $v\in [0,1]^d$ such that $\Norm{f(v) - v}_1 < \eps$ in time $O(d^d\log(1/\eps))$. 
\end{theorem}

\begin{proof}
  In the worst case, the $i$th recursive call to $\ApproxFindFP$ will terminate with $t^h_i - t^l_i < \eps_i = \eps/4^i$ so will require at most $O(\log(1/\eps)i)$ iterations. The total runtime will thus be bounded by $O(d^d\log^d(1/\eps))$. The final point returned will satisfy
  \[ \Abs{f(v)_i - v_i} < \eps_i,\quad \forall i\in [d] \] which implies that \[\Norm{f(v) - v} = \sum_{i=1}^d \Abs{f(v)_i - v_i} \leq \sum_{i=1}^d \eps_i < \eps\text{.} \]
\end{proof}

\begin{theorem}
  For a contraction map under $\Norm{\cdot}_p$ for $2 \leq p < \infty$, $\ApproxFindFP(\blank,\blank,\dotsc,\blank)$ returns a point $v\in [0,1]^d$ such that $\Norm{f(v) - v}_p < \eps$ in time $O(p^{d^2}\log^d(1/\eps)\log^d(p))$.
\end{theorem}

\begin{proof}
  In the worst case, the $i$th recursive call to $\ApproxFindFP$ will terminate with $t^h_i - t^l_i < \eps_i = \eps^{p^i}p^{-2\sum_{j=0}^i p^j}$ so will require $O(p^i\log(1/\eps)\log(p))$ iterations. The total runtime will thus be bound by $O(p^{d^2}\log^d(1/\eps)\log^d(p))$. The final point returned will satisfy
    \[ \Abs{f(v)_i - v_i} < \eps_i,\quad \forall i\in [d] \] which implies that \[\Norm{f(v) - v}_p^p = \sum_{i=1}^d \Abs{f(v)_i - v_i}^p \leq \sum_{i=1}^d \eps^p_i < \eps^p\text{.} \]
\end{proof}


\section{\MMCM is CLS-Complete}
\label{sec:MMCMisCLScomplete}

In this section, we define \MMCM and show that it is \CLS-complete.
First, following~\cite{daskalakis2011continuous}, we define the complexity class
$\CLS$ as the class of problems that are reducible to the following problem
\CLO.

\begin{definition}[\CLO~\cite{daskalakis2011continuous}]
\label{def:CLO}
Given two arithmetic circuits computing functions $f : [0,1]^3\to [0,1]^3$ and $p :
[0,1]^3 \to [0,1]$ and parameters $\e, \lambda > 0$, find either:
\begin{enumerate}[leftmargin=*,label=(C\arabic*)]
\itemsep1mm
\item a point $x\in [0,1]^3$ such that $p(x) \leq p(f(x)) - \e$ or \label{c_fixpoint}
\item a pair of points $x,y\in [0,1]^3$ satisfying either \label{c_violation}
  \begin{enumerate}[label=(C\arabic{enumi}\alph*)] 
  \item $\Norm{f(x) - f(y)} > \lambda \Norm{x-y}$ or \label{c_bad_f}
  \item $\Norm{p(x) - p(y)} > \lambda \Norm{x-y}$. \label{c_bad_p}
  \end{enumerate}
\end{enumerate}
\end{definition}

In Definition~\ref{def:CLO}, $p$ should be thought of as a \emph{potential}
function, and $f$ as a \emph{neighbourhood} function that gives a candidate
solution with better potential if one exists. Both of these functions are 
purported to be Lipschitz continuous. A solution to the problem is either an approximate
potential minimizer or a witness for a violation of Lipschitz continuity.

\begin{definition}[\CM~\cite{daskalakis2011continuous}]
\label{def:cm}
We are given as input an arithmetic circuit computing $f: [0,1]^3\to [0,1]^3$,
a choice of norm $\Norm{\cdot}$, constants \mbox{$c \in (0,1)$} and $\delta > 0$, and we are promised that $f$ is $c$-contracting w.r.t. $\Norm{\cdot}$.
The goal is to find
\begin{enumerate}[label=(CM\arabic*)]
\itemsep1mm
\item a point $x\in [0,1]^3$ such that $d(f(x),x) \leq \delta$, 
\item or two points $x,y\in [0,1]^3$ such that $\Norm{f(x) - f(y)}/\Norm{x-y} > c$. 
\end{enumerate}
\end{definition}

In other words, the problem asks either for an approximate fixpoint of $f$ or
a violation of contraction. As shown in~\cite{daskalakis2011continuous}, \CM is
easily seen to be in \CLS by creating instances of \CLO with $p(x) =
\Norm{f(x)-x}$, $f$ remains as $f$, Lipschitz constant $\lambda = c+1$, and $\epsilon =
(1-c)\delta$.

In a \emph{meta-metric}, all the requirements of a metric are satisfied except
that the distance between identical points is not necessarily zero. The
requirements for $d$ to be a meta-metric are given in the following definition.

\begin{definition}[Meta-metric]
\label{def:metametric}
Let $\mathcal{D}$ be a set and $d:\mathcal{D}^2 \mapsto \Real$ a function such that:
\begin{enumerate}
\itemsep1mm
\item $d(x, y) \ge 0$;
\item $d(x, y) = 0$ implies $x = y$ (but, unlike for a metric, the converse is not required);
\item $d(x, y) = d(y, x)$;
\item $d(x, z) \le d(x, y) + d(y, z)$.
\end{enumerate}
Then $d$ is a meta-metric on $\mathcal{D}$.
\end{definition}

The problem \CM, as defined in~\cite{daskalakis2011continuous}, was inspired by
Banach's fixpoint theorem, where the contraction can be with respect to any
metric.  In~\cite{daskalakis2011continuous}, for \CM the assumed metric was any
metric induced by a norm. The choice of this norm (and thus metric) was
considered part of the definition of the problem, rather than part of the
problem input. In the following definition of \MMCM, the contraction is with
respect to a meta-metric, rather than a metric, and this meta-metric is given as part of the input of
the problem.

\begin{definition}[\MMCM]
\label{def:MMCM}
We are given as input an arithmetic circuit computing $f: [0,1]^3\to [0,1]^3$,
an arithmetic circuit computing a meta-metric $d : [0,1]^3\times [0,1]^3 \to
[0,1]$, some $p$-norm $\Norm{\cdot}_r$ and constants \mbox{$\e,c \in (0,1)$}
and $\delta > 0$, and we are promised that $f$ is $c$-contracting with
respect to $d$, and $\lambda$-continuous with respect to $\Norm{\cdot}$, and
that $d$ is $\gamma$-continuous with respect to $\Norm{\cdot}$. The goal is
to find
\begin{enumerate}[label=(M\arabic*)]
\itemsep1mm
\item a point $x\in [0,1]^3$ such that $d(f(x),x) \leq \e$, \label{m_fixpoint}
\item or two points $x,y\in [0,1]^3$ such that \label{m_violation}
  \begin{enumerate}[label=(M\arabic{enumi}\alph*)]
	\itemsep1mm
    \item $d(f(x),f(y))/d(x,y) > c$, \label{m_not_contracting}
    \item $\Norm{d(x,y) - d(x',y')}/\Norm{(x,y)-(x',y')} > \delta$, or \label{m_bad_metametric}
    \item $\Norm{f(x) - f(y)}/\Norm{x-y} > \lambda$. \label{m_bad_f}
  \end{enumerate}
\item points $x,y$, or $x,y,z$ in $[0,1]^3$ that witness a violation of one 
	of the four defining properties of a meta-metric (Definition~\ref{def:metametric}). \label{mm_violation}
\end{enumerate}
\end{definition}

\begin{definition}[\GCM]
\label{def:GCM}
The definition is identical to that of Definition~\ref{def:MMCM} identical except for the fact that 
solutions of type \ref{mm_violation} are not allowed.
\end{definition}

So, while \MMCM allows violations of $d$ being a meta-metric
as solutions, \GCM does not. 

\begin{theorem}
\label{thm:GCMinCLS}
  \GCM is in \CLS.
\end{theorem}
\begin{proof}
  Given an instance $X=(f,d,\e,c,\lambda,\delta)$ of \GCM, we set $p(x) \triangleq d(f(x),x)$. Then our \CLO instance is the following:
  \[Y=(f, p, \lambda' \triangleq (\lambda + 1) \delta, \e' \triangleq (1-c)\e).\]
Now consider any solution to $Y$. If our solution is of type \ref{c_fixpoint}, a
	point $x$ such that $p(f(x)) > p(x) - \e'$, then we have $d(f(f(x)),f(x))
	> d(f(x),x) - (1-c)\e$, and either $d(f(x),x) \leq \e$, in which case $x$ is
	a solution for $X$, or $d(f(x),x) > \e$. In the latter case, we can divide
	on both sides to get \[ \frac{d(f(f(x)),f(x))}{d(f(x),x)} > 1-
	\frac{(1-c)\e}{d(f(x),x)} \geq 1- (1-c) = c\text{,} \] giving us a violation
	of the claimed contraction factor of $c$, and a solution of type
	\ref{m_not_contracting}.

If our solution is a pair of points $x,y$ of type \ref{c_bad_f} satisfying $\Norm{f(x) - f(y)}/\Norm{x-y} > \lambda' \geq \lambda$, then this gives a violation of the $\lambda$-continuity of $f$. If instead $x,y$ are of type \ref{c_bad_p} so that $\Norm{p(x) - p(y)}/\Norm{x-y} > \lambda'$, then we have
\[ \Abs{d(f(x),x) - d(f(y),y)} = \Abs{p(x) - p(y)} > (\lambda+1)\delta \Norm{x-y}\text{.} \]
We now observe that if
\[ \Abs{d(f(x),x) - d(f(y),y)} \leq \delta (\Norm{f(x)-f(y)} + \Norm{x - y}) \quad \text{and} \Norm{f(x) - f(y)}/\Norm{x-y} \leq \lambda,\] 
	then we would have
\[\Abs{d(f(x),x) - d(f(y),y)} \leq \delta (\Norm{f(x) - f(y)} + \Norm{x-y}) \leq (\lambda + 1)\delta \Norm{x-y},\] 
which contradicts the above inequality, so either the $\delta$ continuity of $d$ must be violated giving a solution to $X$ of type \ref{m_bad_metametric} or the $\lambda$ continuity of $f$ must be violated giving a solution of type \ref{m_bad_f}.
Thus we have shown that \GCM is in \CLS.
\end{proof}

Now that we have shown that \GCM is total, we note 
that since the solutions of \GCM are a subset of those for \MMCM, we have the following.

\begin{observation}
\label{obs:MMCMtoGCM}
\MMCM can be reduced in polynomial-time to \GCM.
\end{observation}

Thus, by Theorem~\ref{thm:GCMinCLS}, we have that \MMCM is in \CLS.
Next, we show that \MMCM is \CLS-hard by a reduction 
from the canonical \CLS-complete problem \CLO to an instance of \MMCM.
By Observation~\ref{obs:MMCMtoGCM}, we then also have that \GCM is \CLS-hard.

\begin{theorem}
\label{thm:MMCMisCLShard}
  \MMCM is \CLS-hard.
\end{theorem}
\begin{proof}
  Given an instance $X = (f,p,\e,\lambda)$ of \CLO, we construct a meta-metric $d(x,y) = p(x) + p(y) + 1$. 
	Since $p$ is non-negative, $d$ is non-negative, and by construction, $d$ is symmetric and satisfies the triangle inequality. Finally, $d(x,y) > 0$ for all choices of $x$ and $y$ so $d$ is a valid meta-metric (Definition~\ref{def:metametric}) Furthermore, if $p$ is $\lambda$-continuous with respect to the given $p$-norm $\Norm{\cdot}_r$, then $d$ is ($2^{1/r-1}\lambda$)-continuous with respect to $\Norm{\cdot}_r$. For clarity, in the below proof we'll omit the subscript $r$ when writing the norm of an expression. To see this we observe that $x,x',y,y'\in [0,1]^n$, we have $\Norm{p(x)-p(x')}/\Norm{x-x'} \leq \lambda$ and $\Norm{p(y) - p(y')}/\Norm{y-y'} \leq \lambda$, so
  \begin{align*}
    \frac{\Norm{d(x,y) - d(x',y')}}{\Norm{(x,y) - (x',y')}}
    &= \frac{\Norm{p(x) - p(x') + p(y) - p(y') + 1 - 1}}{\Norm{(x,y) -(x',y')}}\ \leq\ \frac{\lambda\Norm{x-x'} + \lambda\Norm{y-y'}}{\Norm{(x,y) -(x',y')}}\\
    &\leq \frac{\lambda\Norm{x-x'} + \lambda\Norm{y-y'}}{2^{1-1/r}(\Norm{x-x'} + \Norm{y-y'})} \leq 2^{1/r-1}\lambda \text{.}
  \end{align*}
We output an instance $Y = (f,d,\e'=\e,c = 1-\e/4,\delta=\lambda, \lambda'=2^{1/r-1}\lambda)$.

Now we consider solutions for the instance $Y$ and show that they correspond to solutions for our input instance $X$.
First, we consider a solution of type \ref{m_fixpoint}, a point $x\in [0,1]^3$ such that $d(f(x),x) \leq \e'=\e$. We have $p(f(x)) + p(x) + 1 \leq \e$, but this can't happen since $\e < 1$ and $p$ is non-negative, so solutions of this type cannot exist.

Now consider a solution that is a pair of points $x,y\in [0,1]^3$ satisfying one of the conditions in \ref{m_violation}. If the solution is of type \ref{m_not_contracting}, we have $d(f(x),f(y)) > c d(x,y)$, and by our choice of $c$ this is exactly
\[\frac{d(f(x),f(y))}{d(x,y)} > (1-\e/4)\] and
\begin{align*}
  p(f(x)) + p(f(y)) + 1 &> (1-\e/4) (p(x) + p(y) + 1)\\
                        &\geq p(x) + p(y) - 3\e/4
\end{align*}
so either $p(f(x)) > p(x) - \e$ or $p(f(y)) > p(y) - \e$, and one of $x$ or $y$ must be a fixpoint solution to our input instance.
Solutions of type \ref{m_bad_metametric} or \ref{m_bad_f} immediately give us violations of the $\lambda$-continuity of $f$, and thus solutions to $X$.

This completes the proof that \MMCM is \CLS-hard.
\end{proof}

So combining these results we have the following.

\begin{theorem}
\label{thm:MMCM-CLScomplete}
\MMCM and \GCM are \CLS-complete.
\end{theorem}

Finally, as mentioned in the introduction, we note the following.
Contemporaneously and independently of our work, Daskalakis, Tzamos, and
Zampetakis~\cite{DTZ17} defined the problem \MBanach, which is like \MMCM except
that it requires a metric, as opposed to a meta-metric.  They show that \MBanach
is \CLS-complete.  Since every metric is a meta-metric, \MBanach can be
trivially reduced in polynomial-time to \MMCM. Thus, their \CLS-hardness result
is stronger than our Theorem~\ref{thm:MMCMisCLShard}.
The
containment of \MBanach in \CLS is implied by the containment of \MMCM in \CLS. 
To prove
that \MMCM is in \CLS, we first reduce to \GCM, which we then show is in \CLS.
Likewise, the proof in~\cite{DTZ17} that \MBanach is in \CLS works even when
violations of the metric properties are not allowed as solutions, so they, like
us, actually show that \GCM is in \CLS.

\end{document}